\documentclass[12pt]{article}
\usepackage{amscd, amsmath}
\usepackage{amsthm, amssymb, mathrsfs}
\usepackage{mathtools}
\usepackage{sectsty}
\usepackage{slashed}
\usepackage{scalerel}
\usepackage[bottom,flushmargin]{footmisc}
\usepackage{mathrsfs}
\usepackage{float}
\usepackage{array}
\usepackage{accents}
\usepackage{empheq}
\usepackage[dvipsnames]{xcolor}
\usepackage[most]{tcolorbox}
\usepackage[british]{babel}
\usepackage{tikz-cd} 
\usepackage{hyperref}
\usepackage{tocloft}
\usepackage{titlesec}
\usepackage{cite} 
\usepackage{lmodern}
\usepackage{setspace}
\usepackage{enumitem}
\usepackage{appendix}
\usepackage{doi}

\titlelabel{\thetitle.\quad}
\titleformat*{\section}{\fontsize{16}{19}\bfseries\selectfont}
\titleformat*{\subsection}{\fontsize{13}{17}\bfseries\selectfont}
\subsubsectionfont{\normalfont\itshape}

\renewcommand{\baselinestretch}{1.2}
\newcommand{\linebr}{\\[0.2em]}

\DeclareMathAlphabet{\mathpzc}{OT1}{pzc}{m}{it}

\definecolor{myblue1}{RGB}{0, 0, 139}

\hypersetup{
    linktocpage,
    colorlinks,
    citecolor=myblue1,
    filecolor=black,
    linkcolor=myblue1,
    urlcolor=myblue1
}

\newtheorem{theorem}{Theorem}[section]
\newtheorem{lemma}[theorem]{Lemma}
\newtheorem{proposition}[theorem]{Proposition}
\newtheorem*{proposition*}{Proposition}
\newtheorem{definition}{Definition}[section]

\newtheorem{corollary}[theorem]{Corollary}

\newtheorem*{assumption(A)}{Condition (A)}

\theoremstyle{remark}
\newtheorem{remark}{{\bf Remark}}[section]

\numberwithin{equation}{section}

\newcommand{\Lie}{\mathcal L}
\newcommand{\VV}{\mathscr V}
\newcommand{\HH}{\mathcal H}

\newcommand{\cL}{\mathscr{L}}
\newcommand{\hcL}{\hat{\mathscr{L}}}
\newcommand{\sL}{{\sf L}}

\newcommand{\hg}{\hat{g}}
\newcommand{\thg}{\thickhat{g}}

\newcommand{\SU}{\mathrm{SU}}
\newcommand{\su}{\mathfrak{su}}

\newcommand{\spin}{\mathfrak{spin}}

\newcommand{\Spin}{\mathrm{Spin}}

\newcommand{\Cl}{\mathrm{Cl}}
\newcommand{\CC}{\mathbb{C}}

\newcommand{\dd}{{\mathrm d}}
\newcommand{\vol}{\mathrm{vol}}

\newcommand{\diag}{\mathrm{diag}}

\newcommand{\divergence}{\mathrm{div}}
\newcommand{\grad}{\mathrm{grad}}

\newcommand{\sD}{\slashed{D}}
\newcommand{\parity}{\mathcal{P}}

\newcommand{\beq}{\begin{equation}}
\newcommand{\eeq}{\end{equation}}
\newcommand{\bal}{\begin{align}}
\newcommand{\eal}{\end{align}}
\newcommand{\bmatr}{\begin{bmatrix}}
\newcommand{\ematr}{\end{bmatrix}}

\newcommand*\conj[1]{\overbar{#1}}
\newcommand{\overbar}[1]{\mkern 1.5mu\overline{\mkern-1.5mu#1\mkern-1.5mu}\mkern 1.5mu}

\newcommand{\blangle}{\boldsymbol{\langle}}
\newcommand{\brangle}{{\boldsymbol{\rangle}}}

\newcommand{\thickhat}[1]{\mathbf{\hat{\text{$#1$}}}}

 \newtcolorbox{empheqboxed}{ 
 opacityback=0,
 enhanced jigsaw,
 width=\textwidth,
 boxrule=.5pt,
 sharpish corners,
 left=0pt,
 right=2pt,
 top=-9pt, % default value 2mm
 bottom=3pt
}

\newcommand{\RR}{{\text{\scalebox{.9}{$R$}}}}
\newcommand{\sparity}{{\text{\scalebox{.9}{${\mathcal{P}}$}}}}
\newcommand{\PPP}{{\text{\scalebox{.8}{$P$}}}}
\newcommand{\MMM}{{\text{\scalebox{.8}{$M$}}}}
\newcommand{\KKK}{{\text{\scalebox{.8}{$K$}}}}

\DeclareSymbolFont{TOneChars}{T1}{\familydefault}{m}{it}
\SetSymbolFont{TOneChars}{bold}{T1}{\familydefault}{bx}{it} % or b
\DeclareMathSymbol{\mathdh}{\mathord}{TOneChars}{"F0}

\makeatletter
\def\blfootnote{\xdef\@thefnmark{}\@footnotetext}
\makeatother

\begin{document}

\begin{titlepage}

\title{{\bf The geometry of CP symmetry \\ in Kaluza-Klein models\vspace{.5cm}}}

\author{Jo\~ao Baptista}
\date{\vspace{-.3cm}August 2026}

\maketitle

\thispagestyle{empty}
\vspace{1cm}
\vskip 10pt
{\centerline{{\large \bf{Abstract}}}}
\noindent
We investigate the dimensional reduction of the free, massless Dirac equation $\sD \Psi = 0$ on a manifold $M_4 \times K$ equipped with a submersion metric. These background metrics generalize the Kaluza ansatz. They encode 4D massive gauge fields and Higgs-like scalars, alongside the usual metric on $M_4$ and massless gauge fields. In this framework, the interactions of 4D gauge fields with spinors are determined by the Kosmann derivatives of the spinors' internal components. For massive gauge fields, however, those derivatives do not commute with the internal Dirac operator. This causes a natural misalignment between the mass eigenspinors and the spinor representation bases. Thus, a complex, infinite-dimensional, CKM-like matrix emerges directly from the Dirac equation on $M_4 \times K$. Examining this general framework, we prove that it nevertheless cannot produce CP violation in 4D for any compact $K$ with constant internal geometry. The same holds for the Weyl equation unless  $\dim K = 2 \pmod{4}$. Using the language of spin geometry, we develop detailed descriptions of parity and conjugation symmetries in KK models. We find that the gauge representations are always self-conjugate when $\dim K \neq 1 \pmod{4}$ (hence anomaly-free), discuss fermion generations, and introduce a new Lie derivative of spinors along non-Killing vector fields induced by actions of compact groups.

\vspace{-1.0cm}
\let\thefootnote\relax\footnote{ 
\noindent
{\small {\sl \bf  Keywords:} Kaluza-Klein theories; Riemannian submersions; Dirac operator; parity transformation; charge conjugation; CP violation. \vspace{0.15cm}
}
}

\addtocounter{footnote}{-1}

\end{titlepage}

\pagenumbering{roman}

\renewcommand{\baselinestretch}{1.1}\normalsize
\tableofcontents
\renewcommand{\baselinestretch}{1.19}\normalsize

\newpage

\pagenumbering{arabic}

\section{Introduction and overview of results}
\label{Introduction}

\subsection{Introduction}

The discovery of CP violation was the observation that, in nature, left-handed particles do not interact with the weak force in exactly the same way as right-handed antiparticles do. The experimental tests of CP violation rely on measuring this asymmetry in a variety of physical processes \cite{CCFT, Fanti, Alavi, Navas}.
Although modest in magnitude, the observed asymmetry can be regarded as natural. Particles and antiparticles respond differently to the strong and electromagnetic forces. Why would they respond identically to the weak force? 

The only reason the observed asymmetry seems unexpected is a technical one. In the Standard Model (SM), one usually goes from particles to antiparticles by complex conjugating gauge representations. Conjugating the electromagnetic $\mathfrak{u} (1)$ representation of a particle, for instance, or its colour $\su (3)$ representation, produces distinct representations associated with the antiparticle. However, in the same SM framework, the weak gauge bosons act on fields either trivially or through the fundamental $\su (2)$ representation. Notably, this representation is equivalent to its complex conjugate. So the theoretical structure of the SM would suggest, at first glance, identical behaviours of particles and antiparticles when interacting with the weak force.

Nature insists on the asymmetry, though, despite the equivalence of $\su (2)$ representations. So physicists had to devise an entirely new mechanism to embed the particle-antiparticle asymmetry into the SM equations. Conjugating fermion representations does not change anything for the weak force. The skilful solution was to introduce ad hoc complex phases in the CKM and PMNS matrices \cite{KM, BLS, BS}. This is a minimal adjustment that works effectively. Those phases are not invariant under complex conjugation. So they produce the desired inequivalence between the equations for particles and antiparticles.

A sense of theoretical unease may linger, though. What is the origin of the new phases? What kind of mathematical structures can produce them and govern their specific values? Are there natural frameworks in which the equations of motion of particles and antiparticles, as derived from first principles, differ by something beyond simple conjugation of gauge representations? Such equations would break CP.

The purpose of this paper is to investigate whether a generalized version of the Kaluza-Klein (KK) framework has that property. We consider the massless Dirac equation $\sD \Psi = 0$, together with its Weyl counterpart, on a manifold $P = M_4\times K$ with submersion background metrics. These background metrics generalize the Kaluza ansatz by encoding 4D massive gauge fields and Higgs-like scalars, in addition to the metric on $M_4$ and massless gauge fields. In this setting, the interactions of 4D gauge fields with spinors are determined by the Kosmann-Lichnerowicz derivatives of the spinors' internal components. However, for massive gauge fields those derivatives do not commute with the internal Dirac operator. This creates a misalignment between the mass eigenspinors and the representation bases. So a complex, infinite-dimensional, CKM-like matrix emerges directly from the Dirac equation on $M_4 \times K$. A priori, this opens new scope for CP violation in these models. The present paper investigates this possibility and aims to settle the question within a rigorous geometric framework.

Early works studying the discrete symmetries of the Dirac equation in pure higher-dimensional gravity are \cite{Thirring, GN, Wetterich84}, for example. They establish most of the properties of those spinor symmetries for different numbers of extra dimensions. Following the traditional form of the Kaluza ansatz, they focus on higher-dimensional metrics encoding only abelian or massless 4D gauge fields. This hides the rich spinorial structure that emerges when the background also encodes 4D massive gauge fields. That new structure is our candidate to produce CP violation here. More recent studies of CP violation induced by extra dimensions are \cite{CM, BGR, HLWY, GW, LMN, FLM}, for example. They investigate models with orbifold compactifications and higher-dimensional gauge fields, with a phenomenological and brane world bent. More references in the review \cite{IN}. They differ significantly from the pure-gravity KK approach of the present paper.

Kaluza-Klein models with massive gauge fields based on general submersion metrics are investigated in \cite{Bap1, Bap2, Bap3}. The first paper establishes that submersions with fibres that are not totally geodesic correspond to gauge theories with massive 4D gauge fields and Higgs-like scalars. It computes the mass of gauge fields from the second fundamental form of the fibres of the submersion. The paper \cite{Bap2} studies geodesic motion on these higher-dimensional backgrounds, and how it is physically perceived after projection to $M_4$. The paper \cite{Bap3} studies the free Dirac equation $\sD \Psi = 0$ and its KK reduction to 4D, like the present one. It establishes that, on the same submersion backgrounds, the dimensionally reduced Weyl equation in general exhibits chiral couplings of 4D fermions and massive gauge fields. 
This is a new way to circumvent well-known no-go arguments against the existence of chiral fermions in KK models \cite{CS, Wetterich83a, Wetterich83, Witten83}.

Extended reviews of the Kaluza-Klein framework can be found in \cite{Bailin, Bou, CFD, CJBook, Duff, WessonOverduin, Witten81}. Some of the early original references are \cite{Kaluza, Klein, EB, Jordan, Thiry, DeWitt, Kerner, Cho}. This paper follows the treatment given in \cite{Bap1, Bap2, Bap3} of massive gauge fields, Higgs-like scalars and spinors. Clear limitations of the present analysis are that it is classical, all throughout, and that it investigates mainly the mathematical properties of the theories. So no phenomenology at this stage. It also does not discuss the dynamical origin of the vacuum metric on $K$. It assumes the existence and stability of some $g_K$. Then it relates the properties of 4D fermions and gauge fields to the properties of that metric. We now give an overview of the main results in the paper.

\subsection{Overview of the main results}

\subsubsection*{\bf Spinors on Riemannian submersions}

Section \ref{SpinorsSubmersions} recalls the main properties of spinors on a manifold $P= M \times K$ equipped with a submersion metric $g_P$. The classical results about Riemannian submersions were developed by O'Neill in \cite{ONeill1}, after foundational work in \cite{Ehresmann, Hermann}. They are presented in \cite{Besse, FIP}, for example. We use the translation of those results to the Kaluza-Klein language provided in \cite{Bap1}. In that language, submersion metrics on $P$ are equivalent to triples $(g_M, A, g_K)$ encoding the metric on $M$, 4D gauge fields (massless and massive), and an internal metric $g_K$ that can change along $M$. 
The higher-dimensional Einstein-Hilbert action can be decomposed as 
 \beq \label{GaugedSigmaModelAction0}
\int_P \, R_{g_\PPP}  \vol_{g_\PPP}   \ = \ \int_P \, \Big[\, R_{g_\MMM}  \, + \, R_{g_\KKK} \, - \, \frac{1}{4}\, |F_A|^2 \, - \,  \frac{1}{4}\,  |\dd^A  g_K|^2  \, + \,  |\dd^A \, (\vol_{g_\KKK})|^2 \, \Big] \,\vol_{g_\PPP} \, . 
\eeq
 This extends the usual Kaluza-Klein result to the setting of general Riemannian submersions. It suggests that the internal metric $g_K$ can be understood as a geometric version of the traditional Higgs fields. More details can be found in \cite{Bap1, Besse}. 

For a general submersion metric on $P$, the gauge one-form $A_\mu$ on $M$ has values in the full space of vector fields on $K$, which is the Lie algebra of the diffeomorphism group ${\rm Diff} (K)$. So it can be expanded as $A^a_\mu\, e_a$, where $\{ e_a \}$ is a set of independent vector fields on $K$, some of them Killing and most others non-Killing. The gauge group is ${\rm Diff} (K)$ or a subgroup. It need not act on $K$ only through isometries of $g_K$.
 The classical mass of the 4D gauge field linked to a divergence-free $e_a$ is calculated in \cite{Bap1} to be
\beq   \label{MassFormula}
\left(\text{Mass} \ A_\mu^a \right)^2 \ \ \propto \ \ \frac{ \int_K  \; \left\langle \Lie_{e_a}\, g_K,  \; \Lie_{e_a}\, g_K \right\rangle  \, \vol_{g_\KKK} }{ 2 \int_K  \,  g_K (e_a ,  e_a  ) \ \vol_{g_\KKK} }  \ ,
\eeq
where $\Lie_{e_a} g_K$ denotes the Lie derivative of the internal vacuum metric along $e_a$. This suggests that massive gauge fields, however light, should not be linked to exact isometries of $g_K$. The derivatives $\Lie_{e_a} g_K$ can be small yet non-zero.

Spinors on general Riemannian submersions were studied in \cite{Moroianu, LS, Reynolds, Bap3}. Section \ref{SpinorsSubmersions} summarizes the results that are most relevant for us in the language of \cite{Bap3}, suitable for explicit Kaluza-Klein physics. Assuming that $M$ is even-dimensional, it describes the decomposition of higher-dimensional spinors as tensor products of horizontal and vertical spinors. Using the equivalence $g_P \simeq (g_M, A, g_K)$, it presents the explicit formula for the Dirac operator on $P$ written in terms of the components of the triple.

 \subsubsection*{\bf Reflections and parity transformations on $M \times K$}
  
 Reflections and parity inversions on $M$ have natural extensions to diffeomorphisms of $P = M \times K$ that do not change the internal coordinates. The extended diffeomorphisms are denoted $R$ and $\parity$, respectively. They are not isometries of the submersion metric on $P$, in general. The transformation rules of $g_P$ under pullback by those diffeomorphisms, for example $g_P \rightarrow \parity^\ast g_P$, encapsulate the usual transformation rules of 4D gauge fields $A^a_\mu$ and Higgs-like fields under reflections and parity inversions.
  
 The diffeomorphisms of $P$ induced by reflections and parity inversions also determine transformations of the higher-dimensional spinors. However, the construction of spinor bundles depends on the background metric, and $g_P$ is not preserved by the diffeomorphisms. Thus, reflections and parity inversions relate spinors defined on different bundles over $P$. For example, the spinor bundles $S_{g_\PPP}$ and $S_{\parity^\ast g_\PPP}$. This is a slight variation from the usual story on Minkowski space, which deals with isometries of $g_M$ and automorphisms of a single bundle \cite{Hamilton, Stone, Srednicki}. Despite this difference, the KK extensions of reflections and parity inversions still induce symmetries of the higher-dimensional Dirac equation, $\sD^P \Psi = 0$. This is the main point for us. Section \ref{ReflectionsParity} and appendix \ref{ProofsReflectionsParity} spell out the main properties of the spinor transformations on $M \times K$ induced by those symmetries. The tone is more geometric than in earlier work \cite{GN, Wetterich84}. We extend those results to our KK models, with general submersion metrics and spinor transformations between different bundles.

\subsubsection*{\bf Conjugations of spinors $j_\pm$}
 
 Section \ref{ConjugationSymmetries} and appendix \ref{ConventionsSpinors} give succinct accounts of the conjugation automorphisms $j_\pm$ of spinor bundles on manifolds of arbitrary signature $(s,t)$. They describe the essential properties of these symmetries of the massless Dirac equation.  
 Section \ref{ConjugationSymmetries} uses a geometric approach and hides the basis-dependent constructions, with gamma matrices in special representations, that are generally used to prove the existence of conjugation maps. The bundle-version of the maps $j_\pm$ includes their relations with differential operators such as covariant derivatives, the Dirac operator and, importantly for us, the Kosmann-Lichnerowicz derivative of spinors. The spinor conjugation maps in arbitrary signatures are a well-studied topic in the literature. See for instance \cite{Figueroa, Stone, Hitoshi}. The purpose of section \ref{ConjugationSymmetries} is to provide a short and hopefully clean summary that covers the general case of signature $(s,t)$ and the bundle-related properties. At the end, it also presents results that help to understand conjugation maps in the specific setting of Riemannian submersions.

 \subsubsection*{\bf Actions of compact groups on $K$ and its spinors}
 
 An isometric action of a connected group on a compact spin manifold $K$ has a canonical lift to an action on the spinors in $S_{g_\KKK}$, after passing to a covering group if necessary. In its infinitesimal version, this lift defines a Lie derivative of spinors along Killing vector fields \cite{Lich}. The standard extension of this derivative to arbitrary vector fields on $K$ is the Kosmann-Lichnerowicz derivative \cite{Kosmann}, defined by
 \beq
 \label{KLDerivative}
 \cL_V \psi \ := \ \nabla_V \psi \; - \; \frac{1}{8} \, \sum_{j,\, k}\, \big[ \, g(\nabla_{v_j} V, v_k) \: -  \: g(\nabla_{v_k} V, v_j)  \, \big] \, v_j \cdot v_k \cdot \psi  \ .
 \eeq
 This extended derivative couples the 4D gauge fields to internal spinors in general Kaluza-Klein models, as described in \cite{Bap3}. A drawback of this derivative, however, is that it satisfies the closure relation ($\cL_U \cL_V - \cL_V \cL_U) \psi = \cL_{[U, V]}\psi$ only when $U$ or $V$ are conformal Killing for $g_K$. Otherwise, it is not a proper Lie derivative.
 
 In section \ref{ActionsCompactGroups} we show that, when a connected compact group $G$ acts on $K$, there is a useful alternative derivative of spinors given by
 \beq
 \sL_V \psi  \ = \ \cL_V \psi \; - \; \tau_{V} \psi \ = \ \cL_V \psi \; + \;   \frac{1}{4} \; \sum_{j\neq k} \; g \big( \alpha^{-1} (\Lie_V \alpha) (v_j),\, v_k \big) \; v_j \cdot v_k \cdot \psi \ .
 \eeq
 It coincides with the Kosmann-Lichnerowicz derivative when $V$ is conformal Killing. It satisfies ($\sL_U \sL_V - \sL_V \sL_U) \psi = \sL_{[U, V]}\psi$ for all fundamental vector fields on $K$ induced by the $G$-action, even when $U$ and $V$ are not conformal Killing with respect to $g_K$. So we obtain a new lift of non-isometric actions to spinors, at least when $K$ is compact.
 
 The new derivative is defined through a canonical map $\alpha: S_{g_\KKK} \rightarrow S_{\hg_\KKK}$ between two spinor bundles on $K$, corresponding to the metric $g_K$ and its $G$-averaged metric $\hg_K$. The relations between spinors associated with different Riemannian metrics on the same manifold have been previously studied in \cite{BG, Wang, AHermann}, for example. The parts of those works that we need here are summarized in appendix \ref{SpinorsDifferentMetrics}. That appendix also presents several new formulas that are necessary for our purposes, mostly related to the transport of the Kosmann-Lichnerowicz derivative through $\alpha$.

\subsubsection*{\bf Spinor representation spaces, gauge anomalies and fermion generations}

Let $\mathfrak{g}$ denote the Lie algebra of vector fields on $K$ induced by the action of a compact group $G$. The new spinor derivative $\sL_V$ allows us to define a unitary representation of $\mathfrak{g}$ on the space of spinors in $S_{g_\KKK}$ equipped with its natural $L^2$ Hermitian product:
\beq
\rho_V (\psi) \ := \ \sL_V \psi \; + \; \frac{1}{2} \, \divergence_g(V) \, \psi \ .
\eeq
In section \ref{RepresentationSpaces}, we describe two decompositions of the space of spinors on $K$ as a sum of finite-dimensional subspaces that are irreducible under $\rho$.
When the $G$-action is not isometric, the operators $\rho_V$ do not commute with the Dirac operator on $S_{g_\KKK}$. So the irreducible $\mathfrak{g}$-spaces are non-trivially related to the  $\sD$-eigenspaces. 

Following \cite{Bap3}, we remark that, for non-isometric $G$-actions on an even-dimensional $K$, our KK models can have chiral fermions. In this case, the chiral interactions generated by the representation $\rho$ are free of local gauge anomalies. This happens because $\rho$ commutes with the internal spinor conjugation maps $j_\pm$. So it is a self-conjugate representation. The irreducible $\mathfrak{g}$-spaces of complex type will always appear in conjugate pairs.

Section \ref{RepresentationSpaces} also discusses geometric mechanisms that can potentially produce distinct generations of 4D fermions in KK models. For example, a perturbation of the internal metric $g_K$ that breaks the isometry group $G \rightarrow G'$ leads to a splitting of the degenerate eigenvalues of $\sD_{g_\KKK}$ and respective eigenspaces, $E_\mu \rightarrow E_{\mu_1} \oplus \cdots \oplus E_{\mu_r}$. The $\mathfrak{g}$-representation on $E_\mu$ branches to $\mathfrak{g}'$-representations on the summands. If different summands contain the same $\mathfrak{g}'$-irreducible, this will correspond, in the KK model, to 4D fermions in the same $\mathfrak{g}'$-representation but with slightly different masses.

\subsubsection*{\bf The CP-transformed Dirac equations in 4D}

In section \ref{DimensionalReductionCPViolation} we simplify the treatment and assume that the metric $g_K$ on the internal space is constant along $M$. Then any higher-dimensional spinor can be written as $\Psi (x,y)  =  \sum_\alpha \, \varphi_\alpha^\HH(x) \otimes \psi^\alpha(y)$, where $x$ and $y$ are coordinates on $M$ and $K$, respectively.  The $\varphi_\alpha$ are spinors on $M$ and the set $\{ \psi^\alpha \}$ is a $L^2$-orthonormal basis of the space of internal spinors on $K$. For a higher-dimensional spinor of this form, the Dirac equation $\sD^P \Psi = 0$ is equivalent to an infinite set of equations on $M$:
\begin{multline}
\label{DemonstrativeDiracEquation}
 i\, \gamma^\mu \,  \big\{  \nabla^{M}_{X_\mu} \,  \varphi_\alpha \, + \,  A^a_\mu \,  \,  \blangle \, \psi_\alpha  \, , ( \rho_{e_a} +  \tau_{e_a} ) \psi^\beta \, \brangle_{L^2}\,  \, \varphi_\beta \big\}  \,  
+ \,  \blangle \, \psi_\alpha  \, , \,  \sD^K  \psi^\beta \, \brangle_{L^2}\, \, \varphi_\beta \; \linebr +\; \frac{1}{8}\, \,  (F_A^a)_{\mu \nu } \, \blangle \psi_\alpha \, , e_a \cdot \psi^\beta \brangle_{L^2}\, \,  \gamma^\mu  \gamma^\nu \, \varphi_\beta  \ = \ 0 \ .
\end{multline}
A basis of internal spinors $\{ \psi^\alpha \}$ formed with eigenspinors of the internal Dirac operator $\sD^K$, with respective eigenvalues $m_\alpha$, is called a mass basis. In such bases, equation \eqref{DemonstrativeDiracEquation} reduces to a standard 4D gauged Dirac equation with a mass term and a Pauli term.

Conjugations of spinors and parity inversions are exact symmetries of the massless Dirac equation on $M \times K$. So are the compositions $j_\pm \parity$. This is described in sections \ref{ReflectionsParity} and \ref{ConjugationSymmetries}. If a higher-dimensional spinor satisfies $\sD \Psi = 0$, we always have $\sD (j_\pm \parity \, \Psi) = 0$ as well.  In section \ref{DimensionalReductionCPViolation} we determine how the second equation looks after dimensional reduction to $M$. In the case of the conjugation $j_-$, the result is the set of equations:
\begin{multline}
\label{DemonstrativeCPDiracEquation}
 i \, \gamma^\mu \,  \big\{  \nabla^{M}_{X_\mu} \,  \varphi^{cp}_\alpha \, + \,  (A^p)^a_\mu \,  \, \conj{ \blangle \, \psi_\alpha  \, , ( \rho_{e_a} +  \tau_{e_a} ) \psi^\beta \, \brangle}_{L^2} \,  \,  \varphi^{cp}_\beta  \big\}  \,  
+ \,  \conj{\blangle \, \psi_\alpha  \, , \,  \sD^K  \psi^\beta \, \brangle}_{L^2}\, \,  \varphi^{cp}_\beta  \; \linebr +\; \frac{1}{8}\, \,  (F_{A^p}^a)_{\mu \nu } \, \conj{\blangle \psi_\alpha \, , e_a \cdot \psi^\beta \brangle}_{L^2}\, \,  \gamma^\mu  \gamma^\nu \,  \varphi^{cp}_\beta  \ = \ 0
\end{multline}
where $\varphi^{cp}_\alpha \ =  \ j_-^M \parity_M \, (\varphi_\alpha)$ are the CP-transformed spinors in 4D. Thus, the spinors $\varphi^{cp}_\alpha$ satisfy equations of motion very similar to those satisfied by the original $\varphi_\alpha$. The differences are that the gauge form is now the parity-transformed $A^p$, instead of $A$, and four types of matrices on the space of internal spinors ---  determined by the operators $\rho_{e_a}$, $\tau_{e_a}$, $\sD^K$, and $e_a \cdot$  --- appear complex-conjugated in the new equation.

Section \ref{DimensionalReductionCPViolation} also describes how CP transformations commute with dimensional reduction, in an appropriate sense.

\subsubsection*{\bf Paths for CP violation}

The 4D Dirac equations \eqref{DemonstrativeDiracEquation} and \eqref{DemonstrativeCPDiracEquation} look genuinely different. There is no apparent reason for the existence of a representation basis $\{ \psi^\alpha \}$ that renders the CKM-like mass matrix $\blangle \, \psi_\alpha  \, , \,  \sD^K  \psi^\beta \, \brangle_{L^2}$, the Pauli term matrix $\blangle \psi_\alpha \, , \, e_a \cdot \psi^\beta \brangle_{L^2}$, and the new non-minimal gauge coupling matrix $\blangle \, \psi_\alpha  \, , \, \tau_{e_a} \, \psi^\beta \, \brangle_{L^2}$, all simultaneously real on general grounds. 

Nevertheless, in section \ref{RestoringCPSymmetry}, we show that the two sets of equations are physically equivalent. More precisely, there exists a unitary transformation ${\mathrm U}$ on the infinite-dimensional space of spinors that transforms \eqref{DemonstrativeDiracEquation} and \eqref{DemonstrativeCPDiracEquation} into each other. This remains true in the presence of massive gauge fields, when $ \sD^K$ does not commute with some gauge representations and the non-minimal gauge coupling is present. The unitary ${\mathrm U}$ is constructed from the internal conjugation maps $j_\pm^K$ and the chosen basis $\{\psi^\alpha \}$ of spinors on $K$.

When $\dim K \neq 2 \pmod{4}$, the higher-dimensional CP transformations $j_\pm \parity$ preserve the Weyl condition $\Gamma_P \Psi = \Psi$ for spinors on $M \times K$. Thus, the same arguments show that the dimensionally reduced Weyl equation is CP-invariant as well, even in the presence of massive gauge fields. However, when $\dim K = 2 \pmod{4}$, the Weyl equation is CP-invariant if and only if it is P-invariant. So CP violation is equivalent to intrinsic chirality. An explicit example described in \cite{Bap3}, with internal space $K=T^2$, then indicates that CP violation should indeed emerge in these pure-gravity KK models with submersion backgrounds. 
Away from $\dim K = 2 \pmod{4}$, the emergence of CP violation requires a starting point more elaborate than the plain Dirac equation in higher dimensions.

\section{Spinors on Riemannian submersions}
\label{SpinorsSubmersions}

\subsection{Riemannian submersions}
\label{Submersions}

This section recalls relevant properties of spinors on manifolds equipped with submersion metrics. These metrics generalize the Kaluza ansatz by encoding not only the 4D metric and massless gauge fields, but also 4D massive gauge fields and an internal metric that can vary along $M$. The main classical results about Riemannian submersions were developed in \cite{ONeill1, Ehresmann, Hermann} and are presented in \cite{Besse, FIP}, for example. We use the translation of those results to the Kaluza-Klein language provided in \cite{Bap1}.

Let $g_P$ be a Lorentzian metric on the higher-dimensional space $P =  M \times K$ such that the projection $P \rightarrow M$ is a Riemannian submersion. As described in \cite{Bap1}, this is equivalent to taking a $g_P$ determined by three simpler objects: 
\begin{itemize}
\item[{\bf i)}]  a Lorentzian metric $g_M$ on the base $M$; 
\vspace{-.1cm}
\item[{\bf ii)}]  a family of Riemannian metrics $g_K(x)$ on the fibres $K_x$ parameterized by the points $ x \in M$;
\vspace{-.2cm}
\item[{\bf iii)}] a gauge one-form $A$ on $M$ with values in the Lie algebra of vector fields on $K$.
\end{itemize}
These objects determine the higher-dimensional metric through the relations
\bal \label{MetricDecomposition}
g_P (U, V) \ &= \ g_K (U, V) \nonumber \\
g_P (X, V) \ &= \  - \ g_K \left(A (X), V \right) \nonumber \\
g_P (X, Y) \ &= \ g_M (X, Y) \ + \  g_K \left(A(X) , A(Y) \right) \ ,
\end{align}
valid for all tangent vectors $X,Y \in TM$ and vertical vectors $U, V \in TK$. These relations generalize the usual Kaluza ansatz for $g_P$. Choosing a set $\{ e_a \}$ of independent vector fields on $K$, the one-form on $M$ can be decomposed as a sum 
\beq \label{GaugeFieldExpansion}
A(X) \ = \ \sum\nolimits_a \,A^a(X) \, e_a \ .
\eeq
The real-valued coefficients $A^a(X)$ are the traditional gauge fields on $M$. For general submersion metrics on $P$, this can be an infinite sum, with $\{ e_a \}$ being a basis of the full space of vector fields on $K$, which is the Lie algebra of the diffeomorphism group ${\rm Diff} (K)$. The gauge group need not act on $K$ only through isometries of $g_K$.

The curvature $F_A$ is a two-form on $M$ with values in the Lie algebra of vector fields on $K$. It can be defined by 
\beq
\label{CurvatureDefinition}
F_{A} (X, Y)  \ := \ (\dd_M A^a) (X, Y) \,\, e_a  \ + \ A^a (X)\, A^b (Y) \, [e_a, e_b]  \ ,
\eeq
where the last term is just the Lie bracket $ [A(X), A(Y)] $ of vector fields on $K$.

The tangent bundle of $P$ has two natural decompositions: 
\beq \label{HorizontalDistribution}
T P \; =\;  TM \oplus TK \; =\; \HH \oplus \VV   \ .
\eeq
Here $\VV := TK$ is the vertical subbundle while its orthogonal complement, $\HH := (TK)^\perp$, is the horizontal subbundle.
So a tangent vector $w\in TP$ can be decomposed in two different ways, $w= w_M + w_K = w^\HH  +  w^\VV$. The relation between them is 
\beq \label{DefinitionHorizontalDistribution}
w^\VV  \ = \   w_K \; - \;  A (w_M)    \qquad \qquad  w^\HH  \ = \  w_M \; + \;   A (w_M)   \ .
\eeq
The information contained in the gauge one-form $A$ on $M$ is equivalent to the information contained in the horizontal distribution $\HH \subset TP$. 

One can construct local, $g_P$-orthonormal trivializations of $TP$ using only horizontal and vertical vectors. They can take the form $\{ X_\mu^\HH, v_j \}$. Here the $v_j$ form an orthonormal basis of $TK$ with respect to $g_K (x)$, for each $x \in M$. The $X_\mu$ form a $g_M$-orthonormal basis of $TM$. The horizontal lift of $X_\mu$ to $P$ is denoted $X_\mu^\HH$. It is given by
\beq
\label{BasicLiftX}
X_\mu^\HH =  X_\mu \; + \; A^a(X_\mu) \; e_a  \ .
\eeq
Such horizontal lifts are called basic vector fields on $P$ \cite{Besse, FIP}.

\subsection{Spinors on $M_4 \times K$}
\label{SpinorsOnP}

\subsubsection{Horizontal and vertical spinors}
\label{HorizontalVerticalSpinors}

The conventions for spin geometry used in this paper are described in appendix \ref{ConventionsSpinors}. In this section, we recall the specific features pertaining to spinors on Riemannian submersions. These were investigated in \cite{Moroianu, LS, Reynolds} and later in \cite{Bap3}, for example. Here we use the notation of \cite{Bap3}, which is adapted to the study of Kaluza-Klein physics.

Locally, spinors on $M \times K$ have values on the higher-dimensional spinor space $\Delta_{m+k}$. This space can be written as the tensor product $\Delta_{m} \otimes \Delta_k$, where $m$ and $k$ are the dimensions of $M$ and $K$. All three spaces have irreducible representations of Clifford algebras. There is a standard isomorphism between $\Cl(m+k-1,1)$ and the $\mathbb{Z}_2$-graded tensor product of algebras $\Cl(m-1,1) \, \hat{\otimes}  \! \Cl(k)$. It is determined by the correspondence of generators
\begin{align}
\label{CliffordMultiplication}
1 \;&= \;1\otimes 1   \nonumber   \linebr
\Gamma_\mu \; &= \; \gamma_\mu \otimes 1 \quad \ \ {\rm for} \ \ \mu = 0, \ldots, m-1       \nonumber  \linebr
\Gamma_{m -1+j} \;  &= \; \Gamma_M \otimes \tilde{\gamma}_j \quad \ {\rm for} \ \ j = 1, ..., k \ .
\end{align}
This is a recipe to construct higher-dimensional gamma matrices $\Gamma_l$ from the lower-dimensional ones. The $\tilde{\gamma}_j$ are Euclidean gamma matrices acting on $\Delta_k$ and the $ \gamma_\mu$ are gamma matrices for the Minkowski metric on $M$. 
The complex chirality operators are 
\begin{align}
\label{ChiralOperators}
\Gamma_M \; &:= \; i^{\lfloor\frac{m-1}{2} \rfloor}\, \gamma_0 \, \cdots \, \gamma_{m-1}  \qquad \qquad  \qquad  \qquad  \quad   \Gamma_K\; := \; i^{\lfloor\frac{k+1}{2}\rfloor}\, \tilde{\gamma}_1 \cdots \tilde{\gamma}_k    \nonumber \linebr
\Gamma_P \; &:= \;  i^{\lfloor\frac{m+ k -1}{2} \rfloor}\, \, \Gamma_0 \, \Gamma_1 \cdots \Gamma_{m+k-1}  \; = \;   (\Gamma_M)^{k+1}  \otimes \Gamma_K  \ ,
\end{align}
where $\lfloor s \rfloor$ denotes the integral part of $s$. They are normalized so that their square is the identity operator on the respective spinor spaces $\Delta_{m}$,  $\Delta_{k}$ and $\Delta_{m+k}$.

Assume that $TK$ has a topological spin structure in the sense of \cite{Bourguignon}, so a double cover of its oriented frame bundle. Together with the trivial spin structure on the contractible $M$, it determines topological spin structures on $TP$ and on the horizontal and vertical bundles in decomposition \eqref{HorizontalDistribution}. The metric $g_P$ then determines subordinate spin structures in the usual sense of \cite{LM}, so double covers of the oriented, orthonormal frame bundles of $TP$, $\VV$ and $\HH$. We fix those structures for the rest of the paper. 

Given an oriented, real vector bundle $E$ with a metric and spin structure, its complex spinor bundle is denoted $S(E)$. For a tangent bundle, the notation is simplified as in $S(TP) = S(P) = S_{g_\PPP}$.  Sections of $S(\HH)$ are called horizontal spinors over $P$, while sections of $S(\VV)$ are the vertical spinors. Calling $\pi$ the projection from $P$ to $M$, there are natural isomorphisms
\beq 
\label{IsomorphismHorizontalBundle}
S(\HH) \; \simeq \;  S (\pi^\ast(TM)) \; \simeq \; \pi^\ast [S (M)] \ .
\eeq
In particular, a spinor $\varphi$ on $M$ has a unique lift as a horizontal spinor on $P$. It coincides with the pullback $\pi^\ast \varphi$ under this isomorphism. It is denoted $\varphi^\HH$ and is called the basic lift of $\varphi$ to $P$. At the same time, if we fix a point $x \in M$ and consider the fibre $K_x := \{ x\} \times K$ inside $P$ with its metric $g_K (x)$, there is a natural isomorphism between the restriction of $S(\VV)$ to $K_x$ and the spinor bundle of the fibre,
\beq
\label{IsomorphismVerticalBundle}
S(\VV) \, |_{K_x} \; \simeq \; S(K_x)  \ .
\eeq
Overall, since $M$ is even-dimensional, there is a natural isomorphism of spinor bundles
\beq
\label{TensorProductSpinorBundle}
S(P) \ \simeq \ S(\HH) \otimes S(\VV)  \ 
\eeq
that is compatible with the Clifford multiplication implied by \eqref{CliffordMultiplication}, in the sense that 
\begin{align}
\label{EquivarianceClifford2}
U \cdot (\varphi ^\HH \otimes \psi) \ &= \ (\Gamma_M \, \varphi) ^\HH \otimes (U \cdot \psi)  \linebr
X^\HH \cdot (\varphi ^\HH \otimes \psi) \ &= \ (X \cdot \varphi) ^\HH \otimes \psi \ . \nonumber
\end{align}
Here $U$ is any vertical vector field on $P$. It is regarded as such on the left-hand side and as a section of $\VV$ on the right-hand side. As before, $X$ is any vector in $TM$ and $X^\HH$ denotes its basic lift to $P$, as in \eqref{BasicLiftX}.

Since $M$ is contractible, its spinor bundle $S (M)$ is trivial. Due to \eqref{IsomorphismHorizontalBundle}, so is $S(\HH)$ as a bundle over $P$. This implies that a spinor $\Psi$ on $P$ can always be written as a sum 
 \beq
 \label{TensorHDSpinor}
\Psi (x,y) \; = \; \sum\nolimits_{b=1}^{2^{\lfloor \frac{m}{2} \rfloor}} \, \varphi_b^\HH(x) \otimes \psi^b(x,y) \ ,
\eeq
where $m$ is the dimension of $M$, the $\varphi_b$ are Dirac spinors on $M$ and the $\psi^b$ are vertical spinors on $P$. Here $x$ and $y$ denote coordinates on $M$ and $K$, respectively.
When $K$ is compact, the vertical spinors over a fibre $\{x\} \times K$ can always be written as a (possibly infinite) sum of eigenspinors of the internal Dirac operator $\sD^K$. Since the metric $g_K$ depends on $x$, the operator $\sD^K$ and its eigenspinors will also change along $M$, in general. 

Now suppose that the internal metric $g_K (x)$ is independent of $x$. Then an $L^2$-orthonormal basis of eigenspinors $\{ \psi^\alpha (y)\}$ on $K$ can be chosen uniformly over  $M$. So we can take each $\psi^b$ in \eqref{TensorHDSpinor} and expand its $y$-dependence as $\psi^b(x, y) = \sum_\alpha c^b_{\alpha}(x) \, \psi^\alpha (y)$. Inserting this into \eqref{TensorHDSpinor}, it is clear that the higher-dimensional spinor can then be written as a (possibly infinite) sum
\beq
\label{EigenMassDecompositionSpinors}
\Psi (x,y) \ = \ \sum\nolimits_{\alpha} \, \varphi_\alpha^\HH(x) \otimes \psi^\alpha(y) \ .
\eeq
Here the $\varphi_\alpha$ are Dirac spinors on $M$, the $ \varphi_\alpha^\HH$ are their horizontal lifts to $P$, and the $\psi^\alpha$ are eigenspinors of $\sD^K$ independent of the point on $M$.

\subsection{Decomposing the higher-dimensional Dirac operator}

The decomposition of spinors on $P$ into a tensor product of horizontal and vertical parts, as in \eqref{TensorProductSpinorBundle}, leads to a decomposition of the Levi-Civita connection and the Dirac operator on $P$. See \cite{Bap3, Reynolds}. To describe it, take an oriented, orthonormal trivialization of $TP$ adapted to the submersion metric $g_P \simeq (g_M, A, g_K)$. As in section \ref{Submersions}, this means a trivialization $\{ X_\mu^\HH, v_j \}$ formed by basic and vertical vector fields on $P$. Then the following formula was obtained in \cite{Bap3}, developing previous work in \cite{Moroianu, Reynolds, LS}.
\begin{proposition}
\label{thm:DecompositionDiracOperator}
Consider a spinor on $P$ of the form $\Psi = \varphi^\HH(x) \otimes \psi(x,y)$, as in \eqref{TensorHDSpinor}. The action of the higher-dimensional Dirac operator on $\Psi$ can be locally decomposed as 
\begin{align}
\label{GeneralFormulaDecompositionDiracOperator}
\sD^P \Psi \; =& \; \, g_M^{\mu \nu}\,  (X_\mu \cdot  \nabla^M_{\nu} \varphi)^\HH \otimes \psi  \; + \;   g_M^{\mu \nu}\, \, A_\nu^a\,  \, (X_\mu \cdot  \varphi)^\HH  \otimes  \big[ \cL_{e_a} + \frac{1}{2} \, \divergence (e_a) \big] \psi  \nonumber   \linebr
&+  \; ( \Gamma_M \, \varphi)^\HH  \otimes  \sD^K\psi   \; + \;  \frac{1}{8} \, \,  (F_A^a)^{\mu \nu } \,  (  X_\mu \cdot X_\nu \cdot \Gamma_M \, \varphi  )^\HH \otimes (e_a \cdot \psi)   \linebr
 &+\;  g_M^{\mu \nu}\,\, (X_\mu \cdot \varphi)^\HH \otimes \, \Big[  \partial_{X_\nu}  + \frac{1}{2}  \,\big( \partial_{X_\nu}   \log \sqrt{|g_K |}  \, \big)  \Big]  \, \psi  \  \nonumber \linebr
  &+\;  g_M^{\mu \nu}\,  (X_\mu \cdot \varphi)^\HH \otimes \Big( \frac{1}{8} \, \sum\nolimits_{ij}   \big\{ g_P([X_\nu, v_i], v_j)  -  g_P([X_\nu, v_j ] , v_i) \big\}  \, v_i \cdot v_j \cdot \psi  \Big) \, . \nonumber 
 \end{align}
 Here $\cL_{e_a}$ denotes the derivative \eqref{KLDerivative} of spinors on $K$; $\divergence (e_a)$ denotes the divergence of the internal vector field $e_a$ with respect to $g_K$; and $|g_K |$ is the modulus of the determinant of the matrix representing $g_K$ in a fixed coordinate system on $K$. 
\end{proposition}

Formula \eqref{GeneralFormulaDecompositionDiracOperator} is simpler in regions where the Higgs-like scalars are constant, i.e. where the internal metric $g_K$ does not change along $M$. In this case $[X_\mu, v_i] = 0$ and every higher-dimensional spinor can be expressed as a sum of simpler tensor products of the form $\varphi(x)^\HH \otimes \psi(y)$, as in \eqref{EigenMassDecompositionSpinors}. In this simpler setting, we have:
\begin{corollary}
In regions where $g_K$ is constant along $M$, the action of the Dirac operator on a spinor of the form $\varphi(x)^\HH \otimes \psi(y)$ can be decomposed as 
\begin{align}
\label{SimplerDecompositionDiracOperator}
\sD^P (\varphi^\HH \otimes \psi ) \; =& \; \, g_M^{\mu \nu}\,  (X_\mu \cdot  \nabla^M_{\nu} \varphi)^\HH \otimes \psi  \; + \;   g_M^{\mu \nu}\, \, A_\nu^a\,  \, (X_\mu \cdot  \varphi)^\HH  \otimes  \big[ \cL_{e_a} + \frac{1}{2} \, \divergence (e_a) \big] \psi  \nonumber   \linebr
&+  \; ( \Gamma_M \, \varphi)^\HH  \otimes  \sD^K\psi   \; + \;  \frac{1}{8} \, \,  (F_A^a)^{\mu \nu } \,  (  X_\mu \cdot X_\nu \cdot \Gamma_M \, \varphi  )^\HH \otimes (e_a \cdot \psi) \ .
 \end{align}
\end{corollary}
This expression is valid for a general gauge one-form $A_\mu = A^a_\mu \, e_a$ on $M$ with values in the space of vector fields on $K$. Be they Killing or non-Killing with respect to $g_K$. The first term on the right-hand side contains the Dirac operator on $M$. The second term determines the couplings between gauge fields and fermions. The term with the internal Dirac operator $\sD^K$ generates mass terms for fermions on $M$. The last term is a Pauli-type coupling between the gauge field strength and spinors. It is a standard feature in Kaluza-Klein dimensional reductions. The dimensional reduction of \eqref{SimplerDecompositionDiracOperator} to $M$ will be discussed in section \ref{DimensionalReductionCPViolation}.

\begin{remark}
\label{RemarkPauliTerm1}
Consider the operator on spinors over $P$ defined by the Pauli term,
\beq
C (\varphi^\HH \otimes \psi) \, := \,   \frac{1}{8} \, \,  (F_A^a)^{\mu \nu } \,  (  X_\mu \cdot X_\nu \cdot \Gamma_M \, \varphi  )^\HH \otimes (e_a \cdot \psi) \ .
\eeq
It is algebraic and anti-self-adjoint with respect to the pairing $\langle \cdot, \cdot \rangle$ of spinors. Thus, the modified operator $\sD^P\! - C$ retains most of the useful properties of $\sD^P$, such as ellipticity, anti-self-adjointness with respect to $\langle \cdot, \cdot \rangle$, and the implicit coupling of 4D gauge fields to spinors through the Kosmann-Lichnerowicz derivative. So one could also consider $(\sD^P\! - C) \Psi = 0$ as a candidate for the physical equation of motion for spinors on $P$. An important conceptual disadvantage of the modified operator, however, is that it is defined only for submersion metrics of the form $g_P \simeq (g_M, A, g_K)$, not for general metrics on $P$.
\end{remark}

\section[Spinor symmetries on $M_4\times K$ induced by 4D reflections and parity]{Spinor symmetries on $M\times K$ induced by reflections and parity transformations on $M$}
\label{ReflectionsParity}

Let the spacetime $P = M \times K$ be equipped with a submersion metric $g_P$ equivalent to a triple $(g_M, A , g_K)$, as in the previous section. Reflections and parity transformations on Minkowski space admit natural extensions to diffeomorphisms of $P$ that leave the internal coordinates unchanged. The extended diffeomorphisms, however, are not isometries of the submersion metric. The transformation rules of $g_P$ under those diffeomorphisms encapsulate the usual transformation rules of 4D gauge fields $A^a_\mu$ and Higgs-like fields under reflections and parity inversions.

Now suppose that $P$ has a fixed orientation and topological spin structure. The diffeomorphisms of $P$ induced by reflections and parity inversions determine transformations of the higher-dimensional spinor fields, as we will see. But the construction of spinor bundles depends on the background metric, and $g_P$ is not preserved by those diffeomorphisms. So reflections and parity inversions will relate spinors defined on different bundles over $P$. This is slightly different from the usual 4D story, in which those transformations are isometries of $g_M$, and hence define automorphisms of the same 4D spinor bundle. Despite this difference, the Kaluza-Klein extensions of reflections and parity inversions still induce symmetries of the higher-dimensional Dirac equation $\sD^P \Psi = 0$. This is the main point for us. The purpose of this section is to describe these matters concisely. It spells out the main properties of the spinor transformations on $P$ corresponding to those symmetries. To our knowledge, the extension of the analyses in \cite{GN, Wetterich84} to our type of KK models, with general submersion metrics and transformations between different spinor bundles, has not been described previously in the literature.

\subsection{Reflections of a single coordinate}

Let $M$ be an even-dimensional Minkowski space and $R: M \rightarrow M$  denote the reflection diffeomorphism that changes the sign of the $\nu$-th canonical coordinate. We denote by the same symbol its natural extension to a diffeomorphism of $P$. If $g_P$ is a submersion metric on $P$ equivalent to the triple $(g_M, A , g_K)$, then the pullback tensor $R^\ast g_P$ is a submersion metric equivalent to the triple $(g_M, R^\ast A, R^\ast g_K)$. Note that reflections are isometries of the Minkowski metric $g_M$. The extended diffeomorphism $R: P\rightarrow P$ can also be used to pushforward vector fields $W$ on $P$, which are denoted $R_\ast W$.

Now assume that $P$ has a fixed orientation and topological spin structure, as in section \ref{HorizontalVerticalSpinors} and appendix \ref{ConventionsSpinors}. The metrics $g_P$ and $R^\ast g_P$ determine two distinct spinor bundles over $P$, denoted $S_{g_\PPP}$ and $S_{\RR^\ast g_\PPP}$, respectively. Then the diffeomorphism $R$ can be lifted to a map of spinors described in the following proposition. 
\begin{proposition}  
\label{thm:ReflectionsSpinors}
Given a submersion metric $g_P$ on $P$, there exists a $\CC$-linear map of higher-dimensional spinors, denoted  $R: \Gamma (S_{g_\PPP}) \rightarrow \Gamma (S_{\RR^\ast g_\PPP})$, with the following properties:
\begin{align}
%\label{ReflectionsSpinors}
R(f\, \Psi) \ &= \  (R^\ast f)\, R(\Psi)    \label{ReflectionsFunctions}     \linebr
R \circ R (\Psi) \ &= \  \eta\; \Psi       \label{ReflectionsSquared}    \linebr
R(W \cdot \Psi) \ &= \ -\, (R_\ast W) \cdot  R(\Psi)     \label{ReflectionsCliffordMultiplications}    \linebr
R( \nabla_W^{g_\PPP}\, \Psi) \ &= \  \nabla_{R_\ast W}^{\RR^\ast g_\PPP}  \, [  R(\Psi) ]    \label{ReflectionsCovariantDerivative}  
\end{align}
for all spinors $\Psi \in \Gamma (S_{g_\PPP})$, all functions $ f \in C^\infty (P; \CC)$, all vector fields $W$ on $P$, and for some phase $\eta \in {\rm U}(1)$. This  map is unique up to multiplication by a constant phase.
\end{proposition}  
In the second equality, we have used the identity $R^\ast  R^\ast g_\PPP = g_\PPP$ to regard $R \circ R$ as an automorphism of $S_{g_\PPP}$. The third equality is a compatibility relation between spinor reflections and Clifford multiplication. The last equality is a relation between reflections and the lifted Levi-Civita connections on the spinor bundles $S_{g_\PPP}$ and $S_{\RR^\ast g_\PPP}$. 

The existence and uniqueness of these maps of higher-dimensional spinors is proved in appendix \ref{ProofsReflectionsParity}. Here we only mention that, in appropriately chosen trivializations of the bundles $S_{g_\PPP}$ and $S_{\RR^\ast g_\PPP}$ over a domain $M \times \mathcal{U}$ inside $M \times K$, the representative of the spinor $\Psi$ is a function $\Psi_{ \mathcal{U}} :  M \times \mathcal{U} \rightarrow \Delta_{m+k}$, and the representative of the reflected spinor $R (\Psi)$ is the function 
\beq
\label{LocalSpinorReflections}
(R \Psi)_{ \mathcal{U}} (x, y) \ = \ e^{i \xi}\: \Gamma_\nu \: \Psi_{ \mathcal{U}}(R(x), y) \ .
\eeq
Here $(x, y)$ are the coordinates on $M \times \mathcal{U}$; the reflection $R(x)$ acts on $x \in M$ by changing the sign of the coordinate $x^\nu$; the factor $e^{i \xi}$ is a fixed complex phase; and $\Gamma_\nu$ is the gamma matrix on spinor space $\Delta_{m+k}$ corresponding to the coordinate $x^\nu$ on $M$. The local formula \eqref{LocalSpinorReflections} extends the most common convention for reflections on 4D Minkowski space \cite{Stone, Witten2016}.\footnote{A factor $\Gamma_P$ could be inserted in \eqref{LocalSpinorReflections}. Then identities \eqref{ReflectionsCliffordMultiplications}, \eqref{ReflectionsDiracOperator} and \eqref{ReflectionsInnerProduct} would appear with a flipped sign for even-dimensional $K$.} 
Part of the existence proof is to choose and establish the consistency of the trivializations of $TP$, $S_{g_\PPP}$ and $S_{\RR^\ast g_\PPP}$ where this local formula applies. In this process, one should insist on choosing trivializations consistent with the fixed initial orientation of $P$, even though the reflection diffeomorphism inverts that orientation.

Denote by $\sD^{g_\PPP}$ and $\sD^{\RR^\ast g_\PPP}$ the standard Dirac operators on $S_{g_\PPP}$ and $S_{\RR^\ast g_\PPP}$. Denote by $\Gamma_P$ the chirality operators both on $S_{g_\PPP}$ and $S_{\RR^\ast g_\PPP}$, as defined on \eqref{ChiralOperators}. Then appendix \ref{ProofsReflectionsParity} also shows that:

\begin{proposition}  
\label{thm:PropertiesReflectionsSpinors}
The maps of spinors described in proposition \ref{thm:ReflectionsSpinors} satisfy:
\begin{align}
R \, ( \sD^{g_\PPP} \Psi) \ &= \  - \, \sD^{\RR^\ast g_\PPP}  (R \, \Psi)    \label{ReflectionsDiracOperator}    \linebr
R \, ( \Gamma_P \Psi) \ &= \ (-1)^{k+1}\,  \Gamma_P  (R \, \Psi)     \label{ReflectionsChirality}     \linebr
\langle R \, \Psi_1,  R \, \Psi_2 \rangle \ &= \ - \, (g_M)_{\nu \nu} \,\,  \langle \Psi_1, \Psi_2 \rangle \circ R   \label{ReflectionsInnerProduct}        
\end{align}
for all spinors $\Psi \in \Gamma (S_{g_\PPP})$. Here $k$ denotes the dimension of $K$, the inner product of spinors $\langle \cdot , \cdot \rangle$ is defined in \eqref{GeneralInnerProduct}, and  $\nu$ is the index of the reflected coordinate in $M$.
\end{proposition}  
The first identity implies that if $\Psi$ is in the kernel of $\sD^{g_\PPP}$, then the reflected spinor $R(\Psi)$ is in the kernel of $\sD^{\RR^\ast g_\PPP}$. In this sense,  reflections are a symmetry of the higher-dimensional massless Dirac equation. The second identity shows that spinor reflections are a symmetry of the Weyl equation only when $K$ is odd-dimensional. The third identity shows that a reflection will preserve or not the inner product of spinors depending on whether the reflected coordinate is spacelike or timelike.

\subsection{Parity transformations}

Let $M$ be an even-dimensional Minkowski space and $\parity: M \rightarrow M$ denote the parity diffeomorphism that changes the sign of all spatial coordinates. The extension of this diffeomorphism to $P$ is denoted by the same symbol. If $g_P$ is a submersion metric on $P$ equivalent to the triple $(g_M, A , g_K)$, then the pullback tensor $\parity^\ast g_P$ is a submersion metric equivalent to the triple $(g_M, \parity^\ast A, \parity^\ast g_K)$. The extended diffeomorphism $\parity: P\rightarrow P$ can also be used to pushforward vector fields $W$ on $P$, which are denoted $\parity_\ast W$.

Assume that $P$ has a fixed orientation and topological spin structure, as in section \ref{HorizontalVerticalSpinors} and appendix \ref{ConventionsSpinors}. The metrics $g_P$ and $\parity^\ast g_P$ determine two distinct spinor bundles over $P$, denoted $S_{g_\PPP}$ and $S_{\sparity^\ast g_\PPP}$, respectively. Then the diffeomorphism $\parity$ can be lifted to a map of spinors described in the following proposition.
\begin{proposition}  
\label{thm:ParitySpinors}
Given a submersion metric $g_P$ on $P$, there exists a $\CC$-linear map of higher-dimensional spinors, denoted  $\parity: \Gamma (S_{g_\PPP}) \rightarrow \Gamma (S_{\sparity^\ast g_\PPP})$, with the following properties:
\begin{align}
\label{ParitySpinors}
\parity(f\, \Psi) \ &= \  (\parity^\ast f)\, \parity(\Psi)       \linebr
\parity \circ \parity (\Psi) \ &= \  \eta\; \Psi         \linebr
\parity(W \cdot \Psi) \ &= \  - \, (\parity_\ast W) \cdot  \parity(\Psi)        \linebr
\parity( \nabla_W^{g_\PPP}\, \Psi) \ &= \  \nabla_{\parity_\ast W}^{\sparity^\ast g_\PPP}  \, [  \parity(\Psi) ]   
\end{align}
for all spinors $\Psi \in \Gamma (S_{g_\PPP})$, all functions $ f \in C^\infty (P; \CC)$, all vector fields $W \in \Gamma (TP)$, and for some phase $\eta \in {\rm U}(1)$. This  map is unique up to multiplication by a constant phase.
\end{proposition}  

The parity diffeomorphism of $M$ can be written as a sequence of reflections of the spatial coordinates, 
\beq
\parity (x) \ = \ R_1 \cdots R_{m-1} (x) \ .
\eeq
Using the properties of reflections stated in proposition \ref{thm:ReflectionsSpinors}, one can easily verify that maps of spinors of the form
\beq
\label{SpinorParityReflections}
\parity (\Psi) \ := \ e^{i \zeta} \, R_1 \cdots R_{m-1}\,  \Psi \ 
\eeq
satisfy all the properties of parity maps, as stated in proposition \ref{thm:ParitySpinors}. This proves the existence part of proposition \ref{thm:ParitySpinors}. The uniqueness part can be proved in a way entirely analogous to the proof of proposition \ref{thm:ReflectionsSpinors}. In appropriately chosen trivializations of the spinor bundles $S_{g_\PPP}$ and $S_{\sparity^\ast g_\PPP}$ over a domain $M \times \mathcal{U}$, as in the discussion leading to \eqref{LocalSpinorReflections}, the local formula for the spinor parity transformations is 
\beq
\label{LocalSpinorParity}
(\parity \Psi)_{ \mathcal{U}} (x, y) \ = \ e^{i [\zeta + (m-1) \xi ]}\;  \Gamma_{1} \cdots \Gamma_{m-1} \;\Psi_{ \mathcal{U}}(\parity(x), y) \ .
\eeq
Here the $ \Gamma_{l}$ are gamma matrices on $\Delta_{m+k}$ corresponding to the spacelike directions in $M$, as in \eqref{CliffordMultiplication}. We note that after dimensional reduction and the redefinition of 4D spinors stated in \eqref{Redefinition4DSpinors} --- which is necessary to bring the 4D Dirac equation to its traditional form --- the 4D component of parity inversion acts through multiplication of the 4D spinor by $\gamma_0$ only, as in the usual prescription. See \eqref{Redefinition4DParity}.

Using \eqref{SpinorParityReflections} and the properties of reflections stated in proposition \ref{thm:PropertiesReflectionsSpinors}, one can also verify directly that:
\begin{proposition}  
\label{thm:PropertiesParitySpinors}
The maps of spinors described in proposition \ref{thm:ParitySpinors} satisfy:
\begin{align}
\label{PropertiesParitySpinors}
\parity \, ( \sD^{g_\PPP} \Psi) \ &= \  - \, \sD^{\sparity^\ast g_\PPP}  (\parity \, \Psi)       \linebr
\parity \, ( \Gamma_P \Psi) \ &= \ (-1)^{k+1}\,  \Gamma_P  (\parity \, \Psi)         \linebr
\langle \parity \, \Psi_1,  \parity \, \Psi_2 \rangle \ &= \ -  \, \langle \Psi_1, \Psi_2 \rangle  \circ \parity     
\end{align}
for all spinors $\Psi \in \Gamma (S_{g_\PPP})$. Here $k$ denotes the dimension of $K$ and $\langle \cdot , \cdot \rangle$ is the inner product of spinors defined in \eqref{GeneralInnerProduct}.
\end{proposition}  
These identities imply that parity transformations are always a symmetry of the Dirac equation. They are a symmetry of the Weyl equation only for odd-dimensional $K$. They flip the sign of inner products of spinors for all even-dimensional, Minkowski $M$.

\section{Conjugation symmetries on spinor bundles}
\label{ConjugationSymmetries}

 This section gives a succinct account of conjugation automorphisms of spinor bundles in arbitrary signature $(s,t)$. It describes the essential properties of these symmetries of the massless Dirac equation. It adopts a geometric approach and hides the basis-dependent constructions, with gamma matrices in special representations, that are generally used to prove the existence of the maps. The vector bundle conjugation maps are a modest extension of the vector space version, described in appendix \ref{ConventionsSpinors} using the traditional Majorana forms. The bundle version includes the relations of conjugations with covariant derivatives, the Kosmann-Lichnerowicz derivative and the Dirac operator. Conjugation maps in arbitrary signatures are well studied. See \cite{Figueroa, Stone, Hitoshi}, for instance. The purpose of this section is to provide a short and hopefully clean summary that covers the general signatures and the bundle-related properties.
 At the end, it also describes results that help to understand conjugation maps in the specific setting of Riemannian submersions and their relations with reflections and parity transformations.

\subsection{Conjugation maps}

Let $P$ be a general oriented, connected manifold equipped with a metric of signature $(s,t)$ and a spin structure. Let $\nabla$ denote the Levi-Civita connection on $TP$ and also its lift to the complex spinor bundle $S_{g_\PPP}$. Define the sets
\begin{equation}
\label{DefinitionSetConjMaps}
H_{s,t} \ := \ 
\begin{cases}
\{-1, 1 \}  &  \text{if $s-t$ is even} \\
\{ (-1)^{\frac{s-t+1}{2}} \}  & \text{if $s-t$ is odd} \ .
\end{cases}
\end{equation}
The results in appendices \ref{AppendixMajoranaConjugation} and \ref{AppendixSpinorBundles} imply that:
\begin{proposition}
\label{thm: ConjLinearAutomorphisms}
For each value $\sigma \in H_{s,t}$, there exists a conjugate-linear automorphism of spinors, denoted $j_\sigma: \Gamma (S_{g_\PPP} ) \rightarrow \Gamma (S_{g_\PPP} )$, with the following properties:
\begin{align}
j_\sigma (\psi_1 + f \, \psi _2)\ &= \ j_\sigma (\psi_1) \;  + \;  \conj{f} \,  j_\sigma (\psi _2)    \linebr
j_\sigma (V \cdot \psi) \ &=  \  \sigma \ V \cdot    j_\sigma (\psi)   \label{CommutatorConjCliffordMultiplication} \linebr
j_\sigma  j_\sigma (\psi) \ &= \ 
\begin{cases}
\; (-1)^{ \lfloor \frac{s-t}{4} \rfloor } \ (-\sigma)^{\frac{s-t}{2}} \, \psi &  \text{if $s-t$ is even} \\
\; (-1)^{\frac{(s-t)^2 -1}{8}} \, \psi & \text{if $s-t$ is odd}
\end{cases}      \label{SquareConjLinearAutomorphisms} \linebr
j_\sigma (\nabla_V \psi) \ &=\  \nabla_V (j_\sigma \, \psi)
\end{align}
for all spinors $\psi \in \Gamma (S_{g_\PPP} )$, all functions $f \in C^\infty (P; \CC)$ and all vector fields $V$ on $P$. This automorphism is unique up to multiplication by a constant complex phase.
\end{proposition}
A conjugation map with $j_\sigma  j_\sigma = 1$ is called a real structure on the spinor bundle. It allows the consistent imposition of the Majorana condition on spinors, $j_\sigma(\psi) = \psi$. A conjugation with $j_\sigma  j_\sigma = -1$ is called a quaternionic (or pseudo-real) structure on $S_{g_\PPP}$. It determines an action of the quaternions on spinors through the representation $(\mathbf{i}, \mathbf{j}, \mathbf{k}) \rightarrow (i, j_\sigma, i\, j_\sigma)$. The sign $\sigma$ tells us whether $j_\sigma$ commutes or anticommutes with Clifford multiplication and the Dirac operator.

Let $\cL_V \psi$ denote the Kosmann-Lichnerowicz derivative of the spinor $\psi$ along the vector field $V$. Let $\sD$ denote the standard Dirac operator on $S_{g_\PPP}$. Let $\Gamma_P$ denote the complex chirality operator on $S_{g_\PPP}$, normalized so that $\Gamma_P \Gamma_P = 1$. Applying the properties of $j_\sigma$ in proposition \ref{thm: ConjLinearAutomorphisms} to the definition of those operators, it is straightforward to obtain:
\begin{proposition}
\label{thm: PropertiesConjLinearAutomorphisms}
The automorphisms of proposition \ref{thm: ConjLinearAutomorphisms} have the additional properties:
\begin{align}
j_\sigma (\cL_V \psi) \ &= \ \cL_V (j_\sigma \psi)  \label{CommutatorConjKLDerivative}   \linebr
j_\sigma (\sD \,  \psi) \ &=  \   \sigma  \ \sD (j_\sigma \, \psi)   \label{CommutatorConjDiracOperator} \linebr
j_\sigma  ( \Gamma_P \, \psi) \ &= \ 
\begin{cases}
\; (-1)^{ \frac{s-t}{2}} \ \Gamma_P \; j_\sigma (\psi) &  \text{if $s-t$ is even} \\
\; \Gamma_P \; j_\sigma (\psi) & \text{if $s-t$ is odd} \ 
\end{cases}        \label{CommutatorConjChiralityOperator}  \linebr
j_{-}  (\psi) &= \ \zeta \; \Gamma_P \, j_+ (\psi)   \quad   \qquad \qquad  \text{if $s-t$ is even} \, .         \label{RelationTwoConjugations}
\end{align}
for all spinors $\psi \in \Gamma (S_{g_\PPP} )$, all vector fields $V$ on $P$, and some phase $\zeta \in U(1)$.
\end{proposition}
Thus, conjugation maps are always symmetries of the Dirac equation $\sD \psi = 0$. They preserve the condition $\Gamma_P \, \psi = \psi$ in all signatures except $s-t = 2$ (mod  4).

 One can also write useful relations between the conjugation maps $j_\sigma$ and the spinor inner products in \eqref{GeneralInnerProduct}.
The relations follow directly from lemma \ref{thm: MorePropertiesConjLinearMaps} and property \eqref{CommutatorConjDiracOperator}.
\begin{proposition}
\label{thm: MorePropertiesConjLinearAutomorphisms}
The maps $j_\sigma$ satisfy:
\begin{align}
\langle j_\sigma \, \psi_1 , \,  j_\sigma \, \psi_2 \rangle \ &= \ 
\begin{cases}
(-1)^{\frac{t(s+1)}{2}}\, \sigma^t \;   \conj{\langle \psi_1 , \, \psi_2 \rangle}   &  \text{if $s-t$ is even} \\ 
(-1)^{\frac{st}{2}}\; \conj{\langle \psi_1 , \, \psi_2 \rangle}   &  \text{if $s-t$ is odd} 
\end{cases}       \label{InnerProductConjLinearAutomorphisms}  \linebr
\langle j_\sigma \, \psi_1 , \,  \sD\,  j_\sigma \, \psi_2 \rangle \ &= \ 
\begin{cases}
(-1)^{\frac{t(s+1)}{2}}\, \sigma^{t+1} \;   \conj{\langle \psi_1 , \, \sD\, \psi_2 \rangle}   &  \text{if $s-t$ is even} \\ 
(-1)^{\frac{st + s - t +1}{2}}  \; \conj{\langle \psi_1 , \, \sD \, \psi_2 \rangle}   &  \text{if $s-t$ is odd} 
\end{cases} 
\end{align}
for all spinors $\psi_1, \psi_2 \in \Gamma (S_{g_\PPP} )$ and for the inner product in \eqref{GeneralInnerProduct}.
\end{proposition}
Now suppose that the tangent bundle of $P$ decomposes as a sum of two orthogonal subbundles, $TP = E_1 \oplus E_2$. Assume that the metric inherited by $E_l$ is non-degenerate with signature $(s_l, t_l)$. Choose orientations on the summands $E_l$ compatible with the orientation of $TP$. If one of the subbundles has even rank, we have an isomorphism of spinor bundles $S_{g_\PPP} \simeq S (E_1) \otimes S (E_2)$. There is also an inherited Clifford multiplication of vector fields in $E_l$ and spinors in $S (E_l)$. In these conditions, the conjugation automorphisms of $S_{g_\PPP}$ can be reconstructed using the conjugations in the two factors $S (E_l)$.

\begin{proposition}
\label{FactorizationConjugations}
Choose a sign $\sigma \in H_{s,t}$. When $E_1$ has even rank, the conjugation automorphism $j_\sigma$ of proposition \ref{thm: ConjLinearAutomorphisms} can be decomposed as
\begin{equation}
\label{ConjugationProduct}
j_\sigma  \ =\   j^1_{\sigma} \otimes j^2_{\sigma'}
\end{equation}
under the isomorphism $S_{g_\PPP} \simeq S (E_1) \otimes S (E_2)$. Here we have defined $\sigma' = (-1)^{\frac{s_1 -t_1}{2}} \sigma$. It can be checked that $\sigma' \in H_{s_2, t_2}$.
\end{proposition}

\subsection{Composition of conjugations and parity transformations}

Suppose now that  $P$ is a spacetime of the form $M\times K$, where $M$ is an even-dimensional Minkowski space. As discussed in section \ref{ReflectionsParity}, reflections and parity inversions on $M$ extend to diffeomorphisms of $P$ that can be lifted to transformations of spinors over $P$. One can verify that those transformations commute with conjugations of spinors, essentially.

\begin{lemma}  
Fix the sign $\sigma$ and the corresponding conjugate-linear automorphisms $j_\sigma$ of $S_{g_\PPP}$, $S_{R^\ast g_\PPP}$ and $S_{\parity^\ast g_\PPP}$. For phase choices in definitions \eqref{LocalSpinorReflections} and \eqref{SpinorParityReflections} satisfying $e^{2i \xi} = \sigma$ and $e^{2i \zeta} = 1$, all spinor reflections and parity transformations commute with $j_\sigma$.
\end{lemma} 
Choosing such a value for the phase of $\parity$, we now consider the composed maps $j_\sigma \, \parity = \parity \, j_\sigma$ between spinors in $S_{g_\PPP}$ and spinors in $S_{\sparity^\ast g_\PPP}$.
\begin{proposition}
\label{thm: MorePropertiesCP}
The composed maps $j_\sigma \parity$ satisfy:
\begin{align}
j_\sigma \parity\, (\sD \,  \Psi) \ &=  \  - \, \sigma  \ \sD (j_\sigma \parity \, \Psi)   \label{CommutatorCPDiracOperator} \linebr
j_\sigma \parity \, ( \Gamma_P \, \Psi) \ &= \ 
\begin{cases}
\; (-1)^{ \frac{m+k}{2}} \ \Gamma_P \; ( j_\sigma \parity \, \Psi) &  \text{if $k$ is even} \\
\; \Gamma_P \; ( j_\sigma \parity \, \Psi) & \text{if $k$ is odd} \ 
\end{cases}    \label{CommutatorCPChiralityOperator}   \linebr
\langle j_\sigma \parity \, \psi_1 , \,  j_\sigma  \parity \, \psi_2 \rangle \ &= \ 
\begin{cases}
(-1)^{\frac{m+k+2}{2}}\, \sigma \;   \conj{\langle \Psi_1 , \, \Psi_2 \rangle}   &  \text{if $k$ is even} \\ 
(-1)^{\frac{m+k+1}{2}}\; \conj{\langle \Psi_1 , \, \Psi_2 \rangle}   &  \text{if $k$ is odd} 
\end{cases}  \, 
\end{align}
for all spinors $\Psi_1, \Psi_2 \in \Gamma (S_{g_\PPP} )$ and for the inner product in \eqref{GeneralInnerProduct}.
\end{proposition}
Thus, the transformations $j_\sigma \parity$ are always symmetries of the Dirac equation $\sD \,  \Psi = 0$ in higher dimensions. They also preserve the Weyl condition $\Gamma_P \Psi = \Psi$ except when $m+k = 2$ (mod 4). That condition is usually imposed in Kaluza-Klein models with even-dimensional $K$ in order to obtain a correlation between the Minkowski and internal chiralities of spinors, as described in \cite{Witten83, Bap3}.

\section{Actions of compact groups on $K$ and its spinors}
\label{ActionsCompactGroups}

\subsection{Lie derivatives of spinors along fundamental vector fields}

Let $K$ be an orientable, connected manifold with a fixed topological spin structure. For any given Riemannian metric $g$ on $K$, we have the complex spinor bundle $S_g \rightarrow K$. For any vector field $V$ on $K$, the Kosmann-Lichnerowicz derivative $\cL_V$ acts on sections of that bundle as in \eqref{KLDerivative}. It satisfies the formula
 \begin{align}
\label{NonClosureKLDerivative}
( \, [\cL_U, \cL_V] \, -\, \cL_{[U, V]} \,  ) \, \psi  \ =& \  \frac{1}{4}\, g^{ir} g^{js}  g^{kl} \, \big\{ \, (\Lie_U g)(v_r , v_k) \;  (\Lie_V g)(v_s , v_l)  \linebr 
 &\, - \;  (\Lie_U g)(v_s , v_k) \;  (\Lie_V g)(v_r, v_l)  \, \big\} \; v_i \cdot v_j \cdot \psi  \ . \nonumber
\end{align}
So the closure relation $[\cL_U, \cL_V] = \cL_{[U, V]}$ is satisfied when $U$ or $V$ are conformal Killing with respect to $g$, but not in general \cite{Kosmann}. In this section we show that, when a compact group acts on $K$, there is a natural modification of the Kosmann-Lichnerowicz derivative that satisfies the closure relation for all fundamental vector fields of the action. Even if they are not conformal Killing.
 
Let $G$ be a compact, connected Lie group. Suppose that $G$ acts on $K$ on the left through diffeomorphisms that preserve the orientation and the topological spin structure, but not necessarily the metric $g$. Denote by $\mathfrak{g}$ the Lie algebra of fundamental vector fields on $K$ induced by the $G$-action. These need not be Killing or conformal Killing vector fields with respect to $g$. Then we have the following result:
\begin{proposition}
There are natural derivatives of spinors, $\sL_V :  \Gamma (S_g ) \rightarrow \Gamma (S_g)$, with 
\begin{align}
\sL_{U + \lambda V} \, \psi  \ &= \ \sL_{U} \, \psi  \: + \:  \lambda\,  \sL_{V} \, \psi   \label{NewDerivativeLinearity}  \linebr
\sL_{V} \, ( \psi_1 + f\, \psi_2)  \ &= \  \sL_{V} \,  \psi_1 \: + \:   f\, \sL_{V} \, \psi_2  \: + \:   (\dd f)(V) \, \psi_2    \label{NewDerivativeDerivation}     \linebr
\Lie_V \langle  \psi_1, \, \psi_2 \rangle \ &= \ \langle  \sL_V \psi_1, \, \psi_2 \rangle \: + \: \langle  \psi_1, \, \sL_V \psi_2 \rangle  \label{NewDerivativeInnerProduct} 
\end{align}
for all vector fields $U$ and $V$ on $K$, all $\lambda \in \mathbb{R}$ and $f \in C^\infty (K, \CC)$, that also satisfy
\beq 
[\sL_{U} , \sL_{V} ]  \, \psi  \ = \  \sL_{[U, V]}  \, \psi      \label{NewDerivativeClosure} 
\eeq
when $U$ or $V$ are fundamental fields in $\mathfrak{g}$. The operator $\sL_V$ coincides with the Kosmann-Lichnerowicz derivative when $V \in  \mathfrak{g}$ is conformal Killing with respect to $g$.
\end{proposition}
The derivatives $\sL_V$ define a representation of the Lie algebra $\mathfrak{g}$ on the space of spinors $\Gamma (S_g)$. They can be constructed as follows. Define the auxiliary, average metric
\beq
\label{DefinitionAverageMetric}
\hg (U, V)  \ \  :=  \  \int_{h \in G} (r_h^\ast g) (U, V) \ \vol_G  \ .
\eeq
Here $r_h$ is the diffeomorphism $K$ corresponding to $h$ and $\vol_G$ denotes the bi-invariant, normalized, Haar measure on $G$. Then $\hg$ is a $G$-invariant metric on $K$. All the vector fields in $\mathfrak{g}$ are Killing with respect to $\hg$. Thus, denoting $\hcL_V$ the Kosmann-Lichnerowicz derivative on $S_{\thg}$, it follows from \eqref{NonClosureKLDerivative} applied to $\hg$ that 
\beq
\label{Closure2}
[\,\hcL_U ,\, \hcL_V \, ] \ = \ \hcL_{[U, V]}  \  \ ,
\eeq
as operators on $\hg$-spinors, when $U$ or $V$ are in $\mathfrak{g}$. So these operators satisfy the desired closure relation, but on the wrong spinor bundle. We have to transport them back to $S_g$.

Now, there exists a unique smooth automorphism $\alpha : TK  \rightarrow TK$ that projects to the identity on $K$, is positive-definite, and satisfies
\begin{align}
g (\alpha(U), V)  \ &= \ g (U, \alpha(V))  \linebr 
\hg (U, V) \ &= \ g (\alpha^{-1}(U), \alpha^{-1}(V))  
\end{align}
for all vector fields $U$ and $V$ on $K$. This map lifts to an automorphism of spinor bundles $\alpha: S_g \rightarrow S_{\thg}$, which we denote by the same symbol. See \cite{BG, AHermann} and appendix \ref{SpinorsDifferentMetrics}. 
Then we define
\beq
\label{TransportedDerivativeDefinition}
\sL_V \, \psi \ := \ \alpha^{-1} \circ \hat{\cL}_V \circ \alpha (\psi)  \ .
\eeq
Due to \eqref{Closure2}, it is clear that the derivatives $\sL_V$ satisfy the closure relation \eqref{NewDerivativeClosure} on $S_g$ for all vector fields in $\mathfrak{g}$. Using formula \eqref{TransportedKosmannDerivative} of appendix \ref{SpinorsDifferentMetrics}, we also have:
\begin{proposition}
The derivatives $\sL_V$ and $\cL_V$ of spinors in $S_g$ are related by 
\beq
\label{RelationTwoDerivatives}
\sL_V \, \psi \ = \ \cL_V\, \psi \; + \;  \frac{1}{4} \; \sum_{j\neq k} \; g \big( \alpha^{-1} (\Lie_V \alpha) (v_j),\, v_k \big) \; v_j \cdot v_k \cdot \psi \ 
\eeq
for all vector fields $V$ on $K$. Here $\Lie_V \alpha$ denotes the standard Lie derivative of $\alpha$. 
\end{proposition}
\noindent
The two operators coincide when $V$ is conformal Killing. Their first variations also agree when $V$ is close to being conformal Killing. This follows from proposition \ref{InvarianceKLDerivativesKilling} and \eqref{FirstVariationKLDerivative}.
\begin{remark}
When $K$ is not compact, the average of the metric $g$ with respect to a maximal compact group acting on $K$ is not necessarily a metric $\hg$ with maximal isometry group, or with maximal Killing algebra. For example, when $K=\mathbb{R}^n$, it is better to take $\hg$ to be the Euclidean metric. Then the transport of the Kosmann-Lichnerowicz derivative $\hat{\cL}_V = \partial_V   - \frac{1}{8} (\partial_j V_k -   \partial_k V_j ) \hat{\gamma}^j \hat{\gamma}^k$ through the map $\alpha$, as in \eqref{TransportedDerivativeDefinition}, gives us a spinor derivative on $(\mathbb{R}^n, g)$ that satisfies the closure relation for all vector fields induced by the action of the Euclidean conformal group on $\mathbb{R}^n$. 
\end{remark}

\section{Representation spaces versus $\slashed{D}$-eigenspaces}
\label{RepresentationSpaces}

\subsection{Spinorial representation spaces}

The new derivatives $\sL_V$ determine a representation of the Lie algebra of fundamental $G$-vector fields on the space of spinors in $S_g(K)$. For non-isometric $G$-actions, this representation does not preserve the eigenspaces of the Dirac operator. In a Kaluza-Klein model with $K$ as the internal manifold, it couples the 4D massive gauge fields to fermions. This coupling can mix fermions with different masses and allows chiral interactions \cite{Bap3}.
\begin{proposition}
\label{thm: DefinitionUnitaryRep}
The map $\rho: \mathfrak{g} \times \Gamma (S_g) \rightarrow \Gamma (S_g)$ given by 
\beq
\label{DefinitionUnitaryRep}
\rho_V (\psi) \ := \ \sL_V \psi \; + \; \frac{1}{2} \, \divergence_g(V) \, \psi
\eeq
defines a unitary representation of the Lie algebra $\mathfrak{g}$ on the space of complex spinors equipped with its natural $L^2$ Hermitian product.
\end{proposition}
Now consider the Dirac operator $\sD$ on $S_g$.
When $K$ is compact, the space of $L^2$-integrable spinors can be decomposed as an infinite, orthogonal sum of $\sD$-eigenspaces,
\beq
\label{DecompositionDiracEigenspaces}
L^2 (S_g)  \  = \  \bigoplus_{m \in \mathbb{Z}}  E_{m}   \ .
\eeq
When the dimension of $K$ is not 3 (mod 4), the spectrum of $\sD$ is symmetric and the eigenvalues associated to $E_m$ satisfy $\mu (-m) = - \mu (m)$ \cite{Ginoux}.
The $\sD$-eigenspaces $E_m$ are not preserved by $\rho$ and the derivatives $\sL_V$, in general. When $G$ does not act on $K$ through isometries, \eqref{CommutatorNewDerivativeDiracOperator} implies that $\rho_V$ does not commute with $\sD$.

Nevertheless, it is possible to define a second decomposition of the space of spinors as a sum of finite-dimensional, complex spaces on which $\mathfrak{g}$ acts irreducibly through $\rho$. These are the $G$-representation spaces $W_{m, \pi}$. They are defined by taking decomposition \eqref{DecompositionDiracEigenspaces} for the spinors of the averaged metric $\hg$; decomposing the $\sD_{\thg}$-eigenspaces into irreducible subspaces $\hat{E}_{m, \pi}$; and then transporting these subspaces back to $S_g$ through the map $\alpha$. So $W_{m, \pi} = \alpha^{-1}(\hat{E}_{m, \pi})$. This is explained in appendix \ref{ProofsRepresentationSpaces}, which contains all the proofs for this section.
\begin{proposition}
\label{thm: DecompositionRepresentationSpaces1}
When $K$ is compact, there is an orthogonal decomposition
\beq
\label{DecompositionRepresentationSpaces1}
L^2 (S_g)  \  = \   \bigoplus_{m \in \mathbb{Z}}  \bigoplus_\pi \, n_{m, \pi} \, W_{m, \pi} \ ,
\eeq
preserved by the representation \eqref{DefinitionUnitaryRep}. Here $n_{m, \pi} \in \mathbb{N}_0$ and the $W_{m, \pi}$ are finite-dimensional spaces of spinors on which the $\mathfrak{g}$-representation \eqref{DefinitionUnitaryRep} is equivalent to the irreducible representation $\pi$. If $G$ acts on $K$ isometrically, and hence preserves the $\sD$-eigenspaces in \eqref{DecompositionDiracEigenspaces}, the representation spaces can be chosen so that $W_{m, \pi} \, \subseteq \, E_m$.
\end{proposition}
\begin{remark}
\label{thm: RemarkComplexRepresentations1}
In \eqref{DecompositionRepresentationSpaces1}, the sum $\bigoplus_\pi$ ranges over all inequivalent irreducible representations of the compact Lie algebra $\mathfrak{g}$. The non-negative integers $n_{m, \pi}$ are the multiplicities of the representations. For any fixed value of $m$, only a finite number of them are non-zero as $\pi$ varies. The representations $W_{m, \pi}$ may be real, complex or quaternionic type. For the complex ones, one can guarantee that $n_{m, \pi} = n_{m, \conj{\pi}}$ when the dimension of $K$ is not 1 (mod 4).  The real representations appear with even multiplicity when $K$ has dimension 2, 3, 4 (mod 8).  The quaternionic ones appear with even multiplicity when $K$ has dimension 0, 6, 7 (mod 8). All this is justified in appendix \ref{ProofsRepresentationSpaces}.
\end{remark}
When the manifold $K$ is even-dimensional, this decomposition of the space of spinors can be organized differently. We can use the splitting $S_g = S_g^+ \oplus S_g^-$ and require the representation spaces to be made of Weyl spinors. So we get a third decomposition:
\begin{proposition}
\label{thm: DecompositionRepresentationSpaces2}
When $K$ is compact and even-dimensional, there is an orthogonal decomposition
\beq
\label{DecompositionRepresentationSpaces2}
L^2 (S_g)  \  = \   \bigoplus_{m \in \mathbb{N}_0}  \bigoplus_\pi \, n_{m, \pi} \, (W^+_{m, \pi} \oplus W^-_{m, \pi}) \ ,
\eeq
preserved by the representation \eqref{DefinitionUnitaryRep}. Here $n_{m, \pi} \in \mathbb{N}_0$ and the $W^\pm_{m, \pi}$ are finite-dimensional spaces of Weyl spinors on which the $\mathfrak{g}$-representation \eqref{DefinitionUnitaryRep} is equivalent to the irreducible representation $\pi$. If $G$ acts on $K$ isometrically, the representation spaces can be chosen so that $W^\pm_{0, \pi} \, \subseteq \, E_0$ and, for all $m > 0$,
\beq
W^+_{m, \pi} \oplus W^-_{m, \pi} \ = \  E_{m, \pi} \oplus E_{-m, \pi} \, 
\eeq
where $E_{\pm m, \pi}$ is a subspace of $E_{\pm m}$ on which $\mathfrak{g}$ acts irreducibly through $\pi$.
\end{proposition}
\begin{remark} 
\label{thm: RemarkComplexRepresentations2}
As before, in \eqref{DecompositionRepresentationSpaces2} the sum $\bigoplus_\pi$ ranges over all inequivalent irreducible $\mathfrak{g}$-representations. For fixed $m$, only a finite number of multiplicities $n_{m, \pi}$ are non-zero as $\pi$ varies. For the representations of complex type, one can again guarantee that  $n_{m, \pi} = n_{m, \conj{\pi}}$. When the dimension of $K$ is 0 (mod 8), the multiplicity $n_{m, \pi}$ of any quaternionic representation $\pi$ must be even. When the dimension of $K$ is 4 (mod 8), the multiplicity $n_{m, \pi}$ of any real representation $\pi$ must be even. The proof of this is in appendix \ref{ProofsRepresentationSpaces}.
\end{remark}

The observation below helps to interpret the effect of conjugations on spinors. It will be useful in the KK setting (see remark \ref{InterpretationHDConjugations}). It extends well-known properties of conjugations on finite-dimensional vector spaces. It is proved in appendix \ref{ProofsRepresentationSpaces}.
\begin{lemma}
\label{SpecialBasesInternalSpinors}
There exists an $L^2$-orthonormal basis $\{ \psi^\alpha \}$ of spinors in $S_g$ satisfying:
\begin{itemize}
\item[1)] Each spinor $\psi^\alpha$ belongs to one of the irreducible representation spaces $W_{m, \pi}$ in \eqref{DecompositionRepresentationSpaces1}, when $k$ is odd, and to one of the irreducible spaces $W_{m, \pi}^\pm$ in \eqref{DecompositionRepresentationSpaces2}, when $k$ is even.
\item[2)] The conjugation maps $j_\sigma$ may preserve some basis spinors, up to a phase, and otherwise will swap them in pairs and multiply them by phases.
\end{itemize}
\end{lemma}

\subsection{Chiral fermions and local gauge anomalies}

Consider the case of even-dimensional $K$. In decomposition \eqref{DecompositionRepresentationSpaces2}, the chiral subspaces $W^+_{m, \pi}$ and $W^-_{m, \pi}$ transform under the same $\mathfrak{g}$-representation and appear with the same multiplicity $n_{m, \pi}$. So the Lie algebra $\mathfrak{g}$ acts isomorphically on the two big spaces of Weyl spinors, $L^2 (S^+_g) $ and $L^2 (S_g^-) $. However, this does not imply that $\mathfrak{g}$ acts identically on the two chiral components of any given $\sD$-eigenspinor, $\psi =  \psi^+ + \psi^-$. Just as in \cite{Bap3}, a simple calculation using the properties of $\sD$ and $\Gamma_K$ yields 
\beq
\label{RelationCommutatorChirality}
\int_K \langle\, \phi, \, [\sD\!,\,  \rho_V] \, \Gamma_K  \psi \, \rangle \, \vol_{g}\, = \, (\mu + \mu') \int_K \big\{ \langle\, \phi^+ , \: \rho_V \psi^+  \, \rangle   -   \langle\, \phi^- , \: \rho_V \psi^-  \, \rangle \big\} \, \vol_{g} \ . 
\eeq
Here $\psi$ and $\phi$ are two $\sD$-eigenspinors with positive eigenvalues $\mu$ and $\mu'$. When the $G$-action on $K$ is non-isometric, the commutator $[\sD\!,\,  \rho_V]$ is in general non-zero, as follows from \eqref{CommutatorNewDerivativeDiracOperator}. So we should have different matrix elements $\blangle\, \phi^+ , \: \rho_V \psi^+  \, \brangle_{L^2}   \neq  \blangle\, \phi^- , \: \rho_V \psi^-  \, \brangle_{L^2}$. In contrast, when $G$ acts on $K$ through isometries, we have $\rho_V = \cL_V$ and $[\sD\!,\,  \rho_V]= 0$. So the $\mathfrak{g}$-action on $\sD$-eigenspinors cannot have chiral asymmetries. A detailed discussion of the appearance of chiral couplings for non-isometric actions on $K$ is given in \cite{Bap3}. 

Note that although $W^+_{m, \pi}$ and $W^-_{m, \pi}$ transform under the same gauge representation, for non-isometric actions the two spaces may be spanned by $\sD$-eigenspinors with different eigenvalues, or by distinct linear combinations of such $\sD$-eigenspinors. What does this imply for a KK model with internal space $K$? It implies that the right- and left-handed components that transform in a given representation may belong to different 4D fermions with distinct masses, rather than be the two chiral components of the same fermion.
For example, suppose that light fermions interact with a force only through their left-handed components, as in the Standard Model. Then there should be fermions interacting similarly with that force through their right-handed components, but now they can be heavy fermions. They need not have the same mass.

Remark \ref{thm: RemarkComplexRepresentations2} says that we always have an equality of multiplicities $n_{m, \pi} = n_{m, \conj{\pi}}$ in decomposition \eqref{DecompositionRepresentationSpaces2}. In particular, inside each chiral subspace $W_m^\pm = \bigoplus_\pi \, n_{m, \pi} \, W^\pm_{m, \pi}$, the irreducible representations of complex type always appear in conjugate pairs. This has useful consequences in our Kaluza-Klein models with internal space $K$. It guarantees that the chiral interactions of 4D fermions generated by the $\mathfrak{g}$-representation $\rho$ are free of local gauge anomalies \cite[sec. 75]{Srednicki}. It does not say anything about the interactions generated by the higher order correction terms $\tau_{e_a}$ in the 4D equations \eqref{DimensionReductionDiracEquation2}. Self-conjugacy of the gauge representations holds whenever $\dim K \neq 1$ (mod 4), as follows from remark \ref{thm: RemarkComplexRepresentations1}.

\subsection{Group representations and fermion generations}

The irreducible unitary representations of the Lie algebra $\mathfrak{g}$ on the finite-dimensional spaces of spinors $W_{m, \pi}$ and $W^\pm_{m, \pi}$ can always be integrated to representations of a covering group $\tilde{G}$ of $G$ on the same spaces. So a non-isometric $G$-action on $(K, g)$ preserving the topological spin structure determines a unitary $\tilde{G}$-representation on the space of spinors on $S_g$. This group representation extends the usual one for isometric actions. Note that, with this simple definition, an element $h \in \tilde{G}$ that does not preserve the metric $g$ still transforms $g$-spinors to $g$-spinors, not $g$-spinors to $(r_h^\ast g)$-spinors. That transformation, however, will not commute with the Dirac operator $\sD^g$, in general.

The elements in the kernel of the covering $\tilde{G} \rightarrow G$ act trivially on $K$, but may act non-trivially on spinors. When $G$ is compact semisimple, so is the cover $\tilde{G}$. Then the centre of $\tilde{G}$ and the kernel of the covering map are both finite subgroups.
Now let $\mathfrak{g}' \subset \mathfrak{g}$ be the subalgebra of fundamental Killing fields on $K$. Assume that it is semisimple. Let $\tilde{G}' \subset \tilde{G}$ be the corresponding exponentiated subgroup. The Dirac operator on $S_g$ commutes with the infinitesimal transformations $\rho_V$ when $V \in \mathfrak{g}'$ is Killing. So it also commutes with the exponentiated transformations determined by the elements of $\tilde{G}'$.

A first observation about this lift to spinors of non-isometric actions on $K$ is the following. Any element in the centre of $\tilde{G}$ defines a transformation of spinors that commutes with the infinitesimal gauge transformations $\rho_V$, for all $V \in \mathfrak{g}$. However, some of these central elements may lie outside the subgroup $\tilde{G}'$. So the corresponding spinor transformations need not commute with the Dirac operator on $K$.
Therefore, in the Kaluza-Klein model on $M \times K$, and assuming that $G$ is semisimple, we get a finite set of transformations between 4D fermions in the same $\mathfrak{g}$-representations but with possibly different masses. This is a potential source of different generations of the same type of particle.
We use the word {\it potential} because we have not verified that, in practice, these transformations do not commute with $\sD$. They could be trivial or act by scalar multiplication of spinors, for example, in which case they would not relate fermions with different masses.

A perhaps more pertinent observation about fermion generations is the following. Suppose that we start with a metric on $K$ with large isometry group $G$. The eigenvalues of $\sD$ are typically degenerate, with the associated eigenspinors transforming under $\mathfrak{g}$-representations, as mentioned before. Let $\mu$ be such an eigenvalue, and let $E_\mu$ be the associated $\sD$-eigenspace. Consider a perturbation of the initial metric that breaks the isometry group to $G' \subset G$. It should reduce the degeneracy of $\mu$ and produce branched eigenvalues $\mu_1, \ldots, \mu_r$ \cite{Wang, BG}. This splits the $\mathfrak{g}$-representation space into a sum of $\mathfrak{g}'$-subspaces $E_\mu \rightarrow E_{\mu_1} \oplus \cdots \oplus E_{\mu_r}$ associated with the new eigenvalues. Now note that the same $\mathfrak{g}'$-irreducible representation may appear in different $\sD$-eigenspaces $E_{\mu_j}$. In the Kaluza-Klein model on $M \times K$, this corresponds to 4D fermions in the same $\mathfrak{g}'$-representation but with slightly different masses, all branched off the initial eigenvalue $\mu$. So this is another potential mechanism for the emergence of different generations of the same particle. In this scenario, the masses of the distinct generations should differ by amounts controlled by the scale of the initial eigenvalue $\mu$ and by the scale of the metric perturbation producing the symmetry breaking $G \rightarrow G'$ (presumably the electroweak scale). So the mass differences should not be at the Planck scale. It would be interesting to study this phenomenon in explicit examples, such as the one described in appendix \ref{Example}.

\section{Dimensionally reduced equations and their CP transformation}
\label{DimensionalReductionCPViolation}

\subsection{Dimensional reduction of the Dirac equation on $M_4 \times K$}
\label{DimensionalReductionViolation}

In this section, we take $M$ to be 4D Minkowski space. For simplicity, we also assume that the submersion metric $g_P \simeq (g_M, A, g_K)$ has constant internal geometry. So the metric $g_K$ is constant along $M$. After the work in sections \ref{ActionsCompactGroups} and \ref{RepresentationSpaces}, the action of the higher-dimensional Dirac operator \eqref{SimplerDecompositionDiracOperator} on a spinor $\varphi^\HH \otimes \psi $ can be written as
\begin{align}
\label{SimplerDecompositionDiracOperator2}
\sD^P (\varphi^\HH \otimes \psi ) \; =& \; \, g_M^{\mu \nu}\,  (X_\mu \cdot  \nabla^M_{\nu} \varphi)^\HH \otimes \psi  \; + \;   g_M^{\mu \nu}\, \, A_\nu^a\,  \, (X_\mu \cdot  \varphi)^\HH  \otimes [\, ( \rho_{e_a} + \tau_{e_a} ) \,  \psi\, ]  \nonumber   \linebr
&+  \; ( \Gamma_M \,  \varphi)^\HH  \otimes  \sD^K\psi   \; + \;  \frac{1}{8} \, \,  (F_A^a)^{\mu \nu } \,  (  X_\mu \cdot X_\nu \cdot \Gamma_M \,  \varphi  )^\HH \otimes (e_a \cdot \psi) \ .
 \end{align}
Here $\rho_{e_a} (\psi)$ is a unitary representation of the gauge algebra $\mathfrak{g}$ on the space of internal spinors, as defined in \eqref{DefinitionUnitaryRep}. The term $\tau_{e_a}(\psi)$ is the non-minimal coupling
\beq
\tau_{e_a} ( \psi) \ := \ - \, \frac{1}{4} \; \sum_{j\neq k} \; g \big( \alpha^{-1} (\Lie_{e_a} \alpha) (v_j),\, v_k \big) \; v_j \cdot v_k \cdot \psi 
\eeq
that emerges from \eqref{TransportedDerivativeDefinition} and \eqref{RelationTwoDerivatives}. From proposition \ref{InvarianceKLDerivativesKilling}, we know that it vanishes when the internal vector field $e_a$ is conformal Killing for the metric $g_K$. In fact, even the first variation of $\tau_{e_a} ( \psi)$ vanishes when $e_a$ is close to being conformal Killing, as shown by \eqref{FirstVariationKLDerivative}. Thus, we expect this non-minimal coupling to be very small for gauge fields $A_\mu^a$ with small mass, for which the Lie derivatives $\Lie_{e_a} g_K$ appearing in \eqref{MassFormula} are small.

Now take a general higher-dimensional spinor and write it in the form $\Psi =   \sum_\alpha \tilde{\varphi}_\alpha^\HH(x) \otimes \psi^\alpha(y)$, as in \eqref{EigenMassDecompositionSpinors}. The internal spinors $\{ \psi^\alpha \}$ are all independent and form a $L^2$-orthonormal basis of the space of sections of $S_{g_\KKK}$. Thus, decomposition \eqref{SimplerDecompositionDiracOperator2} implies that the equation $\sD^P \Psi = 0$ over $P$ is equivalent to an infinite set of equations over $M_4$ for the 4D spinors $\tilde{\varphi}_\alpha$. These equations read
\begin{multline}
\label{DimensionReductionDiracEquation}
 \gamma^\mu \,  \big\{  \nabla^{M}_{X_\mu} \,  \tilde{\varphi}_\alpha \, + \,  A^a_\mu \,  \,  \blangle \, \psi_\alpha  \, , ( \rho_{e_a} +  \tau_{e_a} ) \psi^\beta \, \brangle_{L^2} \,  \,  \tilde{\varphi}_\beta  \big\}  \,  
+ \,  \blangle \, \psi_\alpha  \, , \,  \sD^K  \psi^\beta \, \brangle_{L^2}\, \, \gamma_5 \, \tilde{\varphi}_\beta  \; \linebr +\; \frac{1}{8}\, \,  (F_A^a)_{\mu \nu } \, \blangle \psi_\alpha \, , e_a \cdot \psi^\beta \brangle_{L^2}\, \,  \gamma^\mu  \gamma^\nu \, \gamma_5 \,\tilde{\varphi}_\beta  \ = \ 0 \ , 
\end{multline}
where we sum over $\beta$ and use the traditional $\gamma^\mu$ notation. When $\psi^\beta$ is an eigenspinor of $\sD^K$ with eigenvalue $m_\beta$, the 4D equations \eqref{DimensionReductionDiracEquation} are similar to, but not identical to, the usual gauged Dirac equations. Besides the non-minimal coupling $\tau_{e_a}$ and the Pauli term, they have extra $\gamma_5$ factors and a non-standard kinetic term.  
However, denoting 
\beq
\label{Redefinition4DSpinors}
\varphi_\alpha \; := \;  \frac{1}{\sqrt{2}} \, ( I    +  i  \gamma_5 ) \, \tilde{\varphi}_\alpha \; = \; \exp{(i \pi \gamma_5 / 4)}  \, \tilde{\varphi}_\alpha  \ ,
\eeq
as in \cite[p. 22]{Duff}, it is easy to check that the redefined 4D spinors satisfy 
\begin{multline}
\label{DimensionReductionDiracEquation2}
 i \, \gamma^\mu \,  \big\{  \nabla^{M}_{X_\mu} \,  \varphi_\alpha \, + \,  A^a_\mu \,  \,  \blangle \, \psi_\alpha  \, , ( \rho_{e_a} +  \tau_{e_a} ) \psi^\beta \, \brangle_{L^2} \,  \,  \varphi_\beta  \big\}  \,  
+ \,  \blangle \, \psi_\alpha  \, , \,  \sD^K  \psi^\beta \, \brangle_{L^2}\, \,  \varphi_\beta  \; \linebr +\; \frac{1}{8}\, \,  (F_A^a)_{\mu \nu } \, \blangle \psi_\alpha \, , e_a \cdot \psi^\beta \brangle_{L^2}\, \,  \gamma^\mu  \gamma^\nu \,  \varphi_\beta  \ = \ 0 \ .
\end{multline}
So the $\gamma_5$ factors disappear and we get the traditional Dirac kinetic term in Minkowski space with signature $-+++$ and the gamma matrices conventions in appendix \ref{ConventionsSpinors}. This equation is the main result of this section.

\begin{remark}
Let $G'$ be a subgroup of $G$ that acts on $K$ through isometries, and let $\mathfrak{g}' \subset \mathfrak{g}$ be its algebra of fundamental Killing fields on $K$. Then the gauge fields $A^a_\mu$ are massless and the operators $\rho_{e_a}$ and $\sD^K$ commute for all $e_a \in \mathfrak{g}'$. Thus, as described in section \ref{RepresentationSpaces}, it is possible to choose a $L^2$-orthonormal, complex basis $\{\psi^\beta \}$ formed by $\sD^K$-eigenspinors that transform in irreducible representations of $G'$. In other words, it is possible to align the mass basis and the $G'$-representation basis of the internal spinors. For the non-Killing vector fields $e_a$ in $\mathfrak{g} \setminus \mathfrak{g}'$, which are linked to massive 4D gauge fields, this basis alignment should not be possible in general.
\end{remark}

\begin{remark}
\label{RemarkPauliTerm2}
The appearance of a 4D Pauli term is a departure from the prescriptions of the Standard Model. In the case of the abelian 5D Kaluza-Klein model, this term contributes to the anomalous magnetic moment of the charged fermion. It has been estimated that the magnitude of this contribution is negligible, beyond the reach of current measurements (e.g. \cite{CN, Dolan}). However, those estimates rely on the assumption that the higher-dimensional vacuum metric is determined by the simple Einstein-Hilbert action on $P$. That assumption should not hold in realistic models. Higher-order corrections to that action seem necessary to introduce the different mass scales observed in reality and to stabilize the vacuum metric \cite{Bap1}. So the physical contribution of the Pauli term is not completely clear yet. This term can always be removed as in remark \ref{RemarkPauliTerm1}. The resulting operator is less natural than $\sD$ but preserves its most important features.
\end{remark}

Observe that a systematic application of the field redefinition \eqref{Redefinition4DSpinors} also affects how reflections and parity inversion act on spinors on $M_4$. Denoting by $\tilde{R}_M$ the operator before the redefinition and by $R_M$ the operator after, we have 
\beq
R_M\ =\ e^{\frac{i\pi}{4} \gamma_5} \: \tilde{R}_M \: e^{-\frac{i\pi}{4} \gamma_5}  \ =\ e^{\frac{i\pi}{2} \gamma_5} \: \tilde{R}_M\ =\  i\, \gamma_5 \: \tilde{R}_M \ =\ i \, e^{i \xi} \, \gamma_5 \,  \gamma_\nu   \ ,
\eeq
where the last equality follows from definition \eqref{LocalSpinorReflections}. Similarly, for 4D parity inversion,
\beq
\label{Redefinition4DParity}
\parity_M\ =\ e^{\frac{i\pi}{4} \gamma_5} \: \tilde{\parity}_M \: e^{-\frac{i\pi}{4} \gamma_5}  \ =\ e^{\frac{i\pi}{2} \gamma_5} \: \tilde{\parity}_M\ =\  i\, \gamma_5 \: \tilde{\parity}_M  \ = \  -\,  e^{i [\zeta + 3 \xi]} \, \gamma_0  \ ,
\eeq
where the last equality uses \eqref{LocalSpinorParity}. 
The precise phases that appear in these transformations are not essential. The 4D conjugation maps, in contrast, are not affected by the field redefinition. From \eqref{CommutatorConjChiralityOperator} we know that $j^M_\sigma$ is conjugate-linear and anticommutes with $\gamma_5$, so we have
\beq
\label{RedefinitionInvariance4DConjugation}
e^{\frac{i\pi}{4} \gamma_5}  \; j^M_\sigma \; e^{-\frac{i\pi}{4} \gamma_5}  \ =\  e^{\frac{i\pi}{4} \gamma_5}  \; e^{-\frac{i\pi}{4} \gamma_5}  \; j^M_\sigma \ = \  j^M_\sigma  \ .
\eeq

\subsection{The CP-transformed equation in 4D}

Conjugations of spinors and parity inversions are exact symmetries of the massless Dirac equation on $M_4 \times K$. So is their composition $j_\sigma \parity$. This was described in sections \ref{ReflectionsParity} and \ref{ConjugationSymmetries}. So if a higher-dimensional spinor satisfies $\sD \Psi = 0$, we will always have $\sD (j_\sigma \parity \, \Psi) = 0$ as well.  The purpose of this section is to determine how the second equation looks after dimensional reduction to 4D. 

When the internal metric $g_K$ is constant along $M_4$, any higher-dimensional spinor can be decomposed as $\Psi  =   \sum_\alpha \tilde{\varphi}_\alpha^\HH(x) \otimes \psi^\alpha(y)$, as in \eqref{EigenMassDecompositionSpinors}. Here $\{\psi^\alpha \}$ is a $L^2$-orthonormal basis of the space of internal spinors in $S_{g_\KKK}$. Combining proposition \ref{FactorizationConjugations} with the local formula \eqref{LocalSpinorParity} for parity inversions, it is clear that the action of $j_\sigma \parity$ on such spinors is 
\beq
\label{DimensionalFactorizationCP}
 j_\sigma \, \parity \, (\Psi)\ = \ \sum_\alpha \,  [  j^M_\sigma\, \tilde{\parity}_M \, \tilde{\varphi}_\alpha (x)]^\HH \, \otimes \,  [ j^K_{- \sigma} \, \psi^\alpha(y)] \ .
\eeq
Here $\tilde{\parity}_M $ is the raw parity inversion in 4D, before the field redefinition \eqref{Redefinition4DSpinors}. Now denote the CP-transformed spinor in 4D by the abbreviated symbol
\beq
\tilde{\varphi}^{cp_\sigma}_\alpha \ := \ j_\sigma^M \tilde{\parity}_M \,\tilde{\varphi}_\alpha \ .
\eeq
 Just as was done in \eqref{DimensionReductionDiracEquation}, in the previous section, we can rewrite the higher-dimensional equation $\sD (j_\sigma \parity \Psi) = 0$ as an infinite set of equations for spinors over $M_4$. The only differences are that we now have $\tilde{\varphi}^{cp_\sigma}_\alpha$ instead of $\tilde{\varphi}_\alpha$, and the internal basis $ j^K_{- \sigma} \, \psi^\alpha$ instead of $\psi^\alpha$. However, identity \eqref{InnerProductConjLinearAutomorphisms} in the Riemannian case, together with the commutation relations of $j_\sigma$ with Clifford multiplication and the Dirac operator, imply that
\begin{align}
 \blangle \ j^K_{- \sigma} \psi_\alpha  \, , \,  \sD^K (j^K_{- \sigma} \psi^\beta )\, \brangle_{L^2} \ &= \ - \sigma \, \blangle \, j^K_{- \sigma} \psi_\alpha  \, , \,  j^K_{- \sigma} (\sD^K \psi^\beta) \, \brangle_{L^2}  \ = \ - \sigma \conj{\blangle \, \psi_\alpha  \, , \,  \sD^K  \psi^\beta \, \brangle}_{L^2}  \linebr
\blangle  j^K_{- \sigma}  \psi_\alpha \, , \, e_a \cdot ( j^K_{- \sigma}  \psi^\beta )\brangle_{L^2} &= \ -\, \sigma \, \blangle  j^K_{- \sigma} \psi_\alpha \, , \, j^K_{- \sigma} ( e_a \cdot \psi^\beta) \brangle_{L^2} = -\, \sigma \, \conj{\blangle \psi_\alpha \, , \, e_a \cdot \psi^\beta \brangle}_{L^2} \ .   \nonumber
\end{align}
Moreover, \eqref{InnerProductConjLinearAutomorphisms} and the fact that $j_\sigma$ commutes with the internal operators $\cL_{e_a}$, $\rho_{e_a}$ and $\tau_{e_a}$, produce the third identity
\beq
\label{RepresentationConjugation}
 \blangle \ j^K_{- \sigma} \psi_\alpha  \, , \,   ( \rho_{e_a} +  \tau_{e_a} ) (j^K_{- \sigma} \psi^\beta )\, \brangle_{L^2} \ = \ \conj{\blangle \, \psi_\alpha  \, , \,  ( \rho_{e_a} +  \tau_{e_a} )  \psi^\beta \, \brangle}_{L^2}  \ .
\eeq
So the same arguments that led to \eqref{DimensionReductionDiracEquation} now imply that $\sD (j_\sigma \parity \, \Psi) = 0$ is equivalent to the set of 4D equations
\begin{multline}
\label{DimensionReductionCPDiracEquation}
 \gamma^\mu \,  \big\{  \nabla^{M}_{X_\mu} \,  \tilde{\varphi}^{cp_\sigma}_\alpha \, + \, (A^p)^a_\mu \,  \,  \conj{\blangle \, \psi_\alpha  \, , ( \rho_{e_a} +  \tau_{e_a} ) \psi^\beta \, \brangle}_{L^2} \,  \,  \tilde{\varphi}^{cp_\sigma}_\beta  \big\}  \,  
- \, \sigma \,  \conj{\blangle \, \psi_\alpha  \, , \,  \sD^K  \psi^\beta \, \brangle}_{L^2}\, \, \gamma_5 \, \tilde{\varphi}^{cp_\sigma}_\beta  \; \linebr -\; \frac{\sigma}{8}\, \,  (F_{A^p}^a)_{\mu \nu } \, \conj{\blangle \psi_\alpha \, , e_a \cdot \psi^\beta \brangle}_{L^2}\, \,  \gamma^\mu  \gamma^\nu \, \gamma_5 \,\tilde{\varphi}^{cp_\sigma}_\beta  \ = \ 0 \ . 
\end{multline}
Here we have used that $j_\sigma \parity\, \Psi$ is a spinor in $S_{\parity^\ast g_\PPP}$, as discussed in section \ref{ReflectionsParity}, not in the original bundle $S_{g_\PPP}$. In particular, if the submersion metric $g_P$ is equivalent to the triple $(g_M, A, g_K)$, with $g_M$ the Minkowski metric and $g_K$ constant along $M_4$, then the pullback metric $\parity^\ast g_P$ is equivalent to $(g_M, A^p, g_K)$, where the new connection 1-form is just the pullback $A^p := \parity^\ast A$ of the original one. Explicitly, it is of course given by
\beq
\label{ParityGaugeFields4D}
(A^p)^a_\mu \, \, |_{x} \ :=\  (-1)^{1+ \delta_{0\mu}} \, A^a_\mu \, \,  |_{\parity(x)}  \qquad \quad {(\text{no sum over $\mu$)}} \ .
\eeq
The last step in the derivation is the customary redefinition of the 4D spinors, in order to obtain Dirac equations in their traditional guise. As in \eqref{Redefinition4DSpinors}, we define 
\beq
\label{CPTransformation4D}
\varphi^{cp_\sigma}_\alpha\ := \ e^{\frac{i\pi}{4} \gamma_5} \,  \tilde{\varphi}^{cp_\sigma}_\alpha\  =  \ j_\sigma^M \parity_M \, \varphi_\alpha \ .
\eeq
The second equality is a consequence of the transformation rules \eqref{Redefinition4DParity} and \eqref{RedefinitionInvariance4DConjugation}. In terms of the redefined spinors, equation \eqref{DimensionReductionCPDiracEquation} then takes the final form
\begin{multline}
\label{DimensionReductionCPDiracEquation2}
 i \, \gamma^\mu \,  \big\{  \nabla^{M}_{X_\mu} \,  \varphi^{cp_\sigma}_\alpha \, + \,  (A^p)^a_\mu \,  \, \conj{ \blangle \, \psi_\alpha  \, , ( \rho_{e_a} +  \tau_{e_a} ) \psi^\beta \, \brangle}_{L^2} \,  \,  \varphi^{cp_\sigma}_\beta  \big\}  \,  
- \, \sigma \,  \conj{\blangle \, \psi_\alpha  \, , \,  \sD^K  \psi^\beta \, \brangle}_{L^2}\, \,  \varphi^{cp_\sigma}_\beta  \; \linebr -\; \frac{\sigma}{8}\, \,  (F_{A^p}^a)_{\mu \nu } \, \conj{\blangle \psi_\alpha \, , e_a \cdot \psi^\beta \brangle}_{L^2}\, \,  \gamma^\mu  \gamma^\nu \,  \varphi^{cp_\sigma}_\beta  \ = \ 0 \ .
\end{multline} 
This is the main result in the section. It says that the CP-transformed spinors $\varphi^{cp_\sigma}_\alpha$ satisfy equations of motion on $M_4$ very similar to those satisfied by the original $\varphi_\alpha$, i.e. to \eqref{DimensionReductionDiracEquation2}. The differences are that the gauge form is now $A^p$, instead of $A$; we have different signs when $\sigma = 1$; and four types of matrices ---  determined by the operators $\rho_{e_a}$, $\tau_{e_a}$, $\sD^K$, and $e_a \cdot$ on the space of internal spinors --- appear complex-conjugated in the new equation.

\begin{remark}
\label{PhysicalHDConjugations} 
When $K$ is even-dimensional, hence so is $P$, there are two higher-dimensional conjugations, $j_\pm = j^M_{\pm} \otimes j^K_{\mp}$. See \eqref{DefinitionSetConjMaps} and \eqref{ConjugationProduct}. They are both symmetries of the equation $\sD^P \Psi = 0$. Since $j^K_-$ anticommutes with the internal Dirac operator $\sD^K$, the conjugation $j^M_{+} \otimes j^K_{-}$ flips the signs of all the mass terms in the 4D equations, after dimensional reduction. So the conjugation that most closely adheres to the usual conventions of 4D physics is $j_- = j^M_{-} \otimes j^K_{+}$. It exists when the dimension of $K$ is not 1 (mod 4).
\end{remark}
\begin{remark}
\label{InterpretationHDConjugations}
The action of the higher-dimensional conjugations $j_\sigma$ appearing in  \eqref{DimensionalFactorizationCP} and  \eqref{ConjugationProduct} can be better interpreted with the help of lemma \ref{SpecialBasesInternalSpinors}. In a basis of internal spinors $\{ \psi^\alpha\}$ with the properties stated in the lemma, we see that $j_\sigma$ just C-transforms the spinors $\varphi_\alpha$ on $M_4$, multiplies them by phases, and swaps some of them in pairs.
\end{remark}
\begin{remark}
When $\dim K \neq 2 \pmod{4}$, the higher-dimensional transformation $j_\sigma \parity$ preserves the conditions for Weyl spinors, $\Gamma_P \Psi = \pm \Psi$. This follows from \eqref{CommutatorCPChiralityOperator}. The CP-transformation of spinors on $M_4$, in contrast, always preserves the Weyl conditions. Using the definition $\Gamma_M = i \gamma_0  \gamma_1 \gamma_2 \gamma_3 = \gamma_5$, it is clear from definition \eqref{CPTransformation4D} and properties \eqref{Redefinition4DParity} and \eqref{CommutatorCPChiralityOperator} that $\varphi_\alpha  = \pm\, \gamma_5 \, \varphi_\alpha$ is equivalent to $\varphi^{cp_\sigma}_\alpha  =  \pm\, \gamma_5 \, \varphi^{cp_\sigma}_\alpha$, as usual in 4D.
\end{remark}

We end this section by noting that dimensional reduction commutes with CP transformations, in an appropriate sense. Firstly, recall that the CP-transformed equations \eqref{DimensionReductionCPDiracEquation2} were obtained here by dimensional reduction of the higher-dimensional equation $\sD (j_\sigma \parity \, \Psi) = 0$. However, had we started with the dimensionally reduced equations \eqref{DimensionReductionDiracEquation2} and applied to them the CP transformation $j_\sigma^M \parity_M$ on $M_4$, the result would have been exactly the same set of equations \eqref{DimensionReductionCPDiracEquation2}. This can be checked directly using the properties of CP transformations on Minkowski space that were described before. 

Secondly, at the spinor level, dimensional reduction can be regarded as a map ${\mathcal R}_{\{\psi^\alpha \}}$ sending a higher-dimensional spinor $\Psi =  \sum_\alpha \tilde{\varphi}_\alpha^\HH \otimes \psi^\alpha =  \sum_\alpha  (e^{\frac{-i\pi}{4} \gamma_5} \,\varphi_\alpha)^\HH  \otimes \psi^\alpha$ to the infinite collection $\{ \varphi_\alpha\}$ of spinors on $M_4$. This map depends on the choice of basis $\{\psi^\alpha \}$ of the internal spinor space. Then it follows from \eqref{ConjugationProduct} and our previous discussion that
\begin{equation}
{\mathcal R}_{\{ j^K_{- \sigma} \psi^\alpha \}} \circ ( j_\sigma \parity)  \ = \  ( j_\sigma^M \parity_M) \circ  {\mathcal R}_{\{\psi^\alpha \}}
\end{equation}
as maps from the higher-dimensional spinor space.

\subsection{A path for CP violation?}

The CP-related equations \eqref{DimensionReductionDiracEquation2} and \eqref{DimensionReductionCPDiracEquation2} both describe solutions of the massless Dirac equation on $M \times K$. They are broadly similar equations, but not exactly the same.  The representation matrices $\blangle \, \psi_\alpha  \, , \,  \rho_{e_a}  \psi^\beta \, \brangle_{L^2}$, the CKM-like mass matrix $\blangle \, \psi_\alpha  \, , \,  \sD^K  \psi^\beta \, \brangle_{L^2}$, the Pauli term matrix $\blangle \psi_\alpha \, , \, e_a \cdot \psi^\beta \brangle_{L^2}$, and the non-minimal gauge coupling matrix $\blangle \, \psi_\alpha  \, , \, \tau_{e_a}  \psi^\beta \, \brangle_{L^2}$, are all infinite-dimensional complex matrices, a priori. And they appear conjugated in \eqref{DimensionReductionCPDiracEquation2} but not in the original \eqref{DimensionReductionDiracEquation2}. As they stand, the two sets of equations are different and there could to be scope for CP violation in these models. 

That was the author's original view. When we have massive gauge fields linked to non-Killing  vector fields $e_a$ on $K$, the respective spinor representions $\rho_{e_a}$ do not commute with $\sD^K$ and do not preserve the $\sD^K$-eigenspaces. There is a fundamental misalignment between the representation bases and the mass bases of spinors. The complex, CKM-like matrices relating these two types of bases should have complex phases that are hard to get rid of, at least for generic internal geometries.

In fact, that is not the case. The next section shows that there are always unitary transformations ${\mathrm U}_\sigma$ on the space of spinors that can transform \eqref{DimensionReductionCPDiracEquation2} back into \eqref{DimensionReductionDiracEquation2}. The two sets of CP-related Dirac equations are therefore equivalent, even when massive gauge fields are present. Instead of a path for CP violation, we obtain a general no-go result. For the dimensionally reduced Weyl equation, in contrast, the results are more subtle.

%\begin{remark}
%When the dimension of $K$ is 0, 6 or 7 (mod 8), the physical, higher-dimensional conjugation $j_- = j^M_{-} \otimes j^K_{+}$ of remark \ref{PhysicalHDConjugations} has an internal part that satisfies $j^K_{+} j^K_{+} = 1$. So the space of complex spinors on $K$ has a basis made entirely of Majorana spinors satisfying $j^K_{+}  \psi_\alpha = \psi_\alpha$. In this basis, all the complex matrices that appear in the main CP-transformed equation \eqref{DimensionReductionCPDiracEquation2} are in fact real matrices. This follows from \eqref{InnerProductConjLinearAutomorphisms} and the fact that $j^K_{+} $ commutes with $\sD^K$, with Clifford multiplication, and with both transformations $\rho_{e_a}$ and $\tau_{e_a}$. However, the map $j^K_{+} $ does not preserve the subspaces of spinors in which $\rho$ acts irreducibly through complex or quaternionic representations. Instead, it swaps those spaces in pairs. Such a Majorana basis $\{ \psi_\alpha \}$ will not be a representation basis in the usual sense. Some of the $\psi_\alpha$ will be combinations of spinors belonging to different (complex-conjugated) irreducible representations.
%\end{remark}

\section{Restoring CP symmetry}
\label{RestoringCPSymmetry}

\subsection{Unwinding the CP transformations on $M\times K$}

Let $\{ \psi^\alpha\}$ be a $L^2$-orthonormal basis of the space of spinors on $S_{g_\KKK}$. Let $\Psi$ be a higher-dimensional spinor decomposed as
\[
\Psi  =   \sum_\alpha \big[ \exp{(- i \pi \gamma_5 / 4)}  \, \varphi_\alpha (x) \big]^\HH   \otimes \psi^\alpha(y) \ ,
\]
as in \eqref{Redefinition4DSpinors} and section \ref{DimensionalReductionViolation}. If this spinor satisfies $\sD^P \Psi = 0$, then its 4D components $\varphi_\alpha$ satisfy the set of dimensionally reduced equations \eqref{DimensionReductionDiracEquation2}.  

A higher-dimensional conjugation map can be decomposed as $j_\sigma = j^M_\sigma \otimes j^K_{-\sigma}$, as in \eqref{ConjugationProduct}. It defines the CP-transformation on $M\times K$, denoted $j_\sigma \parity$, which takes the equations of motion \eqref{DimensionReductionDiracEquation2} for the spinors $\varphi_\alpha$ to a different set of equations \eqref{DimensionReductionCPDiracEquation2} for the transformed $\varphi^{cp_\sigma}_\alpha :=   j_\sigma^M \parity_M \, \varphi_\alpha$. The two sets of equations differ through the appearance of complex conjugated coefficients.

The aim of this section is to show that the CP-transformed equations \eqref{DimensionReductionCPDiracEquation2} are, in fact, unitarily equivalent to the original equations \eqref{DimensionReductionDiracEquation2}. In other words, there exists a unitary matrix ${\mathrm U}_\sigma$ such that the linear combinations $\chi_\alpha = \sum_{\beta}({\mathrm U}_\sigma)_{\alpha}^{\beta} \, \varphi^{cp_\sigma}_\beta$ satisfy the same unconjugated equations as the original $\varphi_\alpha$.

The definition of ${\mathrm U}_\sigma$ is simple. Consider the invertible conjugation map $j_{-\sigma}^K$ acting on spinors on $K$. Take a second basis of $L^2 (S_{g_\KKK})$ defined by the image $\{ j_{-\sigma}^K \psi^\alpha \}$. It is $L^2$-orthonormal as well, since \eqref{InnerProductConjLinearAutomorphisms} implies that 
\[
\blangle j_{-\sigma}^K \psi^\alpha, \,  j_{-\sigma}^K \psi^\beta \brangle_{L^2}\: = \: \overline{\blangle \psi^\alpha, \, \psi^\beta \brangle}_{L^2} \: = \: \overline{\delta^{\alpha \beta}} \: =\: \delta^{\alpha \beta} \ .
\]
So there is a unitary transformation of spinors relating the two orthonormal bases. We call this transformation ${\mathrm U}_\sigma$. Observe that ${\mathrm U}_\sigma$ is not the same map as $j_{-\sigma}^K$, as one is $\CC$-linear and the other is conjugate linear. Instead, it is the $\CC$-linear extension to $L^2(S_{g_\KKK})$ of the correspondence $\psi^\alpha \mapsto j_{-\sigma}^K  \psi^\alpha$. Unlike the conjugation $ j_{-\sigma}^K$, the unitary map ${\mathrm U}_\sigma$ depends on the choice of initial basis $\{ \psi^\alpha\}$.

Although $U_\sigma$ is represented by an infinite-dimensional matrix in a general basis $\{ \psi^\alpha\}$, it is either block diagonal or block anti-diagonal in a basis of eigenspinors of the Dirac operator $\sD^K$. In fact, when $\sigma = -1$, we know that $j_{-\sigma}^K$ commutes with $\sD^K$ and preserves its finite-dimensional eigenspaces. Hence so does ${\mathrm U}_\sigma$. When $\sigma = 1$, the conjugation $j_{-\sigma}^K$ anticommutes with $\sD^K$ and maps isomorphically the eigenspaces of $\sD^K$ with opposite eigenvalues. Hence so does ${\mathrm U}_\sigma$. 

The main result of this section is the following:
\begin{proposition}
\label{UnitaryRotation4DSpinors} Let $\{ \varphi^{cp_\sigma}_\alpha \}$ be a set of 4D spinors satisfying the CP-transformed equations \eqref{DimensionReductionCPDiracEquation2}. Then the unitary rotated spinors defined by
\[
 \chi_\alpha\ := \  \sum_\beta \blangle \psi^\alpha, {\mathrm U}_\sigma \, \psi^\beta \brangle_{L^2} \ \varphi^{cp_\sigma}_\beta
 \]
satisfy the same equations \eqref{DimensionReductionDiracEquation2} as the original $\varphi_\alpha$. The only difference is the appearance of the parity-transformed gauge form $A^p = \parity^\ast A$ instead of the original $A$.
\end{proposition}
 The conclusion is that the CP-related equations \eqref{DimensionReductionDiracEquation2} and \eqref{DimensionReductionCPDiracEquation2}  are always unitarily equivalent. To prove this proposition, we start with the following point.
\begin{lemma}
Let ${\mathcal O}$ be an operator on the space  $L^2 (S_{g_\KKK})$ such that $j_{-\sigma}^K {\mathcal O} = \epsilon(\sigma) \, {\mathcal O} j_{-\sigma}^K$, where $ \epsilon(\sigma)$ is a sign. Then there is an identity of 4D spinors
\begin{equation}
\sum_{\alpha, \beta} \blangle  \psi^{\alpha'}, \, {\mathrm U}_\sigma  \psi^{\alpha} \brangle_{L^2} \,  \overline{\blangle \psi^\alpha, \, {\mathcal O}  \psi^\beta \brangle}_{L^2} \ \varphi^{cp_\sigma}_\beta   \ = \  \epsilon(\sigma) \sum_{\beta}   \blangle \psi^{\alpha'}, \, {\mathcal O}   \psi^{\beta} \brangle_{L^2} \  \chi_{\beta} \ .
\end{equation} 
\end{lemma}
\begin{proof}
Using property \eqref{InnerProductConjLinearAutomorphisms} of the conjugation maps and the definition of ${\mathrm U}_\sigma$,
\begin{equation}
\overline{\blangle \psi^\alpha, \, {\mathcal O} \psi^\beta \brangle}_{L^2}   \: = \:  \blangle j_{-\sigma}^K \psi^\alpha, \, j_{-\sigma}^K {\mathcal O}  \psi^\beta \brangle_{L^2}\: = \: \epsilon\,  \blangle j_{-\sigma}^K \psi^\alpha, \, {\mathcal O}  j_{-\sigma}^K \psi^\beta \brangle_{L^2} \: = \:  \epsilon \,  \blangle {\mathrm U}_\sigma \psi^\alpha, \, {\mathcal O} {\mathrm U}_\sigma \psi^\beta \brangle_{L^2}   \ . \nonumber
\end{equation} 
Using the completeness of the bases $\{ \psi^\alpha \}$ and $\{ {\mathrm U}_\sigma \psi^\alpha \}$ and the definition of  $\chi_{\beta}$,
\begin{align}
\sum_{\alpha, \beta} \blangle  \psi^{\alpha'}, \, {\mathrm U}_\sigma  \psi^{\alpha} \brangle_{L^2} \,  \overline{\blangle \psi^\alpha, \, {\mathcal O} \psi^\beta \brangle}_{L^2} \ \varphi^{cp_\sigma}_\beta   \: &= \:   \epsilon \, \sum_{\beta',  \beta} \blangle \psi^{\alpha'}, \, {\mathcal O}  \psi^{\beta'} \brangle_{L^2} \,  \blangle \psi^{\beta'}, \,  {\mathrm U}_\sigma \psi^\beta \brangle_{L^2}  \ \varphi^{cp_\sigma}_\beta  \nonumber \linebr 
&= \:  \epsilon \, \sum_{\beta' }   \blangle \psi^{\alpha'}, \, {\mathcal O}  \psi^{\beta'} \brangle_{L^2} \  \chi_{\beta'}  \nonumber \ .   \qedhere
\end{align}
\end{proof}
\begin{proof}[Proof of proposition \ref{UnitaryRotation4DSpinors}]
Observe that the sign $\epsilon(\sigma)$ is equal to $+1$ for the operators $\rho_{e_a}$ and $\tau_{e_a}$, and is equal to $-\sigma$ for the operators $\sD^K$ and $e_a \cdot$. This follows from \eqref{CommutatorConjKLDerivative} and \eqref{CommutatorConjDiracOperator}, the basic rules  \eqref{CommutatorConjCliffordMultiplication}, \eqref{AlphaCommutesConjugations} and the definition of those operators.
Thus, making the linear combination with constant coefficients $\sum_{\alpha} \blangle  \psi^{\alpha'}, \, {\mathrm U}_\sigma  \psi^{\alpha} \brangle_{L^2}$ of both sides of the CP-transformed equation \eqref{DimensionReductionCPDiracEquation2}, we obtain that the new 4D spinors $\chi_{\alpha}$ satisfy 
\begin{multline}
 i \, \gamma^\mu \,  \big\{  \nabla^{M}_{X_\mu} \, \chi_{\alpha} \, + \,  (A^p)^a_\mu \,  \,  \blangle \, \psi_\alpha  \, , ( \rho_{e_a} +  \tau_{e_a} ) \psi^\beta \, \brangle_{L^2} \,  \,  \chi_\beta  \big\}  \,  
+ \, \sigma^2 \,  \blangle \, \psi_\alpha  \, , \,  \sD^K  \psi^\beta \, \brangle_{L^2}\, \,  \chi_\beta  \; \linebr +\; \frac{\sigma^2}{8}\, \,  (F_{A^p}^a)_{\mu \nu } \, \blangle \psi_\alpha \, , e_a \cdot \psi^\beta \brangle_{L^2}\, \,  \gamma^\mu  \gamma^\nu \,  \chi_\beta  \ = \ 0 \ . \nonumber
\end{multline} 
Since $\sigma^2 = 1$, we conclude that the rotated 4D spinors $\chi_\alpha$ satisfy exactly the original set of dimensionally reduced equations \eqref{DimensionReductionDiracEquation2}, as written prior to the CP-transformation.
\end{proof}

 \subsection{\bf Restoring the physical CP symmetry in 4D} 

The usual convention in physics is to choose as the 4D conjugation map what we call here $j_-^M$. Then the CP transformations on $M_4$ are implemented by the composition $j_-^M \parity_M$. The choice $j_-^M$ preserves the sign of the mass coefficient in the standard 4D Dirac equation, whereas choosing $j_+^M$ would flip that sign. This was noted in remark \ref{PhysicalHDConjugations}.

This rigidity leads us to a small problem, though. Proposition \ref{FactorizationConjugations} says that  $j_-^M$ on Minkowski space can be lifted only to a conjugation on $M_4 \times K$ of the form $j_- = j^M_- \otimes j^K_{+}$. And while a $j^K_{+}$ exists when the dimension of $K$ is even or is 3 (mod 4), it does not exist when $K$ has dimension 1 (mod 4). Thus, in the latter case, the previous discussion about higher-dimensional CP-transformations concerns only $j_+ = j^M_+ \otimes j^K_{-}$. It never dimensionally-reduces to the physical conjugation $j_-^M$ on $M_4$ and never deals with the physical CP-transformed 4D spinors $\varphi^{cp_-}_\alpha = j_-^M \parity_M \varphi_\alpha$. 

Two natural questions can now be raised. What is the $ j_-^M \parity_M$-transformation of the equations of motion \eqref{DimensionReductionDiracEquation2} when $\dim K = 1 \ ({\rm mod \ 4})$? And does the unitary transformation ${\mathrm U}$ that unwinds CP-transformations still exist in this case? Observe that our previous definition of ${\mathrm U}_\sigma$ used the conjugation $j_{- \sigma}^K$, which does not exist in the problematic dimensions. So the doubt is pertinent.

Regarding the first question, a direct application of the general properties of $j_-^M$ and $\parity_M$ to equation \eqref{DimensionReductionDiracEquation2} yields:
\begin{multline}
\label{DimensionReductionCPDiracEquation3}
 i \, \gamma^\mu \,  \big\{  \nabla^{M}_{X_\mu} \,  \varphi^{cp_-}_\alpha \, + \,  (A^p)^a_\mu \,  \, \conj{ \blangle \, \psi_\alpha  \, , ( \rho_{e_a} +  \tau_{e_a} ) \psi^\beta \, \brangle}_{L^2} \,  \,  \varphi^{cp_-}_\beta  \big\}  \,  
+ \,  \conj{\blangle \, \psi_\alpha  \, , \,  \sD^K  \psi^\beta \, \brangle}_{L^2}\, \,  \varphi^{cp_-}_\beta  \; \linebr +\; \frac{1}{8}\, \,  (F_{A^p}^a)_{\mu \nu } \, \conj{\blangle \psi_\alpha \, , e_a \cdot \psi^\beta \brangle}_{L^2}\, \,  \gamma^\mu  \gamma^\nu \,  \varphi^{cp_-}_\beta  \ = \ 0 \ 
\end{multline}
for the physically CP-transformed spinors $\varphi^{cp_-}_\alpha = j_-^M \parity_M \varphi_\alpha$. So nothing new here.

We now deal with the second question. When $\dim K \neq 1$ (mod 4), the lifted conjugation $j_- = j_-^M \otimes j_+^K$ exists and the unitary map ${\mathrm U}_-$ constructed before restores the physical CP symmetry of the 4D equations of motion. This follows from proposition \ref{UnitaryRotation4DSpinors}. When $\dim K = 1$ (mod 4), only the conjugation $j_+^M \otimes j_-^K$ exists, so we only get the unitary ${\mathrm U}_+$. But that unitary restores the $ j_+^M \parity_M$-symmetry in 4D, not the physical $ j_-^M \parity_M$-symmetry that we want. This inconvenience can be circumvented by slightly tweaking the definition of the rotated 4D spinors $\chi_\alpha$. Using
\begin{equation}
\label{SpinorRotationU}
 \chi_\alpha\ := \ 
 \begin{cases}
  \sum_\beta \blangle \psi^\alpha, {\mathrm U}_+ \, \psi^\beta \brangle_{L^2} \ \gamma_5 \, \varphi^{cp_-}_\beta &{\text{if $\dim K = 1$ (mod 4) }} \linebr
    \sum_\beta \blangle \psi^\alpha, {\mathrm U}_- \, \psi^\beta \brangle_{L^2} \ \varphi^{cp_-}_\beta &{\text{otherwise}} \ ,  
 \end{cases}
\end{equation}
one can verify that everything works out and the $ \chi_\alpha$ satisfy the original equations of motion \eqref{DimensionReductionDiracEquation2} in any dimension. 

Unwinding the previous definitions, it is clear that the unitary ${\mathrm U}$-rotations $\varphi^{cp_-}_\alpha \rightarrow \chi_\alpha$ in \eqref{SpinorRotationU} preserve the mass of 4D spinors for all dimensions of $K$, when $\{\psi^\alpha \}$ is a basis of $\sD^K$-eigenspinors. Moreover, it follows from \eqref{RepresentationConjugation} that those internal rotations will take spinors in a $\pi$-representation space to spinors in a $\bar{\pi}$-representation.

\begin{remark}
We have shown here that, when $K$ is compact and $g_K$ is constant, the dimensional reduction of the Dirac equation $\sD \Psi = 0$ on $M_4 \times K$ leads to 4D equations that are unitarily equivalent to their CP-transformed counterparts. Even when arbitrary 4D massive gauge fields are present. We have also noted that, if $\dim K \neq 2  \pmod{4}$, the higher-dimensional CP transformations $j_\sigma \parity$ commute with the chiral operator and preserve the Weyl condition $\Gamma_P \Psi = \Psi$. Moreover, the internal conjugation maps $j_\sigma^K$ also commute with the chiral operator $\Gamma_K$ in these dimensions. So the previous unitaries ${\mathrm U}_\sigma$ that restore CP symmetry can be made block-diagonal with respect to $\Gamma_K$, and all the arguments in the Dirac case go through to the Weyl case. Thus, when $\dim K \neq 2  \pmod{4}$, the dimensional reduction of the Weyl equation is also unitarily equivalent to its CP-transformed counterpart in 4D, even if massive gauge fields are present.
\end{remark}

\begin{remark}
\label{CAndPInvariance}
A simpler version of the arguments in this section shows that the Dirac equations \eqref{DimensionReductionDiracEquation2} also preserve C and P separately. More precisely, when $K$ is compact and $g_K$ is constant, equations \eqref{DimensionReductionDiracEquation2} are unitarily equivalent both to their C-transformed and to their P-transformed counterparts. This does not carry through to the Weyl equations, though, since the chirality operator $\Gamma_P$ can fail to commute with $j_\sigma$ and $\parity$. However, that problem does not arise for $j_\sigma$ when $\dim K =  2 \pmod{4}$, so one can show that the dimensionally reduced Weyl equations are C-invariant in that case. 
\end{remark}

\begin{remark}
\label{ExampleCPViolation}
We have established CP-invariance of the dimensionally reduced Dirac equation in all dimensions and of the Weyl equation when $\dim K \neq 2  \pmod{4}$. In the case $\dim K = 2  \pmod{4}$, the previous remark says that the dimensionally reduced Weyl equation is C-invariant. Hence it will be CP-invariant iff it is P-invariant. In other words, intrinsic chirality and CP violation go hand-in-hand when $\dim K = 2  \pmod{4}$. But \cite{Bap3}, section 6, describes an explicit example in which the dimensional reduction of the Weyl equation on $M_4 \times T^2$ produces chiral interactions of 4D fermions and massive gauge fields. And that chirality cannot be removed through unitary re-definitions of the fermions preserving their masses and abelian charges. Thus, having massive gauge fields with $K = T^2$ should yield a first example of CP-violation in pure-gravity KK models with submersion backgrounds. A deeper understanding of that example would require a dedicated analysis.
\end{remark}

\section{Conclusions}

This paper investigates the properties of pure-gravity Kaluza-Klein models on spacetimes $M\times K$ equipped with submersion metrics. These higher-dimensional metrics generalize the Kaluza ansatz in two ways. They allow the internal geometry to vary arbitrarily along $M$, and allow the off-diagonal components of the metric to encode 4D gauge fields linked to arbitrary vector fields on $K$ --- both Killing and non-Killing. In contrast, the standard Kaluza-Klein metrics only encode gauge fields linked to Killing fields on $K$. In previous work \cite{Bap1, Bap3}, we showed that gauge fields associated with non-Killing vector fields on $K$ through a submersion metric exhibit the following distinctive properties:
\begin{enumerate}[itemsep=1mm, topsep=7pt]
\item[\bf 1.] They gain a non-zero mass \eqref{MassFormula} through a geometric Higgs-like mechanism, in which the internal metric $g_K$ plays the role of the traditional Higgs field.
\item[\bf 2.] Their gauge interactions can mix fermions with different masses.
\item[\bf 3.] They can have asymmetric (chiral) couplings to left- and right-handed fermions.
\end{enumerate} 
The last property circumvents well-known no-go arguments against the existence of chiral fermions in KK models \cite{CS, Wetterich83a, Wetterich83, Witten83}. Those arguments only apply to massless 4D gauge fields, linked to vector fields that are exactly Killing on $K$.

In the present paper, we broaden the investigation of the dimensional reduction of the Dirac equation $\sD \Psi = 0$ on $M\times K$ equipped with submersion metrics. We establish that: 
\begin{enumerate}[itemsep=1mm, topsep=7pt]
\item[\bf 4.] The 4D massive gauge fields encoded in the metric naturally generate a complex, CKM-like matrix, measuring the misalignment between the mass eigenspinors and the spinor representation bases. 
\item[\bf 5.] Nevertheless, the dimensional reduction of the plain Dirac equation does not produce CP violation in 4D, for any compact $K$, in regions where $g_K$ is constant. The same holds for the Weyl equation when $\dim K \! \neq \! 2  \pmod{4}$, but should not hold in general.
\item[\bf 6.] CP transformations commute with dimensional reduction, in an appropriate sense.
\item[\bf 7.] The 4D gauge representations on spinors determined by the higher-dimensional Dirac operator are self-conjugate when $\dim K \neq 1 \pmod{4}$, hence anomaly-free.
\end{enumerate} 
Point 5 and remark \ref{ExampleCPViolation} suggest the need to investigate the Weyl equation on general submersion backgrounds when $\dim K = 2  \pmod{4}$. One should also consider alternative starting points in higher dimensions, beyond the plain Dirac equation.

In a more speculative vein, the paper touches on a possible geometric mechanism for the emergence of fermion generations in these KK models. It is tied to the splitting of the degenerate eigenvalues $\sD^K$ that occurs when the isometry group of $g_K$ is broken to a smaller subgroup. 

Besides the physically motivated results, the paper also develops certain geometric constructions that may be of independent interest. Specifically:
\begin{enumerate}[itemsep=1mm, topsep=7pt]
\item[\bf a.] It considers the extensions of reflections and parity inversion on $M$ to diffeomorphisms of $M \times K$. It describes their non-trivial actions on the submersion metric, the higher-dimensional spinors, and the Dirac operator.
\item[\bf b.] It gives a short, geometric account of the conjugation automorphisms $j_\pm$ of spinor bundles on manifolds of arbitrary signature $(s,t)$. It describes their actions on the Dirac operator and the Kosmann-Lichnerowicz derivative.
\item[\bf c.] It introduces a new Lie derivative of spinors along non-Killing vector fields induced by actions of compact groups. This derivative reduces to the Kosmann-Lichnerowicz one for conformal Killing vector fields, but is otherwise distinct.
\end{enumerate} 

The analysis presented in this paper also has clear limitations. It is purely classical and establishes only structural properties of the KK models. It offers proofs of concept but does not delve into specific examples with phenomenological interest. Thus, from a physical viewpoint, it is still far from describing realistic particle models or producing numerical predictions. In this regard, it would be interesting to choose particular examples of internal geometries and work out explicitly the 4D footprint of the related KK models. One such example could be the geometry suggested in appendix \ref{Example}, where $K$ is the group manifold $\SU (3)$ equipped with a natural left-invariant metric having isometry group $(\SU(3) \times \SU(2) \times \mathrm{U}(1))/\mathbb{Z}_6$. Studying explicit examples will also provide a deeper understanding of the various formal properties established here.

Another limitation of the present analysis is that it does not address the dynamical origin of the vacuum metric on $K$. The discussion assumes the existence and stability of some $g_K$, and then proceeds to relate the properties of 4D fermions and gauge fields to the properties of that metric. Finding a natural mechanism that dynamically generates realistic and stable vacuum metrics --- possibly an appropriate correction to the higher-dimensional Einstein-Hilbert action --- remains an important unresolved challenge in pure-gravity Kaluza-Klein models.

\newpage

\addtocontents{toc}{\cftpagenumbersoff{section}} 

\appendix
\appendixpage
\noappendicestocpagenum
\addappheadtotoc

\addtocontents{toc}{\cftpagenumberson{section}}

\section{Conventions for spinors in signature $(s,t)$}
\label{ConventionsSpinors}

\subsection{Gamma matrices, inner products and chirality operator}

Consider the vector space $\mathbb{R}^{n}$ equipped with the standard, non-degenerate metric $\eta$ of signature $(s,t)$. It has $s$ independent vectors of positive norm. Our convention is that these vectors represent space-like directions. So the internal space $K$ of a Kaluza-Klein model has a Riemannian metric with positive signature.

The spinor space $\Delta_n$ is the complex coordinate space of dimension $2^{\lfloor \frac{n}{2} \rfloor}$. For each choice of signature $(s,t)$, there are $n$ matrices $\gamma_j$ that act on $\Delta_n$ and satisfy the relations
\beq
\label{ConventionsGammaMatrices1}
 \gamma_j \,  \gamma_k   +   \gamma_k \,  \gamma_j  \ = \ - 2 \, \eta_{j k} \ I \ .
\eeq
These matrices are not unique and can be rotated with $\mathrm{O}_{s,t}$ similarity transformations. They can be chosen to have the hermiticity properties
\begin{singlespace}
\vspace{-0.2cm}
\beq
\label{ConventionsGammaMatrices2}
 (\gamma_k)^\dag \ = \ 
 \begin{cases} 
 \ \gamma_k     &  \text{if $ \ 1 \leq k \leq t $} \linebr
 \ - \, \gamma_k &  \text{if $ \ t  < k \leq t + s $}  \ .
 \end{cases}
\eeq
\end{singlespace}
\noindent
Thus, the gamma matrices in spatial dimensions are square roots of $-1$ and are anti-self-adjoint with respect to the product $\psi ^\dag  \psi$ of spinors in $\Delta_n$. This is the most common convention in Riemannian geometry \cite{ABS, LM, Bourguignon, Ginoux, Friedrich}. It is also used in the QFT textbook \cite{Srednicki}.\footnote{The dominant convention in physics textbooks is to write the Clifford relation \eqref{ConventionsGammaMatrices1} with the opposite sign. So our gamma matrices in signature $(s,t)$ are their gamma matrices in signature $(t,s)$.} Vectors in $\mathbb{R}^{n}$ act on spinors through the Clifford multiplication $v\cdot \psi := \sum_j v^j \gamma_j \, \psi$. This determines an irreducible representation of the Clifford algebra $\mathrm{Cl}_{s,t}$ on the spinor space. So $\Delta_n$ also carries representations of the groups $\mathrm{Pin}_{s,t}$ and $\Spin_{s,t}$, which are contained inside the Clifford algebra. 

When $t \geq 1$, the positive-definite Hermitian product on $\Delta_n$ is not invariant under the action of the connected component of the identity of the spin group, denoted $\Spin^0_{s,t}$. So one usually works with a different, $\Spin^0_{s,t}$-invariant inner product. It is defined by
\beq
\label{GeneralInnerProduct}
\langle \psi_1, \, \psi_2 \rangle \ := \ 
\begin{cases}
\  \psi_1^\dag \, \psi_2   & \text{if $t=0$} \linebr
\  i^{\lfloor \frac{t}{2} \rfloor} \; \psi_1^\dag \,  \gamma_1 \cdots \gamma_t \, \psi_2  &  \text{if $t \geq 1$} \ .
\end{cases}
\eeq
Here $\lfloor a \rfloor$ denotes the integral part of $a$. This product satisfies $\conj{\langle \psi_1, \, \psi_2 \rangle}  = \langle \psi_2, \, \psi_1 \rangle $ and, for $t \geq 1$, has indefinite signature 
 $(2^{ \lfloor \frac{n}{2} \rfloor -1},  \, 2^{ \lfloor \frac{n}{2} \rfloor -1})$. Under Clifford multiplication, 
\beq
\langle v\cdot \psi_1,  \, \psi_2 \rangle \ = \ (-1)^{t-1} \; \langle \psi_1,  \,  v\cdot \psi_2 \rangle  \ .
\eeq
The complex chirality operator $\Gamma : \Delta_n \rightarrow \Delta_n$ in signature $(s,t)$ is defined by 
\beq
\label{DefinitionChiralOperator}
\Gamma \, \psi \ := \  i^{\lfloor \frac{s-t +1}{2} \rfloor } \ \gamma_1 \cdots \gamma_{n} \, \psi \ .
\eeq
It satisfies the algebraic relations
\bal
\Gamma \,  \Gamma \, \psi \ &=\  \psi   \linebr
\Gamma(v\cdot \psi) \ &= \ (-1)^{n-1} \,  v\cdot (\Gamma \, \psi)   \label{ChiralOperatorCliffordMultiplication}  \linebr
\langle\Gamma\,  \psi_1,  \, \psi_2\rangle \ &= \  (-1)^{t(n-1)}\, \langle \psi_1,   \, \Gamma\,  \psi_2\rangle \ .
\end{align}
In particular, $\Gamma$ commutes with the elements in the Lie algebra $\spin_{s,t} \simeq \mathrm{Cl}^{(2)}_{s,t}$. So it commutes with the action of the connected group $\Spin^0_{s,t}$ on spinors.

\subsection{Majorana forms and conjugation automorphisms}
\label{AppendixMajoranaConjugation}

Consider again the vector space $\mathbb{R}^{n}$ equipped with the standard metric of signature $(s,t)$ and let $\Delta_n = \CC^{2^{\lfloor\frac{n}{2} \rfloor }}$ be the complex spinor space. Define the set
\begin{equation}
I_{n} \ := \ 
\begin{cases}
\{-1, 1 \}  &  \text{if $n$ is even} \\
\{ (-1)^{\frac{n-1}{2} } \}  & \text{if $n$ is odd} \ .
\end{cases}
\end{equation}

\begin{proposition}
\label{thm: MajoranaForms}
For each value $\tau \in I_{n}$, there exists a $\CC$-bilinear form on $\Delta_n$ such that
\begin{align}
C_\tau ( v \cdot \psi_1, \, \psi_2) \ &= \  \tau \;   C_\tau ( \psi_1, \, v \cdot  \psi_2)   \label{CliffMultiplicationMajoranaForms} \linebr
C_\tau ( \psi_1, \, \psi_2) \ &= \  \varepsilon(n, \tau)  \;   C_\tau ( \psi_2, \,  \psi_1)    \label{SymmetryMajoranaForms}
\end{align}
for all spinors $\psi_j \in \Delta_n$ and vectors $v \in \mathbb{R}^n$. The symmetry sign is 
\beq
\varepsilon(n, \tau) \ := \ 
\begin{cases}
\; (-1)^{ \lfloor \frac{n}{4} \rfloor} \ \tau^{\frac{n}{2}} &  \text{if $n$ is even} \\
\; (-1)^{\frac{n^2 -1}{8}} & \text{if $n$ is odd} \ .
\end{cases}
\eeq
Property \eqref{CliffMultiplicationMajoranaForms} is valid for gamma matrices of any signature with $s+t = n$. The form $C_\tau$ can be chosen so that it is represented by a unitary matrix in any complex basis of $\Delta_n$ orthonormal with respect to the product $\psi^\dag \, \psi$.
\end{proposition}
The bilinear forms $C_\tau$ are called Majorana forms on $\Delta_n$. They only depend on the dimension $n$, not on the signature. The properties in proposition \ref{thm: MajoranaForms} determine them uniquely, up to a complex phase. Property \eqref{CliffMultiplicationMajoranaForms} implies that, in all signatures:
\begin{proposition}
\label{thm: PropertiesMajoranaForms}
The Majorana forms satisfy:
\begin{align}
C_\tau (\, \Gamma \,  \psi_1, \, \psi_2 \, ) \ &= \  (-1)^{(n-1) \lfloor\frac{n}{2} \rfloor } \  C_\tau (\, \psi_1, \,  \Gamma \,  \psi_2 \, )   \linebr
C_\tau ( b \cdot \psi_1, \, \psi_2) \ &= \  - \; C_\tau ( \psi_1, \, b \cdot  \psi_2)   \linebr
C_\tau ( \omega \cdot \psi_1, \,  \omega \cdot \psi_2) \ &= \  C_\tau ( \psi_1, \,   \psi_2)  \ .   
\end{align}
Here $\Gamma$ is the chirality operator on $\Delta_n$, $b$ is any element in the Lie algebra $\spin_{s,t}$, and $\omega$ is any element in the connected group $\Spin^0_{s,t}$.
\end{proposition}
The Majorana forms in proposition \ref{thm: MajoranaForms} can be explicitly constructed using a set of Euclidean gamma matrices in a special representation, in which they are all symmetric or anti-symmetric. See the constructions in \cite{Hitoshi, Stone}. The content of proposition \ref{thm: MajoranaForms}, which we will not prove here, encodes the essential properties of Majorana forms in a concise and representation-independent way.

Using the non-degenerate Majorana forms and the sesquilinear inner product \eqref{GeneralInnerProduct}, one can define natural, conjugate-linear automorphisms of $\Delta_n$.
Consider the sets
\begin{equation}
H_{s,t} \ := \ 
\begin{cases}
\{-1, 1 \}  &  \text{if $s-t$ is even} \\
\{ (-1)^{\frac{s-t+1}{2}} \}  & \text{if $s-t$ is odd} \ .
\end{cases}
\end{equation}

\begin{definition}
\label{DefinitionConjugations}
For each value $\sigma \! \in \! H_{s,t}$, the conjugation map $j_\sigma \!: \Delta_n \rightarrow \Delta_n$ is defined by
\beq
\langle \,  j_\sigma \,  \psi_1, \, \psi_2 \, \rangle \ := \ C_{\tau (\sigma)} (\psi_1, \, \psi_2) \ 
\eeq
for all $\psi_j \in \Delta_n$. Here $\tau (\sigma) =  (-1)^{t+1} \sigma$ and the inner product $\langle \, \cdot , \cdot \, \rangle$ is the one in \eqref{GeneralInnerProduct}.
\end{definition}
Although the Majorana form $C_\tau$ does not depend on the signature of the gamma matrices, the inner product $\langle \, \cdot , \cdot \, \rangle$ does. So the properties of the conjugation maps $j_\sigma$ do depend on the signature.  
\begin{proposition}
\label{thm: ConjLinearMaps}
The maps $j_\sigma$ are conjugate-linear and satisfy:
\begin{align}
j_\sigma (v \cdot \psi) \ &=  \  \sigma \ v \cdot    j_\sigma (\psi)   \label{ConjugationCliffordMultiplication}   \linebr
j_\sigma \,  j_\sigma (\psi) \ &= \ 
\begin{cases}
\; (-1)^{ \lfloor \frac{s-t}{4} \rfloor } \ (-\sigma)^{\frac{s-t}{2}} \, \psi &  \text{if $s-t$ is even} \\
\; (-1)^{\frac{(s-t)^2 -1}{8}}  \, \psi & \text{if $s-t$ is odd}
\end{cases}           \label{SquareConjugationMaps}     \linebr
j_\sigma  ( \Gamma \, \psi) \ &= \ 
\begin{cases}
\; (-1)^{ \frac{s-t}{2}} \ \Gamma \; j_\sigma (\psi) &  \text{if $s-t$ is even} \\
\; \Gamma \; j_\sigma (\psi) & \text{if $s-t$ is odd} \ ,
\end{cases} 
\end{align}
for all spinors $\psi \in \Delta_n$ and all vectors $v \in \mathbb{R}^n$. When $s-t$ is even, the map $j_-$ coincides with $\Gamma \circ j_+$ up to a complex phase.
\end{proposition}
A conjugation with $j_\sigma  j_\sigma = 1$ is called a real structure on $\Delta_n$. It allows the consistent imposition of the Majorana condition on spinors, $j_\sigma(\psi) = \psi$. A conjugation with $j_\sigma  j_\sigma = -1$ is called a quaternionic structure on $\Delta_n$. It defines an action of the quaternions on spinors through the representation $(\mathbf{i}, \mathbf{j}, \mathbf{k}) \rightarrow (i, j_\sigma, i\, j_\sigma)$.

Using the definition of $j_\sigma$ and properties \eqref{SymmetryMajoranaForms} and \eqref{SquareConjugationMaps}, one also calculates that
\begin{lemma}
\label{thm: MorePropertiesConjLinearMaps}
The maps $j_\sigma$ satisfy:
\begin{align}
\langle j_\sigma \, \psi_1 , \,  j_\sigma \, \psi_2 \rangle \ &= \ 
\begin{cases}
(-1)^{\frac{t(s+1)}{2}}\, \sigma^t \;   \conj{\langle \psi_1 , \, \psi_2 \rangle}   &  \text{if $s-t$ is even} \\ 
(-1)^{\frac{st}{2}}\; \conj{\langle \psi_1 , \, \psi_2 \rangle}   &  \text{if $s-t$ is odd} 
\end{cases} 
\end{align}
for all spinors $\psi_1, \psi_2 \in \Delta_n$ and for the inner product in \eqref{GeneralInnerProduct}.
\end{lemma}

\subsection{Spinor bundles, connections and Dirac operator}
\label{AppendixSpinorBundles}

Let $M$ be an oriented manifold of dimension $n$. A topological spin structure on $M$ is a double cover $\pi: F{{\raisebox{-4pt}{$\scriptstyle \widetilde{\mathrm{GL}}^+$}}} \rightarrow F_{\mathrm{GL^+}}$ of the bundle of oriented $n$-frames on $TM$. This cover should have a right action of $\widetilde{\mathrm{GL}}_n^+$, the universal cover of the connected component of the identity of $\mathrm{GL}_n(\mathbb{R})$. The cover should be equivariant with respect to the right actions of $\widetilde{\mathrm{GL}}_n^+$ and $\mathrm{GL}_n^+$ on the respective bundles. See for example \cite{Bourguignon}.

Let $g$ be a metric on $M$ of signature $(s,t)$. Considering only the oriented, $g$-orthonormal frames on $TM$ defines a subbundle $F_{\mathrm{SO}}(g)$ inside $F_{\mathrm{GL}^+}$. The cover $\pi$ can be restricted to a double cover $F_{\mathrm{Spin}}(g) \rightarrow F_{\mathrm{SO}}(g)$ that is equivariant with respect to the right actions of the connected groups $\Spin_{s,t}^0$ and $\mathrm{SO}^0_{s,t}$ on the respective bundles. We call this restriction a metric spin structure on $M$ subordinate to the topological one. See for example \cite{Baum, Bar}.

The complex spinor bundle $S_g \rightarrow M$ is the associated vector bundle $F_{\mathrm{Spin}}(g) \times_{\Spin^0_{s,t}} \Delta_n$. The sections of $S_g$ are spinor fields $\Psi$. There is a natural action of vector fields on spinor fields by Clifford multiplication on each fibre. This extends to a left action of the Clifford bundle on $S_g$. 
The inner product $\langle \cdot , \cdot \rangle$ and the Majorana forms $C_\tau$ on spinor space $\Delta_n$ are $\Spin^0_{s,t}$-invariant. So they have natural extensions to the spinor bundle $S_g$, denoted by the same symbols. Similarly, the chirality operator $\Gamma$ and the conjugation automorphisms $j_\sigma$ commute with the $\Spin^0_{s,t}$-action on spinor space. This follows from \eqref{ChiralOperatorCliffordMultiplication} and \eqref{ConjugationCliffordMultiplication}, respectively. So they have natural extensions as automorphisms of the spinor bundle.

The Levi-Civita connection $\nabla$ on the tangent bundle $(TM, g)$ has a standard lift to the spinor bundle that is compatible with the inner product $\langle \cdot ,  \cdot \rangle$ and with Clifford multiplication. One can also check that it preserves the chirality operator, the Majorana forms and the conjugation automorphisms. In other words, 
\bal
\nabla_U (V \cdot \Psi) \ &= \  V \cdot ( \nabla_U \Psi)  +  (\nabla_U V) \cdot \Psi \\ 
\Lie_U \langle \Psi_1, \,  \Psi_2 \rangle \ &= \   \langle \nabla_U \Psi_1, \,  \Psi_2 \rangle  +   \langle \Psi_1, \, \nabla_U  \Psi_2 \rangle \\ 
\nabla \Gamma \ &= \ \nabla C_\tau \  = \ \nabla j_\sigma \ = \ 0 \ . 
\end{align}
An explicit, local formula for the covariant derivative of spinors is \cite{Bourguignon}:
\beq
\label{StandardSpinorCovariantDerivative}
\nabla_X \Psi \; := \; \partial_X \tilde{\Psi} \;+\;  \frac{1}{4}\, g^{ir} \, g^{js} \, \, g(\nabla_X v_i, v_j) \, v_r \cdot v_s \cdot  \tilde{\Psi}    \ . 
\eeq
Here $\{v_r\}$ denotes a local, oriented, $g$-orthonormal trivialization of $TM$, while $\tilde{\Psi}$ represents the spinor $\Psi$ in the induced trivialization of $S_g$. So $\tilde{\Psi}$ is a local function on $P$ with values in  $\Delta_n$ and we denote by $\partial_X  \tilde{\Psi} = (\dd \tilde{\Psi})(X)$ its directional derivative. 
Since the metric is pseudo-Riemannian, the elements $g^{rs}$ can be $\pm 1$. We will most often abuse notation and simply write $\Psi$ instead of $\tilde{\Psi}$.

The Dirac operator on $M$ can be defined by the local formula
\beq
\sD \Psi \ := \  g^{jk} \, v_j \cdot \nabla_{v_k} \Psi \ .
\eeq
For spinor fields with compact support, it satisfies
\beq
\int_M \langle \sD \Psi_1, \, \Psi_2 \rangle \; \vol_g \ = \ (-1)^t \int_M \langle \Psi_1, \,  \sD \Psi_2 \rangle \; \vol_g \ ,
\eeq
where $\langle \cdot, \cdot \rangle$ is the inner product \eqref{GeneralInnerProduct}.

\section{Relating spinors of different metrics}
\label{SpinorsDifferentMetrics}

\subsection{Transporting the Levi-Civita connection and the Dirac operator}

The purpose of this appendix is to describe the formulas to transport objects between the spinor bundles determined by two different Riemannian metrics on the same manifold with fixed topological spin structure. These objects are the spinors themselves, their covariant derivatives, the Dirac operator and the Kosmann-Lichnerowicz derivatives. Important references on this topic include \cite{BG, Wang, AHermann}. Here we simplify some of the formulas given by Wang in section 1 of \cite{Wang}. Then we extend the calculations to the case of the transported Kosmann-Lichnerowicz derivative, which does not appear to have been considered before. The notation used here is mostly that of \cite{Wang}. We omit the proofs of the formulas that follow from straightforward, though often not short, calculations. The results in this appendix are used in section \ref{ActionsCompactGroups} to define a new Lie derivative of spinors along the fundamental vector fields of a $G$-action.

Let $(M,g)$ be an oriented, Riemannian manifold with a fixed spin structure. Let $\alpha: TM \rightarrow TM$ be an invertible map that is linear on the fibres and projects down to the identity on $M$. Assume also that $\alpha$ is $g$-symmetric, i.e. that $g (\alpha(U), V) \, = \, g (U, \alpha(V))$ for all tangent vectors $U$ and $V$. Then we can define a new Riemannian metric on $M$ by
\beq
\label{DefinitionNewMetric}
g_\alpha (U, V) \ := \ g ( \alpha^{-1}(U) , \alpha^{-1}(V) ) \ .
\eeq
The Levi-Civita connections of the new and old metrics are denoted $\nabla^{g_\alpha}$ and $\nabla^{g} = \nabla$. Using Koszul's formula, one finds that the transport of $\nabla^{g_\alpha}$ by the automorphism $\alpha$ is:
\beq
\label{ComparisonConnectionsTM}
 ( \, \alpha^{-1 } \circ \nabla_V^{g_\alpha} \circ \alpha ) \, U \; = \    \nabla_V U \;+ \; \frac{1}{2} \,  [ \, \theta^\alpha_V \, - \,  (\theta_V^\alpha)^* \, ] \, U    \ .
\eeq
Here $\theta^\alpha_V: TM \rightarrow TM$ is the endomorphism given by
\beq
\theta^\alpha_V : \, U \ \longmapsto \ \alpha \circ (\nabla_{\alpha (U)} \alpha^{-2})(V) \ + \  \alpha \circ (\nabla_V \alpha^{-1}) (U) 
\eeq
and $(\theta_V^\alpha)^*$ denotes its adjoint with respect to the metric $g$. Since $\nabla$ is a metric connection on $TM$, it follows from \eqref{ComparisonConnectionsTM} that  $\alpha^{-1 } \circ \nabla^{g_\alpha} \circ \alpha$ is a metric connection too.

The automorphism $\alpha$ of $TM$ can be lifted to a map $\alpha: S_g \rightarrow S_{g_\alpha}$ between different spinor bundles, which we denote by the same symbol. The lift is equivariant with respect to Clifford multiplication:
\beq
\alpha ( V \cdot \psi) \ = \  \alpha (V) \cdot \alpha(\psi)
\eeq
for any vector $V \in TM$ and spinor $\psi \in S_g$ (see \cite[lem. 2.3.4]{AHermann}). Let $\langle \cdot , \cdot \rangle$ be the inner product of spinors defined in \eqref{GeneralInnerProduct}. It is positive-definite in Euclidean signature. Let $\Gamma_M$ be the chirality operator on spinor bundles determined by \eqref{DefinitionChiralOperator}. Let $j_\sigma$ be any of the conjugation automorphisms in definition \ref{DefinitionConjugations}, after extension to the spinor bundles. We use the same symbols to denote the inner products, chirality operators and conjugations on $S_g$ and $S_{g_\alpha}$. Then one can check that the lift $\alpha : S_g \rightarrow S_{g_\alpha}$ has the natural properties:
\begin{align}
\langle \alpha(\psi_1), \, \alpha(\psi_2) \rangle \ &= \ \langle \psi_1, \, \psi_2 \rangle   \label{AlphaIsometry} \linebr
\alpha (\Gamma_M \, \psi ) \ &= \  \Gamma_M \, \alpha (\psi)      \label{AlphaCommutesChiral} \linebr
\alpha \circ j_\sigma (\psi) \ &= \  j_\sigma \circ \alpha (\psi)   \label{AlphaCommutesConjugations} \ .
\end{align}
The two Levi-Civita connections on $TM$ admit standard lifts to the respective spinor bundles, which we still denote $\nabla^{g_\alpha}$ and $\nabla$. These are connections on different bundles, $S_{g_\alpha}$ and $S_g$. To compare them we need to transport $\nabla^{g_\alpha}$ back to $S_g$, so to consider instead the covariant derivative $\alpha^{-1 } \circ \nabla_V^{g_\alpha} \circ \alpha$. Using \eqref{StandardSpinorCovariantDerivative} and \eqref{ComparisonConnectionsTM}, one calculates that
\beq
\label{ComparisonConnectionsSM}
\alpha^{-1 } \circ \nabla_V^{g_\alpha} \circ \alpha \, (\psi) \; = \; \nabla_V \psi \ + \ \frac{1}{4}\, \sum_{i \neq j} \; g (  \theta^\alpha_V (v_i), \: v_j  )\; v_i \cdot v_j \cdot \psi  \ , 
\eeq
where $\{ v_k\}$ denotes any local, oriented, $g$-orthonormal trivialization of $TM$. This is a simplification of the formula given in \cite{Wang}. One can also calculate how the Dirac operator on $S_{g_\alpha}$, denoted $ \sD^{g_\alpha}$, appears when transported to $S_g$. The result is:
\begin{align}
\label{TransportedDiracOperator}
\big(\alpha^{-1 } \circ \sD^{g_\alpha} \circ \alpha \big) \psi \ = \ &\sum\nolimits_k  \alpha(v_k) \cdot \nabla_{v_k} \psi \ - \ \frac{1}{2} \big\{ \, \delta \alpha \: + \:  \alpha \big[  \, \grad (\log \det \alpha) \, \big]  \, \big\} \cdot \psi  \linebr
&+ \; \frac{1}{4}\, \sum_{i \neq j \neq k} \, g \big( \alpha^{-1} \circ  (\nabla_{\alpha (v_i)} \alpha)(v_j) , \; v_k \big) \; v_i \cdot v_j \cdot v_k \cdot \psi  \nonumber \ .
\end{align}
In this formula, $\det \alpha$ is the function on $M$ relating volume forms, $\vol_{g_\alpha} = (\det \alpha) \, \vol_{g}$. The vector field $\delta \alpha$ on $M$ is defined by $\delta \alpha =  -  \sum\nolimits_k \,(\nabla_{v_k} \alpha)( v_k)$.
The sum over $i \neq j \neq k$ means that all three indices should be different. Formula \eqref{TransportedDiracOperator} is a slightly more developed version of the one given in \cite{Wang}. It isolates the contributions of the volume change $\det \alpha$ and the divergence $\delta \alpha$ to the new operator. Using \eqref{TransportedDiracOperator}, one can then calculate that
\begin{align}
\label{AdjointTransportedDiracOperator}
\int_M \big\langle \big(\alpha^{-1 } \circ \sD^{g_\alpha} \circ \alpha \big) \psi_1 , \, \psi_2 \big\rangle \; \vol_g \; = \; \int_M \big\{ &\big\langle \psi_1, \, \big(\alpha^{-1 } \circ \sD^{g_\alpha} \circ \alpha \big) \psi_2  \big\rangle \; \\ &+ \; \big\langle \psi_1 , \; \alpha \big[  \, \grad (\log \det \alpha) \, \big] \cdot \psi_2 \big\rangle  \big\} \; \vol_g  \nonumber \ .
\end{align}
This implies that the elliptic, first-order differential operator on $S_g$ given by
\beq
\psi \ \longmapsto \ \alpha^{-1 } \circ \sD^{g_\alpha} \circ \alpha (\psi) \;+ \;  \frac{1}{2} \, \alpha \big[  \, \grad (\log \det \alpha) \, \big] \cdot \psi
\eeq 
is formally self-adjoint with respect to the $L^2$ metric on spinors induced by $g$. Just as the Dirac operator $\sD^{g}$ is. The calculation that leads to \eqref{AdjointTransportedDiracOperator} is similar to the calculation in the usual proof that the Dirac operator is self-adjoint (e.g. \cite[lemma 2.28]{Bourguignon}).

\subsection{Transporting the Kosmann-Lichnerowicz derivative}

Consider now the Kosmann-Lichnerowicz derivative $\cL^{g_\alpha}_V$ on the spinor bundle $S_{g_\alpha}$, defined in \eqref{KLDerivative}. Again, one can transport it to $S_{g}$ and compare it with the native derivative $\cL^g_V$. A calculation using \eqref{ComparisonConnectionsTM} and \eqref{ComparisonConnectionsSM}, applied to definition \eqref{KLDerivative}, shows that
\beq
\label{TransportedKosmannDerivative}
\alpha^{-1 } \circ \cL_V^{{g_\alpha}}  \circ \alpha \, (\psi) \ = \ \cL_V^g \, \psi \; + \;  \frac{1}{4} \; \sum_{j\neq k} \; g \big( \alpha^{-1} (\Lie_V \alpha) (v_j),\, v_k \big) \; v_j \cdot v_k \cdot \psi
\eeq
Here $\Lie_V \alpha$ denotes the standard Lie derivative of the automorphism $\alpha$. In the special case where $\Lie_V \alpha$ vanishes, the operator $\alpha^{-1 } \circ \cL_V^{{g_\alpha}}  \circ \alpha$ coincides with $\cL^g_V$. The same thing happens when $\alpha^{-1} (\Lie_V \alpha)$ is self-adjoint with respect to $g$. In this case, the coefficients $g ( \alpha^{-1} (\Lie_V \alpha) (v_j),\, v_k )$ are symmetric in $j, k$, and hence the sum in \eqref{TransportedKosmannDerivative} vanishes.

Just as the original $\cL_V^g$, the transported Kosmann-Lichnerowicz derivative commutes with the chirality operator $\Gamma_K$ and is compatible with the inner product of spinors, 
\beq
\langle \, \alpha^{-1 } \, \cL_V^{{g_\alpha}}  \, \alpha \, (  \psi_1) \, , \, \psi_2 \, \rangle \ + \  \langle \, \psi_1  \, , \, \alpha^{-1 } \, \cL_V^{{g_\alpha}}  \, \alpha \, (  \psi_2)\, \rangle \ = \ \Lie_V  \langle \, \psi_1  \, , \, \psi_2 \, \rangle  \ .
\eeq
However, it acts on the Clifford multiplication of vector and spinor fields through
\beq
\alpha^{-1 } \, \cL_V^{{g_\alpha}}  \, \alpha \, ( \, U \cdot \psi \, ) \ = \ [V, U]\cdot \psi \; + \;  U \cdot \big[  \alpha^{-1 } \, \cL_V^{{g_\alpha}}  \, \alpha \, ( \psi) \big] \; + \;  \sum_k \, H^\alpha_V(U, v_k) \,  v_k \cdot \psi \ .
\eeq
Its commutator with the Dirac operator on $S_g$ is
\begin{align}
\label{CommutatorNewDerivativeDiracOperator}
\big[ \, \sD \, , \, \alpha^{-1 } \,\cL_V^{{g_\alpha}}  \, \alpha \,\big] \psi \; = \ &\sum_{j, k}\;  H^\alpha_V(v_j, v_k) \; v_j \cdot \nabla_{v_k} \psi \;  + \; \frac{1}{4}\, \sum_{i \neq j \neq k}   (\nabla_{v_i} H^\alpha_V) (v_j, v_k) \;   v_i \cdot  v_j  \cdot v_k \cdot \psi  \nonumber \linebr
&+ \; \frac{1}{2}\, \sum_{j, k}  \; \big[  \,  (\nabla_{v_k} H^\alpha_V) (v_j, v_k) \: - \:   (\nabla_{v_j} H^\alpha_V) (v_k, v_k) \, \big] \, v_j \cdot \psi     \ .
\end{align}
In the last two identities, $H^\alpha_V$ is the 2-tensor on $M$ defined by
\beq
\label{NewTensor}
H^\alpha_V (X,Y) \; := \; \frac{1}{2} \, \big[   \, (\Lie_V g)(X, Y) \: + \: g\big(\alpha^{-1} (\Lie_V \alpha) (X) , \, Y \big) \: - \:  g\big(\alpha^{-1} (\Lie_V \alpha) (Y) , \, X \big) \,  \big] 
\eeq
for all vector fields $X$ and $Y$. It has manifestly symmetric and anti-symmetric parts. This tensor vanishes if and only if $V$ is Killing and $\alpha^{-1} (\Lie_V \alpha)$ is self-adjoint with respect to $g$. These properties of the transported Kosmann-Lichnerowicz derivative can be proved using the analogous ones of $\cL_V^g$, as stated for example in \cite{Bap3}, and then calculating the corrections due to the second term on the right-hand side of \eqref{TransportedKosmannDerivative}.

\begin{proposition}
\label{InvarianceKLDerivativesKilling}
Take a $g$-symmetric automorphism $\alpha: TM \rightarrow TM$ that is positive-definite. As before, also call $\alpha$ the lifted spinor map $S_g \rightarrow S_{g_\alpha}$, where the metric $g_\alpha$ is as in \eqref{DefinitionNewMetric}. If a vector field $V$ is conformal Killing both for $g$ and $g_\alpha$, then 
\beq
\alpha^{-1 } \circ \cL_V^{g_\alpha} \circ \alpha  \ = \ \cL_V^g 
\eeq
as differential operators on spinors in $S_g$.
\end{proposition}
\begin{proof}
Using the definition of $g_\alpha$ and the properties of Lie derivatives, one calculates that
\beq
\label{IdentityLieDerivativeAlpha}
(\Lie_V g_\alpha) (X, Y) \ = \ (\Lie_V g)(\alpha^{-2}(X) , Y)  \; +\; g( (\Lie_V \alpha^{-2}) (X), Y )
\eeq
for all vector fields $X$ and $Y$ on $M$. The assumptions on $V$ mean that there are real functions $f$ and $f_\alpha$ on $M$ such that 
\beq
\Lie_V  g  \; = \; f \,g \qquad \qquad  \Lie_V  g_\alpha  \; = \; f_\alpha \,g_\alpha 
\eeq
as 2-tensors on $M$. Applying this to \eqref{IdentityLieDerivativeAlpha} and using the non-degeneracy of $g$, leads to 
\beq
\label{IdentityLieDerivativeConformal}
\Lie_V (\alpha^{-2}) \; = \; (f_\alpha - f) \, \alpha^{-2}
\eeq
as automorphisms of $TM$. Now, the fact that $\alpha$ is $g$-symmetric and positive-definite implies that $\alpha^{-1}$ is as well. So there is a local, $g$-orthonormal trivialization of $TM$ in which the automorphism $\alpha^{-1}$ is represented by a diagonal matrix $\diag (\lambda_1, \ldots , \lambda_m)$. Here $m = \dim M$ and the $\lambda_j$ are local, positive functions on $M$. Denote by $a_{ij}$ the entries of the matrix representing $\Lie_V( \alpha^{-1})$ in the same trivialization. Then \eqref{IdentityLieDerivativeConformal} implies that
\beq
(\lambda_i + \lambda_j)\, a_{ij} \; = \; (f_\alpha - f) \, \delta_{ij} \, \lambda_j^2
\eeq
for all $i$ and $j$. Since $\lambda_i + \lambda_j$ is always positive, it follows that 
\beq
2 \, a_{ij} \; = \; (f_\alpha - f) \, \delta_{ij} \,  \lambda_j \qquad  \Leftrightarrow  \qquad  2\, \Lie_V ( \alpha^{-1}) \; = \; (f_\alpha - f) \, \alpha^{-1} \ .
\eeq
So one also gets 
\beq
\alpha^{-1} (\Lie_V \alpha) \; =\; - \, (\Lie_V  \alpha^{-1}) \, \alpha \; =\; \frac{1}{2}\, ( f - f_\alpha) \, I \ .
\eeq
Finally, using this expression in \eqref{TransportedKosmannDerivative}, it is clear that the sum on the right-hand side of that equation will vanish. So we obtain the desired identity of differential operators.  
\end{proof}

\subsection{Transport to first order}

Let $(M,g)$ be an oriented, Riemannian manifold with a fixed topological spin structure. Let $\alpha_t$ be a smooth 1-parameter family of invertible automorphisms $TM$ that project to the identity on $M$, starting at $\alpha_0 = \mathrm{Id}$. Assume that all the $\alpha_t$ are $g$-symmetric, which implies that the  $\alpha_t^{-1}$ are as well. The derivative $\dot{\alpha} := \frac{\dd}{\dd t} \alpha \, |_{t=0}$ is then a $g$-symmetric endomorphism of $TM$. 
So are the covariant derivatives $\nabla_V \dot{\alpha}$ for all vector fields $V$ on $M$.

For small $t$, define the family of Riemannian metrics $g_{\alpha_t}$ on $M$ by \eqref{DefinitionNewMetric}. As before, the automorphisms $\alpha_t$ of $TM$ can be lifted to maps $\alpha_t: S_g \rightarrow S_{g_{\alpha_t}}$ between spinor bundles. These lifts can be used to transport the various differential operators between the two spinor bundles.
The first-order variations of the transported operators are then given by:
\begin{align}
\frac{\dd}{\dd t} \, (\, \alpha_t^{-1}\, \nabla_V^{g_{\alpha_t}} \, \alpha_t \,) (U) \ |_{t=0} \ &= \ -\, (\nabla_U \dot{\alpha}) (V) \; + \; \sum_j \, g\big( (\nabla_{v_j} \dot{\alpha}) (V) , \, U \big) \, v_j \\
\frac{\dd}{\dd t} \, (\, \alpha_t^{-1}\, \nabla_V^{g_{\alpha_t}} \, \alpha_t \,)\, \psi \ |_{t=0} \ &= \ -\, \frac{1}{4} \, \sum_{i \neq j} g\big( (\nabla_V \dot{\alpha}) (v_i) \: + \: 2\, (\nabla_{v_i} \dot{\alpha}) (V)\, , \, v_j \big) \, v_i \cdot v_j \cdot \psi   \\
\frac{\dd}{\dd t} \, (\, \alpha_t^{-1}\, \cL_V^{g_{\alpha_t}} \, \alpha_t \,)\, \psi \ |_{t=0} \ &= \ -\, \frac{1}{4} \, \sum_{i \neq j} \, (\Lie_V g)( \dot{\alpha} (v_i) \, , \, v_j ) \: v_i \cdot v_j \cdot \psi  \label{FirstVariationKLDerivative} \\
\frac{\dd}{\dd t} \, (\, \alpha_t^{-1}\, \sD^{g_{\alpha_t}} \, \alpha_t \,)\, \psi \ |_{t=0} \ &= \ \sum_j \, \dot{\alpha} (v_j) \cdot \nabla_{v_j} \psi \; -\; \frac{1}{2} \, \big[ \delta \dot{\alpha} \: + \: \grad (\mathrm{tr}^g \, \dot{\alpha}) \big] \cdot \psi  \label{FirstVariationDiracOperator} \ .
\end{align}
Formula \eqref{FirstVariationDiracOperator} for the Dirac operator was given long ago in \cite{Wang, BG}. Formula \eqref{FirstVariationKLDerivative} for the Kosmann-Lichnerowicz derivative appears to be new. It shows that the first variation vanishes when $V$ is conformal Killing on $M$.

\section{Proofs for section \ref{ReflectionsParity}}
\label{ProofsReflectionsParity}

\begin{proof}[\bf Proof of uniqueness in proposition \ref{thm:ReflectionsSpinors}]
Suppose that $\hat{R}$ and $\tilde{R}$ are two $\CC$-linear maps of spinors $\Gamma (S_{g_\PPP}) \rightarrow \Gamma (S_{\RR^\ast g_\PPP})$ with the four properties listed in proposition \ref{thm:ReflectionsSpinors}. Property \eqref{ReflectionsSquared} implies that $\hat{R}$ and $\tilde{R}$ are both invertible maps. The composition
\beq
\hat{R} \, \tilde{R}: \Gamma (S_{g_\PPP}) \ \rightarrow \ \Gamma (S_{\RR^\ast \RR^\ast g_\PPP}) = \Gamma (S_{g_\PPP})
\eeq
is then an automorphism of $\Gamma (S_{g_\PPP})$. From property \eqref{ReflectionsFunctions} of $\hat{R}$ and $\tilde{R}$, it follows that, for any complex function $f$ on $P$,
\beq
\hat{R} \, \tilde{R} (f \Psi) \ = \ (R^\ast R^\ast f ) \; (\hat{R} \, \tilde{R} \, \Psi)   \ = \  f  \;  (\hat{R} \, \tilde{R} \, \Psi) \ .
\eeq
So the map $\hat{R} \, \tilde{R}$ is $C^\infty (P)$-linear, and hence must be induced by a $\CC$-linear automorphisms of the bundle $S_{g_\PPP}$ that preserves its fibres. Moreover, it follows from property \eqref{ReflectionsCliffordMultiplications} of $\hat{R}$ and $\tilde{R}$ that, for any vector field $W$ on $P$, 
\beq
\hat{R} \, \tilde{R} (W \cdot \Psi) \ = \ (R_\ast R_\ast W ) \cdot (\hat{R} \, \tilde{R} \, \Psi)   \ = \  W \cdot (\hat{R} \, \tilde{R} \, \Psi)  \ .
\eeq
This means that the automorphism of the fibres of $S_{g_\PPP}$ determined by $\hat{R} \, \tilde{R}$ commutes with Clifford multiplication of vectors and spinors. Hence, it commutes with the action of the whole Clifford algebra on the fibre of $S_{g_\PPP}$. But these fibres are isomorphic to the complex spinor space $\Delta_{m+k}$, on which the Clifford algebra acts irreducibly, by definition. Thus, Schur's lemma implies that the composition $\hat{R} \, \tilde{R}$ is proportional to the identity on each fibre, and hence must be of the form
\beq
\label{CompositionSimplified}
\hat{R} \, \tilde{R} (\Psi) \ = \ \alpha \; \Psi
\eeq
for some complex function $\alpha$ on $P$. From property \eqref{ReflectionsCovariantDerivative} of  $\hat{R}$ and $\tilde{R}$, we also have that
\beq
\alpha \, \nabla_W \Psi \ = \  \hat{R} \, \tilde{R} ( \nabla_W \Psi) \ = \ \nabla_{\RR_\ast \RR_\ast W }  (\hat{R} \, \tilde{R} \, \Psi)   \ = \   \nabla_{W}  (\alpha \, \Psi) \ = \ \alpha \,  \nabla_{W}  \Psi  \; +\;  \dd\alpha (W) \; \Psi   \nonumber
\eeq
for all vector fields $W$ and spinors $\Psi$. So $\dd \alpha$ must vanish and the function $\alpha$ is constant on $P$. This constant cannot be zero because $\hat{R}$ and $\tilde{R}$ are invertible maps, so $\hat{R} \, \tilde{R}$ must be too. Finally, using property \eqref{ReflectionsSquared} of the maps  $\hat{R}$ and $\tilde{R}$ and denoting by $\hat{\eta}$ and $\tilde{\eta}$ the respective constants in $\mathrm{U} (1)$, it follows from \eqref{CompositionSimplified} that
\begin{align}
\alpha \; \hat{R}\,  (\Psi) \ &= \  \hat{R}( \alpha \; \Psi ) \ = \ \hat{R} \, \hat{R} \, \tilde{R} (\Psi) \ = \ \hat{\eta}\,  \tilde{R} (\Psi)  \label{RelationTwoReflections}   \linebr
\alpha \; \tilde{R}\,  (\Psi) \ &= \   \hat{R} \, \tilde{R}\, \tilde{R}\,  (\Psi)   \ = \ \hat{R} \, ( \tilde{\eta}\,  \Psi) \ = \   \tilde{\eta}\, \hat{R} \, (\Psi)  \ = \  \tilde{\eta}\, \alpha^{-1}\, \alpha \, \hat{R}\,  (\Psi) \ = \  \tilde{\eta}\, \alpha^{-1}\,  \hat{\eta}\,  \tilde{R} (\Psi)  \nonumber
\end{align}
for all spinors $\Psi$. So $\alpha^2 = \tilde{\eta}\, \hat{\eta}$ and $\alpha$ must be a complex phase, just as $\hat{\eta}$ and $\tilde{\eta}$ are. Thus, it follows from \eqref{RelationTwoReflections}, for example, that $\hat{R}$ and $\tilde{R}$ differ only by a constant complex phase.
\end{proof}

\vspace{.3cm}

\begin{proof}[\bf Construction of the spinor map in proposition \ref{thm:ReflectionsSpinors}]
In this part of the proof of proposition \ref{thm:ReflectionsSpinors}, we will construct the spinor reflection map from a local formula and show that it is globally well-defined.
The first task is to define oriented, local trivializations of the spinor bundles $S_{g_\PPP}$ and $S_{\RR^\ast g_\PPP}$ where the local formula \eqref{LocalSpinorReflections} can be applied. Start by picking an oriented, $g_P$-orthonormal trivialization of the tangent bundle $TP= \HH \oplus \VV $ over a domain $M\times {\mathcal U}_\alpha$ formed entirely of horizontal and vertical vector fields. So a trivialization of the form $\{ X_\mu^\HH, v_j^\alpha \}$, as described at the end of section \ref{Submersions}.
The $v_j^\alpha$ form an orthonormal basis of $TK$ over the domain ${\mathcal U}_\alpha \subset K$ with respect to the metric $g_K (x)$, for each $x \in M$. The $X_\mu$ form a $g_M$-orthonormal basis of $TM$ and their horizontal lift to $P$ is
\beq
\label{BasicLiftX2}
X_\mu^\HH =  X_\mu \; + \; A^a(X_\mu) \; e_a  \ .
\eeq
As usual, this local trivialization of $TP$ induces a trivialization of the spinor bundle $S_{g_\PPP}$ over the same domain $M\times {\mathcal U}_\alpha$. Now define a different, $R^\ast g_P$-orthonormal trivialization $\{ \tilde{X}_\mu^\HH, \tilde{v}_j^\alpha \}$ of $TP$ by putting 
\begin{align}
\tilde{v}_j^\alpha \ &:= \ R_\ast \, v_j^\alpha      \label{SecondTrivializationK}      \linebr  
 \tilde{X}_\mu^\HH \ &:= \  (-1)^{\delta_{\mu \nu}} \, R_ \ast (X_\mu^\HH)  \ = \ X_\mu \; + \; (R^\ast A)^a(X_\mu) \; e_a  \ .   \label{SecondTrivialization}
\end{align}
Here we have used the identity $R_\ast X_\mu = (-1)^{\delta_{\mu \nu}}\, X_\mu$ and the fact that $R_\ast \, e_a = e_a$, since $e_a$ is a vector field on $K$. 
The trivialization $\{ \tilde{X}_\mu^\HH, \tilde{v}_j^\alpha \}$ has the same orientation as $\{ X_\mu^\HH, v_j^\alpha \}$. It induces a trivialization of the second spinor bundle $S_{\RR^\ast g_\PPP}$ over the same domain $M\times {\mathcal U}_\alpha$.

To define the spinor reflection map, let $\Psi$ be any spinor on $S_{g_\PPP}$. With respect to the trivialization of this bundle induced by $\{ X_\mu^\HH, v_j^\alpha \}$, the spinor is represented by a function $\Psi_\alpha (x,y) : M \times  {\mathcal U}_\alpha \rightarrow \Delta_{m+k}$. Then we declare that the reflected spinor $R (\Psi)$ should be represented by the function
\beq
\label{LocalSpinorReflectionsAppendix}
(R \,\Psi)_{ \alpha} (x, y) \ := \ e^{i \xi}\: \Gamma_\nu \: \Psi_{ \alpha}(R(x), y) 
\eeq
with respect to the trivialization of $S_{\RR^\ast g_\PPP}$ induced by the second trivialization $\{ \tilde{X}_\mu^\HH, \tilde{v}_j^\alpha \}$ of $TP$. Here $e^{i \xi}$ is a constant complex phase, and $\Gamma_\nu$ is the gamma matrix on $\Delta_{m+k}$ corresponding to the coordinate $x^\nu$ of $M$ that changes sign in the reflection. To see that $R (\Psi)$ is globally well-defined, we have to check that the local definitions are consistent on intersections of trivialization domains.

Suppose that $M \times {\mathcal U}_\beta$ is another domain in $P$ with two oriented trivializations of $TP$, denoted $\{ X_\mu^\HH, v_j^\beta \}$ and $\{ \tilde{X}_\mu^\HH, \tilde{v}_j^\beta \}$, constructed in a way analogous to the previous ones. On  the intersection we have the transition functions $b_{\alpha \beta} : M \times ({\mathcal U}_\alpha \cap {\mathcal U}_\beta) \rightarrow {\mathrm{SO}} (k)$ defined by
\beq
v_j^\alpha \ = \   \sum\nolimits_i \, (b_{\alpha \beta})_j^i \;  v_i^\beta \ .
\eeq
From the definition \eqref{SecondTrivializationK} of $\tilde{v}_j^\alpha$, it is clear that the transition functions for the second family of trivializations are just the pullback functions
\beq
\tilde{b}_{\alpha \beta} \ = \ R^\ast b_{\alpha \beta} \ .
\eeq
Now, the fixed topological spin structure on $K$ and its extension to $M \times K$, determines a lift of the functions $b_{\alpha \beta}$ to functions
$B_{\alpha \beta} : M \times ({\mathcal U}_\alpha \cap {\mathcal U}_\beta) \rightarrow {\mathrm{Spin}} (k)$ satisfying the usual consistency conditions on triple intersections. The transition functions of the bundle $S_{g_\PPP} = S(\HH) \otimes S(\VV)$ are then
\beq
(I \otimes B_{\alpha \beta}):  M \times ({\mathcal U}_\alpha \cap {\mathcal U}_\beta) \rightarrow {\mathrm{Spin}} (m-1, 1) \times {\mathrm{Spin}} (k) \ .
\eeq
One can then check that the pullback functions $I \otimes R^\ast B_{\alpha \beta}$ also satisfy all the consistency conditions on triple intersections and form a set of transition functions for the bundle $S_{\RR^\ast g_\PPP}$ covering the functions $\tilde{b}_{\alpha \beta}$. Using these two sets of transition functions for the two spinor bundles, we finally calculate that, on the intersection of domains $M \times ({\mathcal U}_\alpha \cap {\mathcal U}_\beta)$, 
\begin{align}
\label{ConsistencyReflectionsTransitionFunctions}
(R\, \Psi)_\beta (x,y) \ &= \ e^{i \xi}\: \Gamma_\nu \: \Psi_{ \beta}(R(x), y)  \linebr 
&=  \ e^{i \xi}\: (\gamma_\nu \otimes I ) \:  [I \otimes B_{\beta \alpha}(R(x), y)  ] \, \Psi_{\alpha}(R(x), y)   \nonumber  \linebr 
&=  \ e^{i \xi}\:   [I \otimes (R^\ast B_{\beta \alpha} ) (x,y)] \:  (\gamma_\nu \otimes I )  \, \Psi_{\alpha}(R(x), y)   \nonumber    \linebr 
&=  \  [I \otimes (R^\ast B_{\beta \alpha} ) (x,y)] \; (R\, \Psi)_\alpha (x,y)      \nonumber   \ .
\end{align}
In the second equality we have used the construction of the higher-dimensional gamma matrices from \eqref{CliffordMultiplication}. 
Formula \eqref{ConsistencyReflectionsTransitionFunctions} shows that the local representatives of the reflected spinor $R\, \Psi$, as defined in \eqref{LocalSpinorReflectionsAppendix}, transform as they should in order to be a section of the bundle $S_{\RR^\ast g_\PPP}$. So  the spinor $R(\Psi)$ is globally well-defined on $M \times K$ and hence the map $R: \Gamma(S_{g_\PPP}) \rightarrow \Gamma(S_{\RR^\ast g_\PPP})$ is well-defined too.
\end{proof}

\vspace{.3cm}

\begin{proof}[\bf Verification of the four properties in proposition \ref{thm:ReflectionsSpinors}]
Using the local definition of spinor reflection in \eqref{LocalSpinorReflectionsAppendix} and the standard pullback of a function, we have
\beq
[R \, (f\,\Psi)]_{ \alpha} (x, y) \ = \  e^{i \xi}\: \Gamma_\nu \; f(R(x), y)\;  \Psi_{ \alpha}(R(x), y) \ = \ (R^\ast f)(x,y)\;  (R \,\Psi)_{ \alpha} (x, y) \nonumber \ .
\eeq
The properties of the gamma matrices and the chirality operator on $\Delta_{m+k}$ also imply that
\beq
[R\, R \, (\Psi)]_{ \alpha} (x, y) \ = \   e^{2 i \xi}\: (\Gamma_\nu)^2 \;  \Psi_{ \alpha}(R\, R(x), y) \ = \ \eta \, \Psi_{ \alpha} (x, y) \nonumber
\eeq
for the complex phase $\eta =  - \, e^{2 i \xi} \, (g_M)_{\nu \nu}$. These identities are valid on arbitrary domains $M \times {\mathcal U}_\alpha$, so we obtain that properties \eqref{ReflectionsFunctions} and \eqref{ReflectionsSquared} of proposition \ref{thm:ReflectionsSpinors} are valid globally.

Let $V$ be a vertical vector field on $P$, expressed locally as $V^j\, v^\alpha_j$. Then the local representative of the Clifford product $V \cdot \Psi$ is the function $V^j\, \Gamma_j \Psi_\alpha$ with values on $\Delta_{m+k}$. Since $\Gamma_j$ anti-commutes with $\Gamma_\nu$, it is clear from definition \eqref{LocalSpinorReflectionsAppendix} that property \eqref{ReflectionsCliffordMultiplications} in the proposition is valid for vertical vector fields. Now consider the horizontal, basic vector field $X_\mu^\HH$ of \eqref{BasicLiftX2}.
Using the local definition \eqref{LocalSpinorReflectionsAppendix} and the fact that $X_\mu^\HH$ is part of the $g_P$-orthonormal trivialization of $TP$ while $\tilde{X}_\mu^\HH$ is part of the $R^\ast g_P$-orthonormal trivialization, we calculate that the local representative of $R(X_\mu^\HH \cdot \Psi)$ satisfies
\begin{align} 
\label{CompatibilityCliffMultReflections2}
[ \, R(X_\mu^\HH \cdot \Psi) \, ]_\alpha \, |_{(x, y)} \ &=\   e^{i \xi}\: \Gamma_\nu \: (X_\mu^\HH \cdot \Psi)_\alpha  \, |_{(R(x), y)} \linebr 
&= \   e^{i \xi}\:  \Gamma_\nu \:  \Gamma_\mu\:  \Psi_\alpha  \, |_{(R(x), y)}  \nonumber  \linebr
&= \  - \, (-1)^{\delta_{\mu \nu}}\, \Gamma_\mu \:  e^{i \xi}\:  \Gamma_\nu\:  \Psi_\alpha  \, |_{(R(x), y)}  \nonumber  \linebr
&= \    -\,  (-1)^{\delta_{\mu \nu}}\, \Gamma_\mu \:  (R \, \Psi)_\alpha  \, |_{(x, y)} \nonumber   \linebr
&= \  - \, (-1)^{\delta_{\mu \nu}}\, \big[ \tilde{X}_\mu^\HH \cdot (R \, \Psi) \big]_\alpha  \, |_{(x, y)} \nonumber   \linebr
&= \ -\,  [\, ( R_\ast\, X_\mu^\HH ) \cdot (R \, \Psi) \, ]_\alpha  \, |_{(x, y)}   \ . \nonumber 
\end{align}
This last step uses the definition \eqref{SecondTrivialization} of $\tilde{X}_\mu^\HH$. The calculation shows that property \eqref{ReflectionsCliffordMultiplications} in the proposition is valid for horizontal vector fields too, and hence is always true.

At this point, the only piece missing to finish the proof of proposition \ref{thm:ReflectionsSpinors} is the verification of property \eqref{ReflectionsCovariantDerivative} of spinor reflections. We will do this using the local formula \eqref{StandardSpinorCovariantDerivative} for the spinorial lift of the Levi-Civita connection. Taking the first term in that formula and using the local definition \eqref{LocalSpinorReflectionsAppendix} of reflections, we get
\begin{align}
\label{ReflectionsCovariantDerivativeStep1}
R \,(\partial_W \Psi_\alpha) \, |_{(x, y)} \ &= \  e^{i \xi}\: \Gamma_\nu \: (\partial_W  \Psi_\alpha) \, |_{(R(x), y)}  \linebr
&= \  e^{i \xi}\: \Gamma_\nu \: (R^\ast \partial_W  \Psi_\alpha) \, |_{(x, y)}   \nonumber  \linebr
&= \  \partial_W [\, e^{i \xi}\: \Gamma_\nu \: (R^\ast  \Psi_\alpha) \, ] \,  |_{(x, y)}   \nonumber  \linebr
&=  \ \partial_W [(R \,\Psi)_\alpha \, ] \,  |_{(x, y)}    \nonumber \ .
\end{align}
As for the second term in \eqref{StandardSpinorCovariantDerivative}, the previously established properties \eqref{ReflectionsFunctions} and \eqref{ReflectionsCliffordMultiplications} of spinor reflections guarantee that 
\beq
\label{ReflectionsCovariantDerivativeStep2}
R \,[ \, g(\nabla^g_W u_i, \, u_j) \ u_i \cdot u_j \cdot \Psi  \, ] \ = \ R^\ast [ \, g(\nabla^g_W u_i, \, u_j) \, ] \: (R_\ast u_i) \cdot (R_\ast u_j) \cdot R(\Psi) \ .
\eeq
Here $\{ u_j \}$ denotes an oriented, $g_P$-orthonormal, local trivialization of $TP$. But the spacetime reflection $R: (P , g_P) \rightarrow (P, R^\ast g_P)$ is an isometry of pseudo-Riemannian manifolds. So the naturality properties of the Levi-Civita connection say that the pullback connection $R^\ast (\nabla^{R^\ast g})$ coincides with $\nabla^{R^\ast R^\ast g} =  \nabla^g$. See \cite[prop. 5.13]{Lee}. Therefore, the usual properties of pullback connections imply that, for all vector fields $V$ and $W$ on P,
\beq
R_\ast (\nabla^g_W V)  \ = \ \nabla_{R_\ast W}^{R^\ast g} (R_\ast V) \ .
\eeq
See \cite[ lem. 4.37]{Lee}, for example. Thus we get
\beq
R^\ast [\, g(\nabla^g_W V, \, U) \, ]  \ = \ (R^\ast g) (R_\ast \nabla^g_W V , \, R_\ast U ) \ = \   (R^\ast g) \big(\nabla_{R_\ast W}^{R^\ast g} (R_\ast V) , \, R_\ast U \big) \ .
\eeq
Applying this identity to \eqref{ReflectionsCovariantDerivativeStep2}, we get
\beq
\label{ReflectionsCovariantDerivativeStep3}
R \,[ \, g(\nabla^g_W u_i, \, u_j) \ u_i \cdot u_j \cdot \Psi  \, ] \ = \  \big[  (R^\ast g) \big(\nabla_{R_\ast W}^{R^\ast g} (R_\ast u_i) , \, R_\ast u_j \big) \big] \: (R_\ast u_i) \cdot (R_\ast u_j) \cdot R(\Psi) \ .
\eeq
Now, $\{ R_\ast u_j \}$ is a  $R^\ast g_P$-orthonormal, local trivialization of $TP$. However, it is not an oriented trivialization, as reflections invert the orientation. But the right-hand side of \eqref{ReflectionsCovariantDerivativeStep3} is invariant under a change of sign of any of the vectors $R_\ast u_j$, which would make this an oriented, $R^\ast g_P$-orthonormal trivialization. So the right-hand side of \eqref{ReflectionsCovariantDerivativeStep3} coincides with the second term of the covariant derivative $\nabla^{\RR^\ast g}_{R_\ast W} (R \Psi)$ as it would appear in an application of formula \eqref{StandardSpinorCovariantDerivative}. Combining this fact with \eqref{ReflectionsCovariantDerivativeStep1} and using again \eqref{StandardSpinorCovariantDerivative}, we finally conclude that $R (\nabla^{g}_{W} \Psi)\, = \, \nabla^{\RR^\ast g}_{R_\ast W} (R \Psi)$.
\end{proof}

\vspace{.3cm}

\begin{proof}[\bf Proof of proposition \ref{thm:PropertiesReflectionsSpinors}] Property \eqref{ReflectionsDiracOperator} in the proposition is a direct consequence of the application of identities \eqref{ReflectionsFunctions} and \eqref{ReflectionsCliffordMultiplications} to the definition of the Dirac operator. Property \eqref{ReflectionsChirality} can be obtained from identity \eqref{ReflectionsCliffordMultiplications} through the calculation
\begin{align}
R\, (\Gamma_P \Psi) \ &= \  i^{ \lfloor \frac{m+k -1}{2} \rfloor  } \, R\,(u_0 \cdots u_{m+k-1} \cdot \Psi)  \linebr
&= \  i^{ \lfloor \frac{m+k -1}{2} \rfloor  } \,  (-1)^{m+k} \, (R\,u_0) \cdots (R\, u_{m+k-1}) \cdot (R\, \Psi) \nonumber   \linebr
&= \   (-1)^{k+1} \,   i^{ \lfloor \frac{m+k -1}{2} \rfloor  } \, (- R\,u_0) \cdot (R\, u_1) \cdots (R\, u_{m+k-1}) \cdot (R\, \Psi) \nonumber   \linebr
&= \ (-1)^{k+1} \, \Gamma_P (R\, \Psi)  \nonumber \ .
\end{align}
Here we have used the definition \eqref{DefinitionChiralOperator} of the chirality operator and noted that, if $\{u_j \}$ is a local, oriented, orthonormal trivialization of $TP$ with respect to the metric $g_P$, then $\{- R_\ast u_0, \, R_\ast u_1, \ldots ,  R_\ast u_{m+k -1} \}$ is an oriented, orthonormal trivialization with respect to $R^\ast g_P$. We have also used that the dimension $m$ of $M$ is even by assumption. Finally, property \eqref{ReflectionsInnerProduct} follows directly from the local definition \eqref{LocalSpinorReflectionsAppendix} of spinor reflection, the definition \eqref{GeneralInnerProduct} of the inner product of spinors, the hermiticity properties of the gamma matrices, and the observation that, in our conventions \eqref{ConventionsGammaMatrices1}, we have 
\[
\Gamma_\nu\,  \Gamma_\nu \; =\;  -\,  g_P(X^\HH_\nu, X^\HH_\nu)  \; =\;  - \, (g_M)_{\nu\nu} \ .      \qedhere   
\]  
\end{proof}

\section{Proofs for section \ref{RepresentationSpaces}}
\label{ProofsRepresentationSpaces}

\begin{proof}[\bf Proof of proposition \ref{thm: DefinitionUnitaryRep}]
Using the definition of $\rho$ and property \eqref{NewDerivativeDerivation} of $\sL_V$, we have
\begin{multline}
\rho_U \rho_V (\psi) \ =  \ \sL_U \sL_V \psi  \; + \;  \frac{1}{2} \, \big[ \, (\divergence_g U) \, \sL_V \psi \,+\,  (\divergence_g V)\,  \sL_U \psi \, + \, [\ \Lie_U (\divergence_g V) \,] \, \psi \, \big]  \\ +\;  \frac{1}{4} (\divergence_g U) (\divergence_g V) \psi  \ .
\end{multline}
Cancelling the terms that are symmetric in $U$ and $V$, 
\beq
(\rho_U \rho_V  - \rho_V \rho_U )\psi  \ =  \ (\sL_U \sL_V  - \sL_V \sL_U )\psi \;  + \;  \frac{1}{2}  \big[\, \Lie_U (\divergence_g V)  -  \Lie_V (\divergence_g U)\, \big] \psi \ .
\eeq
From the definition of divergence, $\Lie_V \vol_g = (\divergence_g V) \vol_g$, we also have that
\begin{align}
(\divergence_g  [U,V])  \, \vol_g\  &= \ \Lie_{[U,V]} \, \vol_g = (\Lie_U \Lie_V  - \Lie_V \Lie_U ) \, \vol_g \linebr
&= \  \big[\Lie_U (\divergence_g V)  -  \Lie_V (\divergence_g U) \big]  \, \vol_g  \ .
\end{align}
Thus, using property \eqref{NewDerivativeClosure} of $\sL_V$, we get
\beq
(\rho_U \rho_V  - \rho_V \rho_U ) \, \psi  \ =  \ \sL_{[U, V]} \psi \;  + \;  \frac{1}{2}\,  (\divergence_g  [U,V])   \, \psi \; = \; \rho_{ [U,V]} \,\psi \ .
\eeq
This shows that indeed $\rho$ defines a Lie algebra representation. Finally, using property \eqref{NewDerivativeInnerProduct}, together with the definition of divergence and Stokes' theorem,  
\begin{align}
\int_K \big[  \langle  \rho_V \psi_1,  \psi_2 \rangle  +  \langle \psi_1,  \rho_V  \psi_2 \rangle \big] \vol _g \, &=  \: \int_K \big[   \langle \sL_V  \psi_1,  \psi_2 \rangle +   \langle  \psi_1,  \sL_V \psi_2 \rangle +  \langle \psi_1, \, \psi_2 \rangle  (\divergence_g V)  \big] \vol_g  \nonumber \linebr
& =  \: \int_K \big[ \Lie_V  \langle \psi_1,  \psi_2 \rangle \big] \vol_g +  \langle \psi_1, \, \psi_2 \rangle \, ( \Lie_V \vol_g )  \nonumber \linebr
& =  \: \int_K   \Lie_V \big[ \langle \psi_1,  \psi_2 \rangle \vol_g \big] \: = \: 0 \ .
\end{align}
So the representation is unitary with respect to the $L^2$ Hermitian product.
\end{proof}

\vspace{.3cm}

\begin{proof}[\bf Proof of proposition \ref{thm: DecompositionRepresentationSpaces1}]
Take the averaged metric $\hg$ defined in \eqref{DefinitionAverageMetric} and consider the Dirac operator $\sD^{\thg}$ on the spinor bundle $S_\thg$. We have the $\sD^{\thg}$-eigenspace decomposition
\beq
\label{DecompositionDiracEigenspaces2}
L^2 (S_\thg)  \  = \  \bigoplus_{m \in \mathbb{Z}}  \hat{E}_{m}   \ .
\eeq
By construction, all the vector fields in $\mathfrak{g}$ are Killing with respect to $\hg$. So the Kosmann-Lichnerowicz derivatives $\hcL_V$ commute with $\sD^{\thg}$ for all $V$ in $\mathfrak{g}$. Hence, they preserve all the finite-dimensional eigenspaces $\hat{E}_{m}$. By \eqref{Closure2}, the derivatives $\hcL_V$ define a representation of the compact Lie algebra $\mathfrak{g}$ on $\hat{E}_{m}$, so we have the decomposition
\beq
\label{DecompositionRepresentationSpaces3}
\hat{E}_{m}  \  = \  \bigoplus_\pi \, n_{m, \pi} \, \hat{E}_{m, \pi} \ ,
\eeq
where only a finite number of the integers $n_{m, \pi} \geq 0$ are non-zero as $\pi$ varies. Consider now the linear, invertible automorphism $\alpha: S_g \rightarrow S_\thg$. It induces an invertible linear map $L^2(S_g) \rightarrow L^2 (S_\thg)$. This map does not preserve the $L^2$-metrics, because \eqref{AlphaIsometry} implies that
\beq
\int_K   \langle\, \alpha(\psi_1), \, \alpha(\psi_2)\,  \rangle \, \vol_{\thg} \;= \;  \int_K   \langle \psi_1, \, \psi_2 \rangle \, (\det \alpha) \, \vol_{g} \ ,
\eeq
and in general $\det \alpha$ is not 1. However, it is clear that the corrected map 
\beq
\label{DefinitionCorrectedMap}
\tilde{\alpha} \ := \ (\det \alpha)^{-1/2} \, \alpha 
\eeq
does preserve the $L^2$ inner products. Using the corrected map to transport $\hat{\cL}_V$, a direct calculation shows that
\beq
\label{TransportDerivativesAux1}
\tilde{\alpha}^{-1} \circ \hat{\cL}_V \circ \tilde{\alpha} \ = \  \alpha^{-1} \circ \hat{\cL}_V \circ \alpha \; -\; \frac{1}{2}\, (\det \alpha)^{-1} \big[ \Lie_V (\det \alpha) \big] 
\eeq
But $V$ is Killing with respect to $\hg$, by construction, so
\beq
0 \: = \: \Lie_V \vol_{\thg}  \: = \: \Lie_V [ ( \det \alpha) \, \vol_{g} ]  \: = \:  [ \,   \Lie_V  ( \det \alpha)  +  ( \det \alpha) (\divergence_{g} V)\,  ] \, \vol_g \ .
\eeq
This gives us a formula for $\divergence_{g} V$ in terms of $\det \alpha$ that we can use in \eqref{TransportDerivativesAux1}. Using also definition \eqref{TransportedDerivativeDefinition}, we obtain
\beq
\label{TransportDerivativesAux2}
\tilde{\alpha}^{-1} \circ \hat{\cL}_V \circ \tilde{\alpha} \ = \  \sL_V \; +\; \frac{1}{2}\, \divergence_{g} V  \ =: \ \rho_V
\eeq
as operators on spinors in $S_g$. In other words, the $\mathfrak{g}$-representation on the $\sD^{\thg}$-eigenspaces $\hat{E}_m$ determined by the derivatives $\hat{\cL}_V$ is transported by the isometry of Hilbert spaces $\tilde{\alpha}^{-1}: L^2(S_\thg) \rightarrow L^2(S_{g})$ to the representation $\rho$ on the image spaces $ \tilde{\alpha}^{-1}(\hat{E}_m)$. Thus, if we simply define $W_{m, \pi} := \tilde{\alpha}^{-1}(\hat{E}_{m, \pi})$, the isometry $\tilde{\alpha}^{-1}$ will take the orthogonal decompositions \eqref{DecompositionDiracEigenspaces2} and \eqref{DecompositionRepresentationSpaces3} to the orthogonal decomposition
\beq
\label{DecompositionRepresentationSpaces5}
L^2 (S_g)  \  = \   \bigoplus_{m \in \mathbb{Z}}  \bigoplus_\pi \, n_{m, \pi} \, W_{m, \pi} \ ,
\eeq
with the restriction of $\rho$ to the subspace $W_{m, \pi}$ being equivalent to the irreducible $\mathfrak{g}$-representation $\pi$. Since the Lie group $G$ is compact connected, the $\mathfrak{g}$-representations on the subspaces $W_{m, \pi}$ can be exponentiated, leading to a $G$-representation on the space of spinors on $S_g$ which we also call $\rho$.

Finally, to conclude the proof of proposition \ref{thm: DecompositionRepresentationSpaces1}, it is enough to observe that when $G$ acts isometrically on $(K,g)$, then $\hg = g$, the automorphisms $\alpha = \tilde{\alpha} = I$ are trivial, and $W_{m, \pi} =  \hat{E}_{m, \pi} \subseteq  \hat{E}_m = E_m$.
\end{proof}

\vspace{.3cm}

\begin{proof}[\bf Proof of proposition \ref{thm: DecompositionRepresentationSpaces2}]
Suppose that $K$ is even-dimensional. Recall the definitions in the previous proof and consider again the $\mathfrak{g}$-representation on the finite-dimensional $\sD^{\thg}$-eigenspaces $\hat{E}_m$ determined by the derivatives $\hcL_V$.
The chirality operator $\Gamma_K$ satisfies $\Gamma_K \Gamma_K = 1$ and anti-commutes with the Dirac operator. So $\Gamma_K$ commutes with the square $(\sD^{\thg})^2$ and preserves its eigenspaces $\hat{E}_0$ and $\hat{E}_m \oplus \hat{E}_{-m}$, for all $m>0$. This allows us to decompose the $(\sD^{\thg})^2$-eigenspaces  as a sum of chiral subspaces 
\beq
\hat{E}_0 \; =\;  \hat{E}_0^+ \oplus \hat{E}_0^-   \qquad  \qquad  \quad  \hat{E}_m \oplus \hat{E}_{-m} \; =\; \hat{E}_m^+ \oplus \hat{E}_m^-
\eeq 
for $m>0$. Now, the chirality operator $\Gamma_K$ always commutes with Kosmann-Lichnerowicz derivatives. So the $\mathfrak{g}$-representation determined by the $\hcL_V$ preserves the chiral subspaces $\hat{E}_m^\pm$. Thus, we can decompose these subspaces as a sum of irreducible representations,
\beq
\label{DecompositionRepresentationSpaces6}
\hat{E}_m^\pm \ = \  \bigoplus_\pi \,  n_{m, \pi} \,  \hat{E}_{m, \pi}^\pm  \qquad {\rm for\ all} \ \  m \geq 0 \ .
\eeq
As discussed in section 5.2 of \cite{Bap3}, the integers $n_{m, \pi}$ are the same for both decompositions $\pm$. For $m >0$, this happens because the Dirac operator determines a $\mathfrak{g}$-equivariant isomorphism between $\hat{E}_m^+$ and $\hat{E}_m^-$. For $m=0$, the proof uses an index result of Atiyah and Hirzebruch \cite{Witten83, AH}. Thus, all in all, we get the orthogonal decomposition
\beq
\label{DecompositionRepresentationSpaces4}
L^2 (S_\thg)  \  = \   \bigoplus_{m \in \mathbb{N}_0}  \bigoplus_\pi \, n_{m, \pi} \, (\hat{E}^+_{m, \pi} \oplus \hat{E}^-_{m, \pi}) \ .
\eeq
This decomposition can be transported to the spinor bundle $S_g$ through the isometry $\tilde{\alpha}^{-1}: L^2(S_\thg) \rightarrow L^2(S_{g})$ defined in the proof of proposition \ref{thm: DecompositionRepresentationSpaces1}, more precisely in \eqref{DefinitionCorrectedMap}. Due to \eqref{AlphaCommutesChiral}, the maps $\tilde{\alpha}$ and $\tilde{\alpha}^{-1}$ commute with the chirality operators, and hence preserve the chiral subspaces. So if we define the subspaces of Weyl spinors $W_{m, \pi}^\pm := \tilde{\alpha}^{-1}(\hat{E}_{m, \pi}^\pm)$ inside $\Gamma (S_g)$, we get the orthogonal decomposition \eqref{DecompositionRepresentationSpaces2} stated in the proposition.

Finally, when $G$ acts isometrically on $(K,g)$, then $\hg = g$, the derivatives $\cL_V$ commute with $\sD^g$ for all $V$ in $\mathfrak{g}$, and the automorphisms $\alpha = \tilde{\alpha} = I$ are trivial. In this case, we have the identity $W_{m, \pi}^\pm =  \hat{E}_{m, \pi}^\pm =   E_{m, \pi}^\pm$. For $m>0$, the space $W_{m, \pi}^+ \oplus W_{m, \pi}^-$ is the same as the $(\sD^g)^2$-eigenspace  $E_{m, \pi}^+ \oplus E_{m, \pi}^-$, which, in turn, can be decomposed according to the $\sD^g$-eigenvalues as $E_{m, \pi} \oplus E_{-m, \pi}$.
\end{proof}

\vspace{.3cm}

\begin{proof}[\bf Proof of remark \ref{thm: RemarkComplexRepresentations1}]
The only assertions that remain to be proved concern the multiplicities of the representations. To this end, observe that, when the dimension of $K$ is not 1 (mod 4), there exist conjugation maps $j_+$ of spinors on $S_g$ and $S_{\hg}$. This follows from the construction of the set \eqref{DefinitionSetConjMaps} in proposition \ref{thm: ConjLinearAutomorphisms}. These maps commute with the Dirac operator and Kosmann-Lichnerowicz derivatives on the respective bundles, as follows from \eqref{CommutatorConjDiracOperator} and \eqref{CommutatorConjKLDerivative}. So $j_+$ preserves the $\sD^{\thg}$-eigenspaces $\hat{E}_{m}$ in decomposition \eqref{DecompositionRepresentationSpaces3} and is equivariant with respect to the $\mathfrak{g}$-action on them. This implies that the $\mathfrak{g}$-action on $\hat{E}_{m}$ is self-conjugate. In particular, the irreducible $\mathfrak{g}$-spaces of complex type inside $\hat{E}_{m}$ will always appear in conjugate pairs. Thus, we have $n_{m, \pi} = n_{m, \conj{\pi}}$ in decomposition \eqref{DecompositionRepresentationSpaces3}. Since the multiplicities $n_{m, \pi}$ are the same in decompositions \eqref{DecompositionRepresentationSpaces3} and \eqref{DecompositionRepresentationSpaces1}, by construction of the latter, we conclude that $n_{m, \pi} = n_{m, \conj{\pi}}$ also holds in \eqref{DecompositionRepresentationSpaces1}.

When $K$ has dimension 0, 6, 7 (mod 8), the conjugation $j_+$ defines a real structure on $\hat{E}_{m}$, as follows from \eqref{SquareConjLinearAutomorphisms}. So the quaternionic summands in decomposition \eqref{DecompositionRepresentationSpaces3} must appear with even multiplicity (see \cite[ch. II, sec. 6, exercise 10]{BD}). When $K$ has dimension 2, 3, 4 (mod 8), the conjugation $j_+$ defines a quaternionic structure on $\hat{E}_{m}$. So the real summands in \eqref{DecompositionRepresentationSpaces3} must appear with even multiplicity. Once again, these properties of the multiplicities in decomposition \eqref{DecompositionRepresentationSpaces3} are directly transferable to decomposition \eqref{DecompositionRepresentationSpaces1}, which finishes the proof. 
\end{proof}

\vspace{.3cm}

\begin{proof}[\bf Proof of remark \ref{thm: RemarkComplexRepresentations2}]
This remark only deals with even-dimensional $K$. 
From the proof of proposition \ref{thm: DecompositionRepresentationSpaces2}, we know that the $\mathfrak{g}$-action preserves the $\sD^{\thg}$-eigenspaces $\hat{E}_{m}$ and the chiral spaces $\hat{E}^\pm_{m}$ in \eqref{DecompositionRepresentationSpaces6}. From the proof of remark \ref{thm: RemarkComplexRepresentations1}, we know that $j_+$ is $\mathfrak{g}$-equivariant and preserves all the eigenspaces $\hat{E}_{m}$. 

Thus, when $K$ has dimension 0 (mod 4), it follows from the commutation relation  \eqref{CommutatorConjChiralityOperator} in Riemannian signature that the conjugation map $j_+$ also preserves the chiral spaces $\hat{E}^\pm_{m}$. This implies that the $\mathfrak{g}$-representation on $\hat{E}^\pm_{m}$ is self-conjugate. In particular, the irreducible $\mathfrak{g}$-spaces of complex type inside $\hat{E}^+_{m}$ will always appear in conjugate pairs. Thus, $n_{m, \pi} = n_{m, \conj{\pi}}$ in decompositions \eqref{DecompositionRepresentationSpaces6} and \eqref{DecompositionRepresentationSpaces4}. Moreover, when $K$ has dimension 0 (mod 8), the conjugation $j_+$ defines real structures $\hat{E}^\pm_{m}$, as follows from \eqref{SquareConjLinearAutomorphisms}. So the quaternionic summands in decomposition \eqref{DecompositionRepresentationSpaces6} must appear with even multiplicity (see \cite[ch. II, sec. 6, exercise 10]{BD}). When $K$ has dimension 4 (mod 8), the conjugation $j_+$ defines quaternionic structures on $\hat{E}^\pm_{m}$. So the real summands in \eqref{DecompositionRepresentationSpaces6} must appear with even multiplicity. Since the properties of the multiplicities in decomposition \eqref{DecompositionRepresentationSpaces4} are directly transferable to decomposition \eqref{DecompositionRepresentationSpaces2}, as before, this finishes the proof of all the assertions about dimension 0 (mod 4).

When $K$ has dimension 2 (mod 4), we know that $j_+$ preserves the eigenspace $\hat{E}_{0}$ and the sums $\hat{E}_{m} \oplus \hat{E}_{-m} = \hat{E}^+_{m} \oplus \hat{E}^-_{m}$ for $m>0$. Since $j_+$ is $\mathfrak{g}$-equivariant, it follows that the $\mathfrak{g}$-action on these spaces must be self-conjugate. So the irreducible $\mathfrak{g}$-spaces of complex type inside $\hat{E}_0$ and $\hat{E}^+_{m} \oplus \hat{E}^-_{m}$ will always appear in conjugate pairs. This is possible only if $n_{m, \pi} = n_{m, \conj{\pi}}$ in decomposition \eqref{DecompositionRepresentationSpaces4}. Again, this equality of multiplicities is directly transferable to decomposition \eqref{DecompositionRepresentationSpaces2}, as desired.
\end{proof}

\vspace{.3cm}

\begin{proof}[\bf Proof of lemma \ref{SpecialBasesInternalSpinors}]
Suppose that $K$ has odd dimension $k$. Then only one of the conjugation maps exists, as follows from \eqref{DefinitionSetConjMaps}. When $k = 3$ (mod 4), it is $j_+$. This map commutes with the $G$-action on spinors and preserves the representation spaces $W_m =  \bigoplus_\pi \, n_{m, \pi} \, W_{m, \pi}$. Since the $W_{m, \pi}$ are irreducible, the $G$-invariant conjugation will preserve some of them and will map isomorphically pairs of others. The preserved ones must correspond to real or quaternionic representations $\pi$. In the real case, $W_{m, \pi}$ has a complex basis consisting of $j_+$-invariant vectors. In the quaternionic case, $W_{m, \pi}$ has a complex basis consisting of pairs of vectors that are swapped by $j_+$, up to a sign. In the case of spaces $W_{m, \pi}$ that are not preserved by $j_+$, but are instead mapped isomorphically in pairs, it is enough to choose a complex basis of the first space in the pair, and then apply $j_+$ to define a basis of the second space in the pair. Uniting the bases in all these cases, we get a basis of $W_m$ with the required properties.
	
When $K$ has dimension 1 (mod 4), only $j_-$ exists. This conjugation commutes with the $G$-action and maps the representation spaces $W_m$ and $W_{-m}$ isomorphically. For $m >0$, take a complex basis $\{\psi_{m}^\alpha \}$ of $W_m$ formed by spinors in the irreducible subspaces $W_{m, \pi}$. Then the set $\{\psi_{m}^\alpha ,  \; j_- (\psi_{m}^\alpha) \}$ forms a basis of $W_m \oplus W_{-m}$ with the required properties. For $m=0$, we know that $j_-$ is $G$-equivariant and preserves $W_0$. So an argument similar to the one used in the case of dimension 3 (mod 4) guarantees the existence of the required basis also for $W_0$.
	
Now suppose that $K$ has even dimension. The two conjugations exist simultaneously and are $G$-equivariant. When $K$ has dimension 2 (mod 4), the two conjugations map $W_m^+$ isomorphically to $W_m^-$. Pick a basis $\{(\psi^+_{m})^\alpha \}$ of $W_m^+$ formed by spinors in the irreducible subspaces $W_{m, \pi}^+$. Then the set $\{(\psi_{m}^+)^\alpha ,  \; j_+ (\psi_{m}^+)^\alpha \}$ forms a basis of $W_m^+ \oplus W_{m}^-$ whose elements simply get swapped by $j_+$, up to a phase. Since these basis elements are chiral and $j_- = \zeta \,\Gamma_K \, j_+$, according to \eqref{RelationTwoConjugations}, we recognize that $j_-$ also acts on them by simple swaps of pairs and multiplication by phases, as desired.
	
Finally, when $k = 0$ (mod 4), both commute with the chirality operator $\Gamma_K$ and preserve the spaces $W_m^\pm =  \bigoplus_\pi \, n_{m, \pi} \, W_{m, \pi}^\pm$.  So an argument similar to the one used in the case of dimension 1 (mod 4) guarantees the existence of bases of $W_m^+$ and $W_m^-$ whose elements are in the respective irreducible spaces and are either invariant or get swapped in pairs by $j_+$. Since these bases are made of chiral elements, it follows again from the identity $j_- = \zeta \,\Gamma_K \, j_+$ that also $j_-$ acts on them as desired.
\end{proof}

\section{An example with $(\SU (3) \times \SU (2) \times \mathrm{U}(1) ) / \mathbb{Z}_6$ symmetries}
\label{Example}

The approach to Kaluza-Klein models followed in this paper is quite general. We have allowed the internal space $K$ to be an arbitrary compact spin manifold, only sometimes restricting it to even dimensions. While the properties established in this manner become robust and transverse, it is also interesting to delve into explicit models, with particular choices of internal geometry. This is indispensable from a physical viewpoint, as one tries to find the specific geometries that, after dimensional reduction, most closely resemble the properties of the Standard Model. Studying explicit examples, however, is no easy task in spin geometry. There are very few examples of compact manifolds in which the Dirac operator is explicitly understood, with its spectrum of eigenvalues, eigenspaces, and spinor representations of the isometry group \cite{Ginoux}. The purpose of this appendix is to point out a particular geometry that may be interesting to further investigate. The motivation comes from a speculative physical scenario outlined in \cite{Bap1}.

Let $K$ be the 8-dimensional group manifold $\SU (3)$ equipped with its bi-invariant metric, which is unique up to a scale factor. This metric is Einstein with positive curvature. The orientation-preserving isometry group is $(\SU(3) \times \SU(3))/\mathbb{Z}_3$. As described in section \ref{RepresentationSpaces}, the action of this group on $K$ can always be lifted to an action of $\SU(3) \times \SU(3)$ on the spinors over $K$. 

The bi-invariant metric on $\SU(3)$ is an unstable saddle point of the Einstein-Hilbert action under TT-deformations of the metric, not a local maximum. Most TT-deformations of the bi-invariant metric reduce its scalar curvature, but a particular deformation found in \cite{Jensen} takes the scalar curvature up to positive infinity. Since $- R_{g_\KKK}$ plays the role of a potential for the internal metric in KK models, the initial bi-invariant metric should unravel along that unstable direction. When this happens, the isometry group of the metric is broken down to $ G_{\mathrm{SM}}  = (\SU(3) \times \SU(2) \times \mathrm{U}(1))/\mathbb{Z}_6$, as pointed out in \cite{Bap1}. The broken symmetries cause 4 gauge bosons to gain a mass at the Planck scale.

The unraveling of the bi-invariant metric will continue indefinitely, leading to an internal space of infinite curvature, unless the Einstein-Hilbert Lagrangian is not the full story and there are higher-order corrections to the action on $M\times K$. If that is the case, the correction terms can become relevant as the curvature of $K$ grows, and can contribute to stabilize the unraveling at a scale distinct from the Planck one. The stabilized metric would define the KK internal vacuum. Due to the correction terms, its isometry group can be further broken from $G_{\mathrm{SM}}$ down to a smaller subgroup. This second symmetry breaking occurs at the mass scale of the correction terms (presumably the electroweak scale), not at the Planck scale. So it could produce light gauge bosons, besides the already present 4 heavy ones.

This story is a motivation to study the Dirac operator on $\SU (3)$ and the gauge representations \eqref{DefinitionUnitaryRep} induced on its spinors. First, using the simpler bi-invariant metric. That would be already enough to get an idea of the spinor gauge representations, since decompositions \eqref{DecompositionRepresentationSpaces1} and  \eqref{DecompositionRepresentationSpaces2} into irreducibles preserve the structure of decompositions \eqref{DecompositionRepresentationSpaces3} and \eqref{DecompositionRepresentationSpaces4} for the average metric. (Here, the more symmetric average metric is the bi-invariant one.) The next step would be to understand the eigenspaces and eigenvalues of the Dirac operator on the deformed metric with $(\SU(3) \times \SU(2) \times \mathrm{U}(1))/\mathbb{Z}_6$ isometry group. This metric is written down explicitly in \cite{Bap1}. The metric deformation can lift degeneracies and make some of the initial Dirac eigenvalues (i.e. masses of 4D fermions) much larger, and others smaller. With a $\SU(3) \times \SU(3)$ spinor representation on a manifold with the smaller $G_{\mathrm{SM}}$ isometry group, the relation between the representation spaces and the $\sD^K$-eigenspaces becomes non-trivial. Chiral interactions between the massive gauge fields and the 4D fermions should now appear, as described in \cite{Bap3}.

The third step would be to study the Dirac operator on $\SU(3)$ equipped with a perturbed left-invariant metric that breaks the isometry group from $G_{\mathrm{SM}}$ to $\SU(3) \times \mathrm{U}(1)$. This is the perturbation that would presumably come from the higher-order corrections to the Einstein-Hilbert Lagrangian, as described in the speculative scenario above. Explicit metrics on $\SU(3)$ with this smaller isometry group and a Higgs-like parameter with values in $\mathbb{C}^2$ are written down in \cite{Bap0}. See also the broader discussion in \cite{Coquereaux}. With those metrics, 3 more gauge bosons gain a non-zero mass, now at the perturbation scale.

It would also be interesting to investigate in detail the branching of the degenerate eigenvalues of $\sD$ during the two stages of symmetry breaking of the metric on $K= \SU(3)$. As described at the end of section \ref{RepresentationSpaces}, these branchings can lead to the emergence of different generations of 4D fermions in a KK model.

\newpage

%%%%%%%%%%%%%%%%%%%%%%%%%%%%%%%%%%%%%%%%%%%%%%%%%%%%%%%%%%%%%%

\renewcommand{\baselinestretch}{1.2}\normalsize

\phantomsection       % Creates the hyperref anchor in the table of contents

\addcontentsline{toc}{section}{References}

\renewcommand\doitext{doi:}

\vspace{1cm}

\end{document}